\documentclass[11pt,a4paper]{article}

\usepackage{hyperref}
\usepackage{color}
\usepackage{amsthm}
\usepackage{mathtools}

\usepackage[normalem]{ulem} 
\usepackage{graphicx}
\usepackage{amssymb,latexsym,cite}
\usepackage{amsmath}
\usepackage{amsfonts}
\usepackage{mathrsfs}
\usepackage{bbm}
\usepackage{bm}
\usepackage[T1]{fontenc}

\usepackage[matrix,arrow,color]{xy}

\def\slasha#1{\setbox0=\hbox{$#1$}#1\hskip-\wd0\hbox to\wd0{\hss\sl/\/\hss}}

\def\periodb#1{\setbox0=\hbox{$#1$}#1\hskip-\wd0\hbox to\wd0{-}}

\newcommand{\id}{\mathrm{id}}   			

\newcommand{\CA}{\mathcal{A}}    			

\newcommand{\CF}{\mathcal{F}}
\newcommand{\CCF}{\mathscr{F}}

\newcommand{\CH}{\mathcal{H}}

\newcommand{\CJ}{\mathcal{J}}

\newcommand{\CCM}{\mathscr{M}}

\newcommand{\CP}{\mathcal{P}}

\newcommand{\CR}{\mathcal{R}}

\newcommand{\CV}{\mathcal{V}}
\newcommand{\CCV}{\mathscr{V}}

\newcommand{\CX}{\mathcal{X}}

\newcommand{\frg}{\mathfrak{g}}				

\newcommand{\frv}{\mathfrak{v}}

\def\swone{{\textrm{\tiny (1)}}}
\def\swtwo{{\textrm{\tiny (2)}}}

\def\lact{{\textrm{\tiny L}}}
\def\ract{{\textrm{\tiny R}}}

\newcommand{\FR}{\mathbbm{R}}     			
\newcommand{\FC}{\mathbbm{C}}     			
\newcommand{\RZ}{\mathbbm{Z}}     			

\newcommand{\dd}{\mathrm{d}}     			

\newcommand{\LL}{\mathrm{L}}     			

\newcommand{\aso}{\mathfrak{so}}

\newcommand{\comment}[1]{}     				
\def\tyng(#1){\hbox{\tiny$\yng(#1)$}}			
\def\tyoung(#1){\hbox{\tiny$\young(#1)$}}			

\newcommand{\beq}{\begin{eqnarray}}
\newcommand{\eeq}{\end{eqnarray}}

\newcommand{\sfd}{{\sf d}}

\newcommand{\sff}{{\sf f}}

\newcommand{\sfR}{\mathsf{R}}

\definecolor{outrageousorange}{rgb}{1.0, 0.43, 0.29}

\newcommand{\Tr}{\mathrm{Tr}}

\theoremstyle{plain}

\newtheorem{proposition}[equation]{Proposition}

\theoremstyle{definition}
\newtheorem{definition}[equation]{Definition}

\newtheorem{example}[equation]{Example}

\newcommand{\dwedge}{\curlywedge}
\newcommand{\om}{\omega}
\newcommand{\epsi}{\epsilon}
\newcommand{\nn}{\nonumber}

\newcommand{\midwedge}{\text{\Large$\wedge$}}

\newcommand{\midoplus}{\text{\Large$\otimes$}}

\newcommand{\dsf}{{\mathsf{d}}}

\def\RR{{\mathcal R}}

\def\beq{\begin{equation}}
\def\bee{\begin{equation}}
\def\eeq{\end{equation}}
\def\bea{\begin{eqnarray}}
\def\eea{\end{eqnarray}}
\def\ba{\begin{align}}
\def\ea{\end{align}}

\numberwithin{equation}{section}
\catcode`@=12

\begin{document}

\renewcommand{\thefootnote}{\fnsymbol{footnote}}

\begin{titlepage}
	
	\renewcommand{\thefootnote}{\fnsymbol{footnote}}
	
	\begin{flushright}
		\small
		{\sf EMPG--21--04}
	\end{flushright}
	
	\begin{center}
		
		\vspace{1cm}
		
		\baselineskip=24pt
		
		{\Large\bf Braided $\boldsymbol{L_{\infty}}$-Algebras, Braided Field Theory \\
                  and Noncommutative Gravity}
		
		\baselineskip=14pt
		
		\vspace{1cm}
		
		{\bf Marija Dimitrijevi\'c \'Ciri\'c}${}^{\,(a)\,,\,}$\footnote{Email: \ {\tt
				dmarija@ipb.ac.rs}} \ \ \ \ \ {\bf Grigorios Giotopoulos}${}^{\,(b)\,,\,}$\footnote{Email: \ {\tt
				gg42@hw.ac.uk}} \\[2mm] {\bf Voja Radovanovi\'c}${}^{\,(a)\,,\,}$\footnote{Email: \ {\tt
				rvoja@ipb.ac.rs}} \ \ \ \ \ {\bf Richard
			J. Szabo}${}^{\,(b)\,,\,}$\footnote{Email: \ {\tt R.J.Szabo@hw.ac.uk}}
		\\[6mm]
		
		\noindent  ${}^{(a)}$ {\it Faculty of Physics, University of
			Belgrade}\\ {\it Studentski trg 12, 11000 Beograd, Serbia}
		\\[3mm]
		
		\noindent  ${}^{(b)}$ {\it Department of Mathematics, Heriot-Watt University\\ Colin Maclaurin Building,
			Riccarton, Edinburgh EH14 4AS, U.K.}\\ and {\it Maxwell Institute for
			Mathematical Sciences, Edinburgh, U.K.} \\ and {\it Higgs Centre
			for Theoretical Physics, Edinburgh, U.K.}
		\\[30mm]
		
	\end{center}
	
	\begin{abstract}
		\noindent
We define a new homotopy algebraic structure, that we call a braided $L_\infty$-algebra, and use it to systematically construct a new class of noncommutative field theories, that we call braided field theories. Braided field theories have gauge symmetries which realize a braided Lie algebra, whose Noether identities are inhomogeneous extensions of the classical identities, and which do not act on the solutions of the field equations. We use Drinfel'd twist deformation quantization techniques to generate new noncommutative deformations of classical field theories with braided gauge symmetries, which we compare to the more conventional theories with star-gauge symmetries. We apply our formalism to introduce a braided version of general relativity without matter fields in the Einstein--Cartan--Palatini formalism. In the limit of vanishing deformation parameter, the braided theory of noncommutative gravity reduces to classical gravity without any extensions.
	\end{abstract}
	
\end{titlepage}

{
	\tableofcontents
}

\setcounter{footnote}{0}
\renewcommand{\thefootnote}{\arabic{footnote}}

\newpage

\section{Introduction}

In this paper we will present a new perspective on symmetries of noncommutative field theories, and develop first principles which enable systematic constructions of new examples of field theories with noncommutative gauge symmetries. To achieve this we develop a new notion of homotopy algebra that is suited to organise the symmetries and dynamical content of (generalized) noncommutative gauge theories; we call this new mathematical object a `braided $L_\infty$-algebra'. Although our formalism applies broadly to any perturbative classical field theory with symmetries, in this paper we focus exclusively on examples of diffeomorphism-invariant field theories for simplicity, in order to illustrate clearly and conceptually how our novel constructions work. In particular, we derive new noncommutative extensions of the Einstein--Cartan--Palatini formulation of gravity in three and four spacetime dimensions. 

In this section we begin with a discussion of some background and context, before turning to a more detailed description of our results and the contents of this paper.

\subsubsection*{Classical field theories and $\boldsymbol{L_\infty}$-algebras}

A \emph{classical} field theory on a manifold can for a large part be understood as a dynamical system, that is, a system of partial differential equations which encodes the dynamics of fields. The field theory may have symmetries, particularly local symmetries, and we refer to these broadly as `generalized gauge symmetries' to accomodate more general systems whose fields and symmetries are not necessarily associated with connections on a principal bundle. In the setting of dynamical systems, by `symmetry' we mean that the field equations transform covariantly under the symmetry transformations; the gauge symmetry then acts on the space of classical solutions. The main object of interest is therefore the space of classical physical states, defined to be the moduli space of solutions to the field equations modulo the local symmetries. The field equations may or not follow from a variational principle. For a Lagrangian field theory, the data may be accompanied by a gauge-invariant action functional, however for many (but not all) considerations the action functional is supplementary data which is not needed to analyse the space of classical physical states.

This perspective is naturally and systematically organised by strong homotopy Lie algebras~\cite{LadaStasheff92}, or $L_\infty$-algebras, which are generalizations of differential graded Lie algebras with infinitely-many graded antisymmetric brackets satisfying higher homotopy versions of the Jacobi identity. The data of any perturbative\footnote{In the case of a non-polynomial theory or a non-linear space of fields $\mathcal{M}$, this holds at the level of the pertubative expansion of the theory on the tangent space $T_{\phi}\mathcal{M}$ at a fixed solution $\phi \in \mathcal{M}$.} classical field theory gives rise to an $L_\infty$-algebra. This may be seen in several ways. From a practical standpoint, the $L_\infty$-algebra underlying a field theory may be constructed directly through an explicit ``bootstrapping'' algorithm~\cite{Hohm:2017pnh}. They can also be obtained through the duality between $L_\infty$-algebras and differential graded commutative algebras, with the latter in this case being in the guise of the BV--BRST formalism~\cite{Stasheff:1997iz,BVChristian}. These constructions can also be understood more geometrically in the setting of derived geometry, where the $L_\infty$-algebras which appear are the (derived) tangent complexes on the space of fields describing the deformation theory around critical points of an action functional~\cite{Costello}. For a Lagrangian field theory, the extra data of an action functional is captured by imposing an additional cyclic structure on the underlying $L_\infty$-algebra.

\subsubsection*{Noncommutative gauge theories and gravity}

The dynamical systems perspective becomes particularly fruitful in situations where one is interested in a noncommutative deformation of a given field theory,\footnote{For us, this will be a field theory where the underlying spacetime manifold $M$ is replaced by a noncommutative deformation of the algebra of functions on $M$, along with the relevant differential geometrical objects as modules over this algebra. Being a deformation of a commutative theory, the noncommutative structures and dynamics must reduce to the classical spacetime and theory in the commutative limit.} and in particular of gravity. For example, in the metric formulation of general relativity, it is often problematic to write down a variational principle which deforms the standard Einstein--Hilbert action functional, due to ambiguities in defining noncommutative deformations of volume forms and scalar curvatures, as well as reality properties and the inherent non-polynomial nature of the action (see e.g.~\cite{SchenkelThesis} for a discussion of this). Yet consistent deformations of the field equations are straightforward (see e.g.~\cite{TwistApproach,NAGravity}). It is then natural to seek $L_\infty$-algebras which organise symmetries and dynamics of noncommutative field theories. 

The noncommutative extension of the bootstrap approach has been employed by~\cite{Blumenhagen:2018kwq,Kupriyanov:2019cug} as a means for constructing new examples of noncommutative gauge theories. In this approach the $L_\infty$-algebra framework encodes not only the symmetries and dynamics of the field theory, but also the general failure of the deformed exterior algebra from forming a differential graded algebra and, in the case of nonassociative gauge symmetries, the failure of the Jacobi identities. This generally produces $L_\infty$-algebras with an infinite number of non-zero brackets, which has several limitations: explicit expressions for all brackets are difficult to obtain, the bootstrapping of brackets is particularly difficult and involved for nonassociative gauge symmetries, and so far this approach has only met with some success in the semi-classical limit of a noncommutative gauge theory. The approach based on symplectic embeddings~\cite{Kupriyanov:2021cws} drastically improves the situation in that it allows for explicit and systematic constructions of all brackets of the $L_\infty$-algebra, but so far this approach is also restricted to the semi-classical limit and moreover limited to purely kinematic considerations.

In this paper we discuss a new alternative approach, which arose from our attempts to understand a noncommutative version of the Einstein--Cartan--Palatini theory of gravity in this language; the $L_\infty$-algebra formalism for the classical theory was developed in detail by~\cite{ECPLinfty}. In the standard noncommutative extension of this theory~\cite{AschCast}, the (extended) local Lorentz symmetry is implemented by the usual star-gauge transformations, and in principle one may construct correspondingly an $L_\infty$-algebra structure as in the bootstrapping approach to noncommutative gauge theories. However, diffeomorphism invariance in this theory is implemented by the usual `twisted' actions~\cite{TwistApproach} (see~\cite{Szabo:2006wx} for a review), and the dichotomy between the actions of the two types of symmetries leads to severe difficulties in the implementation of an underlying $L_\infty$-algebra as an organising principle for noncommutative gravity: The twisted diffeomorphism symmetry does not appear to fit naturally into the standard $L_\infty$-algebra picture. The situation is reminescent of what happens in double field theory~\cite[Section~5.1]{Hohm:2017pnh}, which requires extending the diffeomorphism symmetries by ``higher'' symmetries, implementing reducible gauge symmetries of the theory. In the case of noncommutative gravity this quickly becomes very complicated, and it is not clear how to manage the higher brackets of the $L_\infty$-algebra. In other words, the standard $L_\infty$-algebra formalism does not seem to be an appropriate or useful tool for organising the dynamics of noncommutative gravity. Our proposal then is to  deform the $L_\infty$-algebra structure itself in such a way that it becomes compatible with the twisted diffeomorphisms. This leads to a new type of homotopy algebra and a new type of gauge symmetry for noncommutative field theories in general.

\subsubsection*{Braided $\boldsymbol{L_\infty}$-algebras and braided field theories}

Following the classical treatment, we employ the machinery of Hopf algebras and Drinfel’d twist deformation quantization\footnote{We use the phrase ``deformation quantization'' colloquially to mean formal noncommutative deformation. Indeed, our theories will be classical in the physical sense, and no quantum physics
		will take place in this manuscript.}~\cite{MajidBook} to produce \emph{braided} noncommutative deformations of field theories.\footnote{In this paper we use the term `braided' as a synonym for `twisted', as the former is more natural in the context in which our field theories are defined.} By twisting the enveloping Hopf algebra of vector fields on a manifold $M$ to a non-cocommutative Hopf algebra and simultaneously deforming its category of modules, we deform the $L_\infty$-algebras into \emph{braided $L_\infty$-algebras}, a structure that was sketched in~\cite{Ciric:2020eab}. In essence, one invokes the definition of an $L_\infty$-algebra in terms of morphisms in a  symmetric monoidal category with a non-trivial braiding. In practice, the axioms of a classical $L_\infty$-algebra are supplemented with a non-trivial isomorphism $\RR:V\otimes V'\to V\otimes V'$ whenever a transposition isomorphism between two modules $V$ and $V'$ occurs. For example, antisymmetry is now braided: $\ell_2(v\otimes v')=-\ell_2\circ \RR(v'\otimes v)$, and similarly for all higher homotopy Jacobi relations. Braided $L_\infty$-algebras may be regarded as a higher homotopy coherent extension of the notion of a braided Lie algebra, which was introduced by Woronowicz under the name `quantum Lie algebra'~\cite{Woronowicz:1989rt}, and later formulated in a more general categorical setting by Majid~\cite{Majid:1993yp}.

When applied to the $L_\infty$-algebra underlying a classical dynamical system, this procedure deforms the collection of kinematical and dynamical data simultaneously and consistently, producing a genuinely new type of noncommutative field theory, that we call a \emph{braided field theory}. These are field theories with braided gauge symmetries which close as braided Lie algebras. The spaces of fields, or more precisely their tangent spaces, now form braided representations of these braided Lie algebras. Various aspects of the kinematical part of braided gauge symmetries have appeared before in e.g.~\cite{BraidedGauge,SchenkelThesis}, but as far as we are aware a complete and general treatment of noncommutative dynamics with a variational principle has not been addressed before. From our perspective, the braided homotopy relations guarantee the braided representation properties, braided covariance of the field equations, and braided gauge invariance of the action functional. These symmetries differ from the star-gauge symmetries that are usually employed in noncommutative gauge theories, which close as classical Lie algebras. As a consequence, our braided deformations do not require the introduction of any new degrees of freedom or larger enveloping algebras to achieve gauge closure. We compare our braided gauge symmetries with the more conventional noncommutative gauge symmeties based on star-gauge transformations in Section~\ref{sec:braidedvsstargauge}.

We point out and resolve a conceptual issue concerning braided gauge symmetries: they leave the field equations covariant and action functional invariant, but they do not act on the space of solutions of the field equations, contrary to the classical case. This is a consequence of the braided Leibniz rule of braided Lie algebra representation theory. However, the braided $L_\infty$-algebra framework allows for the derivation of a braided version of the Noether identities, i.e. a set of differential identities among the field equations and the dynamical fields, which imply that not all degrees of freedom are independent. That is, some of them are gauge redundant, justifying the terminology `braided gauge symmetry'. The braided Noether identities are non-trivial deformations of their classical counterparts. 

\subsubsection*{Braided noncommutative gravity}
 
In this paper we apply our formalism of braided $L_\infty$-algebras to produce new braided noncommutative theories of gravity in three and four spacetime dimensions. We consider only the case of parallelizable spacetime manifolds $M$, for which the twist deformations are straightforward to implement and exhibit the salient features of our constructions more clearly; see~\cite{ECPLinfty} for the general treatment in the classical case. In three dimensions, the noncommutative field equations and action functional are naive deformations of the classical ones (with braided commutators), however in four dimensions we uncover non-trivial deformations. Both theories are crucially different from previous approaches for several reasons. 

Firstly, in our braided approach the two sets of infinitesimal gauge symmetries, diffeomorphisms and local coframe rotations, each deform consistently with one another and together combine as a semi-direct product of braided Lie algebras, contrary to earlier approaches where diffeomorphisms are deformed as a braided Lie algebra but rotations as a classical Lie algebra (see e.g.~\cite{AschCast}). Thus no new fields or extended enveloping algebras of symmetries need be introduced, contrary to these other approaches. Secondly, our approach formalises the treatment of braided gauge transformations and extends earlier treatments of twisted diffeomorphism symmetry of gravity to \emph{all} gauge transformations, pointing out and clarifying several novel aspects of the space of solutions to the field equations compared to the unbraided case. Finally, one can in principle study solutions of the braided field equations, and compare our first order noncommutative theory of gravity in four dimensions with the existing metric formulation for noncommutative gravity (see e.g.~\cite{SchenkelThesis}). We expect that these two theories are drastically different, but our approach enjoys an explicit action functional and the coupling to matter fields (such as fermions) is simply a straightforward calculation within the braided $L_\infty$-algebra framework. We stress that the classical solution space of earlier treatments of the metric formulation of noncommutative gravity with twisted diffeomorphism symmetry also suffers from the same interpretational problem mentioned above, while our framework provides an explanation of the meaning of these symmetries in this sense. We compare our braided theories of gravity with other noncommutative theories of gravity in Section~\ref{sec:NCGR}.

\subsubsection*{Summary and outline of the paper}

We begin in the next two sections by giving some background to the main developments of this paper. In Section~\ref{sec:LinftyFT} we review the basics of $L_\infty$-algebras, together with their cyclic versions, and how they are used to organise the symmetries and dynamics of classical field theories. We illustrate the formalism on the two main classes of examples of this paper: Chern--Simons gauge theory and the Einstein--Cartan--Palatini theory of gravity. In Section~\ref{sec:Drinfeldtwist} we review the main ideas of Drinfel'd twist deformation quantization, using the enveloping Hopf algebra of vector fields on a manifold $M$, and look at the two particular examples of Hopf modules which form the backbone of this paper: the exterior algebra of differential forms on $M$ and the algebras of vector fields on $M$ themselves; these two modules also furnish our first examples of braided Lie algebras. At times we appeal to the categorical picture which interprets twist deformation quantization as a functorial equivalence of symmetric monoidal categories. This language is not strictly necessary to understand the main ideas or computations of this paper, but it does provide a convenient conceptual framework and at times simplifies what would otherwise be unwieldy technical computations.

In Section~\ref{sec:braidedLinfty}, we come to the main mathematical definitions and constructions of this paper: we develop in detail the notion of a braided $L_\infty$-algebra. It is here that the categorical language is particularly powerful, as it enables a straightforward formulation of the braided homotopy relations, which would be otherwise cumbersome and difficult to write out in general. Our definition makes sense in any symmetric monoidal category, and parallels the notion of a braided Lie algebra as a Lie algebra in such a category. When a classical $L_\infty$-algebra sits in the category of modules over the enveloping Hopf algebra of vector fields on $M$, we describe in detail how its Drinfel'd twist deformation defines a braided $L_\infty$-algebra (see Proposition~\ref{prop:braidedfromclassical}). This result is very useful in applications, as once we are given a classical $L_\infty$-algebra, we know immediately that its twist quantization is a braided $L_\infty$-algebra, without the need of having to go through lengthy and technically involved checks of the braided homotopy relations. We also consider twist quantization of cyclic $L_\infty$-algebras via Drinfel'd twists which are compatible with the cyclic structure in a suitable sense, and show that they produce strictly cyclic braided $L_\infty$-algebras (see Proposition~\ref{prop:compatiblestrict}), a result which is important for applications to noncommutative field theories.

In Section~\ref{sec:bft} we turn to the main physical applications of the present paper. We formulate the notion of a braided field theory, taking the perspective that it defines a formal dynamical system. We give a general account of braided field theories starting from a braided $L_\infty$-algebra as input. After a detailed general discussion of braided gauge symmetries, we proceed to explain how to build the kinematics and dynamics of a noncommutative field theory in the braided $L_\infty$-algebra formalism, paralleling the constructions of the classical case. In general, the proofs of various properties such as gauge closure, braided covariance of the field equations, braided Noether identities and braided gauge invariance of an action functional are much more complicated compared to the classical case, as the tradeoff of strict graded antisymmetry to its braided version no longer allows simplification of the necessary homotopy relations. In particular, we find that braided covariance is realised non-trivially and results in the peculiar features of the space of classical physical states discussed above. The Noether identities also assume a very different form and contain inhomogeneous field-dependent terms compared to their classical versions.

Our construction gives a straightforward algorithm for producing braided field theories as noncommutative deformations of classical field theories:
\begin{enumerate}
\item Given a classical field theory on a manifold $M$, work out its corresponding $L_\infty$-algebra, either by bootstrapping the initial brackets of its gauge transformations, field equations and Noether identities, or by using its dual BV--BRST formulation.
\item Given an enveloping Hopf algebra of vector fields on $M$ which generate symmetries of the classical field theory, and over which the $L_\infty$-algebra is a module, choose a corresponding Drinfel'd twist element and quantize the $L_\infty$-algebra to a braided $L_\infty$-algebra.
\item Run this braided $L_\infty$-algebra through the formulas of Section~\ref{sec:bft} to produce a noncommutative field theory with braided gauge symmetries.
\end{enumerate}
We apply this algorithm to the example of Chern--Simons gauge theory in Section~\ref{sec:braidedCStheory}. While the resulting braided field theory formally looks like a simple deformation of its classical version, it is quite different from the usual discussions of noncommutative Chern--Simons theory: in particular, it possesses braided gauge symmetries and its Noether identity is a deformation of the usual Bianchi identity. In general, the field equations and action functional of the resulting braided field theory will not simply be the ``naive'' deformations of those of the classical field theory, to which they nevertheless reduce in the classical limit: our algorithm constructs genuinely new examples of noncommutative gauge theories. We illustrate this in Section~\ref{sec:NCgravity} by applying our formalism to construct new braided theories of noncommutative gravity in three dimensions, whose field equations are given by \eqref{eq:Fom3d} and \eqref{eq:Fe3d}, and in four dimensions, whose field equations are given by \eqref{eq:Fom4d} and~\eqref{eq:Fe4d}.

\subsection*{Outlook}

This paper arose as part of an ongoing programme to formulate a nonassociative theory of gravity that is inspired by scenarios in closed non-geometric string theory, see e.g.~\cite{NAGravity,Barnes:2016cjm,Blumenhagen:2016vpb}. We argued in~\cite{ECPLinfty} that the $L_\infty$-algebra formalism should provide the natural receptacle to capture the failure of closure and covariance of field equations under nonassociative gauge transformations. To accomodate these instances, one may in principle deform the classical $L_\infty$-algebra using cochain twists of the enveloping algebra of vector fields, rather than cocycle twists, which results in a quasi-Hopf algebra. The end result will be an $L_\infty$-algebra in a symmetric monoidal category with non-trivial braiding and associator isomorphisms. Using the corresponding deformed homotopy identites, nonassociative gauge transformations ought to have a straightforward definition, and the corresponding field equations and action functionals may be in reach within this formalism. In this paper we focus on associative theories of noncommutative gravity, leaving these more complicated extensions for future investigations.

Although the present paper focuses on examples of diffeomorphism-invariant field theories, it is straightforward to apply our formalism to other field theories with less symmetry by restricting to a Hopf subalgebra of vector fields over which the classical $L_\infty$-algebra is a module. Our approach also offers a straightforward route to ``noncommutative higher gauge theory'' through $k$-term braided $L_\infty$-algebras with $k>4$. Along these lines, it is natural to wonder if our braided noncommutative field theories arise as suitable low energy limits of string theory, similarly to the conventional models based on star-gauge symmetry (see e.g.~\cite{Szabo:2001kg} for a review). A promising direction is to look for connections to braided symmetries of string amplitudes, for instance through the Hopf algebraic symmetries of worldsheet string correlation functions discussed in~\cite{Asakawa:2008nu,Asakawa:2008cc}. Such a connection further has the potential to connect noncommutative gravity with string theory for the first time. All of these interesting extensions of the present work are left for future investigations.

In this paper we treat only the case of classical field theories. Quantization of our braided field theories should be made possible by firstly finding a braided version of the Chevalley--Eilenberg algebra of an $L_\infty$-algebra, which would then produce a braided version of the BV--BRST formalism for field theories. This should then be compared to Oeckl's `braided quantum field theory' framework~\cite{Oeckl:1999zu,Sasai:2007me}, which however has so far been restricted to field theories without gauge symmetries. Some steps in this direction have been discussed recently in~\cite{Nguyen:2021rsa}. The formalism of~\cite{Nguyen:2021rsa} also suggests a natural way to handle the problem of defining the braided moduli space of classical solutions, which in the dual setting of the braided BV formalism may be completely characterised by its braided differential graded algebra of functionals. These possibilities and their physical relevance
will be explored in future investigations.

\subsubsection*{Acknowledgments}

We thank Branislav Jur\v{c}o, Lukas M\"uller and Alexander Schenkel for helpful discussions and correspondence.
The work of {\sc M.D.C.} and {\sc V.R.} is supported by Project
ON171031 of the Serbian Ministry of Education, Science and
Technological Development. The work of M.D.C. and R.J.S. was partially supported by the Croatian Science Foundation Project IP--2019--04--4168. The work of {\sc G.G.} is supported by a
Doctoral Training Grant from the UK Engineering and Physical Sciences
Research Council. The work of {\sc R.J.S.} was supported by
the Consolidated Grant ST/P000363/1 
from the UK Science and Technology Facilities Council.

\paragraph{Conflict of interest statement.} 

On behalf of all authors, the corresponding author states that there is no conflict of interest.

\section{$L_{\infty}$-algebras and classical field theory}
\label{sec:LinftyFT}

\subsection{$L_{\infty}$-algebras and cyclic $L_\infty$-algebras}
\label{sec:Linfty}

An $L_\infty$-algebra is a $\RZ$-graded real vector space $V=\bigoplus_{k\in \RZ}\, V_{k}$ equipped with graded antisymmetric multilinear maps
\begin{align*}
\ell_n: \text{\Large$\otimes$}^{n} V \longrightarrow V \ , \quad  v_1\otimes \dots\otimes v_n \longmapsto \ell_n (v_1,\dots,v_n)
\end{align*}
for each $n\geq1$, which we call $n$-brackets. The graded antisymmetry translates to 
\begin{equation}\label{eq:gradedantisym}
\ell_n (\dots, v,v',\dots) = -(-1)^{|v|\,|v'|}\, \ell_n (\dots, v',v,\dots) \ ,
\end{equation}
where we denote the degree of a homogeneous element $v\in V$ by $|v|$.
The $n$-bracket is a map of degree $|\ell_n|=2-n$, that is
\begin{equation*}
\big|\ell_n(v_{1}, \dots ,v_{n})\big| = 2-n +\sum_{j=1}^n \, |v_j| \ .
\end{equation*}

The $n$-brackets $\ell_n$ are required to fulfill infinitely many
identities ${\cal J}_n(v_1,\dots,v_n)=0$ for each $n\geq1$, called homotopy relations, with
\begin{align} \label{eq:calJndef}
{\cal J}_n(v_1,\dots,v_n) := \sum^n_{i=1}\, (-1)^{i\,(n-i)} \
\sum_{\sigma\in{\rm Sh}_{i,n-i}} \, & \chi(\sigma;v_1,\dots,v_n) \\ &
\qquad
\times
\ell_{n+1-i}\big(
\ell_i(v_{\sigma(1)},\dots,
v_{\sigma(i)}),
v_{\sigma(i+1)},\dots
,v_{\sigma(n)}\big)
\
, \nn
\end{align}
where, for each $i=1,\dots,n$, the second sum runs over $(i,n-i)$-shuffled permutations
$\sigma\in S_n$ of degree $n$ which are restricted as 
\begin{equation*}
\sigma(1)<\dots<\sigma(i) \qquad \mbox{and} \qquad \sigma(i+1)<\dots < \sigma(n) \ .
\end{equation*}
The Koszul sign $\chi(\sigma;v_1,\dots,v_n)=\pm\, 1$ is determined from the grading by
\begin{align*}
v_{\sigma(1)}\wedge\cdots\wedge v_{\sigma(n)} = \chi(\sigma;v_1,\dots,v_n) \ v_1\wedge\cdots\wedge v_n \ .
\end{align*}

Let us look explicitly at the first three identities. For $n=1$ the
identity ${\cal J}_1(v)=0$ is
\begin{align*}
\ell_1\big(\ell_1(v)\big) = 0 \ , 
\end{align*}
which states that the map $\ell_1:V\to V$ is a
differential making $(V,\ell_1)$ into a cochain complex. For $n=2$ the
identity ${\cal J}_2(v_1,v_2)=0$ reads
\begin{align*}
\ell_1\big(\ell_2(v_1,v_2)\big) = \ell_2\big(\ell_1(v_1),v_2\big) +
  (-1)^{|v_1|}\, \ell_2\big(v_1, \ell_1(v_2)\big) \ , 
\end{align*}
which states that $\ell_1$ is a graded
derivation with respect to the $2$-bracket $\ell_2:V\otimes V\to V$. For $n=3$ the
identity ${\cal J}_3(v_1,v_2,v_3)=0$ expands into
\begin{align}
& \ell_2\big(\ell_2(v_1,v_2),v_3\big) - (-1)^{|v_2|\,|v_3|}\,
  \ell_2\big(\ell_2(v_1,v_3),v_2\big) +  (-1)^{(|v_2|+|v_3|)\,|v_1|}\,
  \ell_2\big(\ell_2(v_2,v_3),v_1\big) \nn\\[4pt]
& \hspace{2cm} = -\ell_3\big(\ell_1(v_1),v_2,v_3\big) - (-1)^{|v_1|}\,
  \ell_3\big(v_1, \ell_1(v_2), v_3\big) - (-1)^{|v_1|+|v_2|}\,
  \ell_3\big(v_1,v_2, \ell_1(v_3)\big) \nn\\
& \hspace{4cm} -\ell_1\big(\ell_3(v_1,v_2,v_3)\big) \ ,
\label{I3}\end{align}
which for $\ell_3=0$ is just the graded Jacobi identity for the
$2$-bracket $\ell_2$; in general it expresses the violation of the Jacobi identity by
a cochain homotopy determined by $\ell_1$. In this sense $L_\infty$-algebras are (strong) homotopy
deformations of differential graded Lie algebras which are the special
cases where the ternary and all higher brackets vanish: $\ell_n=0$ for
all~$n\geq3$. In general, the homotopy relations for $n\geq3$ are
generalized Jacobi identities.

We are sometimes interested in the case where an
$L_\infty$-algebra $(V,\{\ell_n\})$ is further endowed with a graded symmetric non-degenerate bilinear pairing $\langle-,-\rangle:V\otimes V\to\FR$ which is cyclic in the sense that
\begin{align*}
\langle v_0,\ell_n(v_1,v_2,\dots,v_n)\rangle = (-1)^{n+(|v_0|+|v_n|)\,n+|v_n| \,\sum_{i=0}^{n-1}\,|v_i|} \ \langle v_n,\ell_n(v_0,v_1,\dots,v_{n-1})\rangle
\end{align*}
for all $n\geq1$. This is the natural notion of an invariant inner product on an
$L_\infty$-algebra, and if such a pairing
exists the resulting algebraic structure is called a cyclic
$L_\infty$-algebra; this generalizes the notion of a quadratic Lie algebra. If in addition the pairing is odd, say of degree
$-2p-1$, then the only non-vanishing pairings are 
$$
\langle-,-\rangle:V_k\otimes V_{2p-k+1}\longrightarrow \FR \qquad
\mbox{for} \quad k\leq p \ ,
$$ 
so that the pairing is strictly symmetric and the cyclicity condition simplifies to
\begin{align*}
\langle v_0,\ell_n(v_1,v_2,\dots,v_n)\rangle = (-1)^{(|v_0|+1)\,n} \ \langle v_n,\ell_n(v_0,v_1,\dots,v_{n-1})\rangle \ .
\end{align*}

\subsection{Classical field theories in the $L_{\infty}$-algebra formalism}
\label{sec:Linftygft}

In this paper, by a classical field theory we will mean a polynomial (or perturbative)
field theory with irreducible gauge symmetries. They can be completely
encoded in $4$-term $L_{\infty}$-algebras concentrated in degrees
$0$, $1$, $2$, and $3$:
$$
V=V_0\oplus V_1\oplus V_2\oplus V_3 \ .
$$
The kinematical data is encoded in the vector space $V_0$ of
gauge parameters and in the vector space $V_1$ of fields, while the
dynamical data is encoded in $V_2$ and $V_3$. These vector spaces are
typically infinite-dimensional spaces of sections of vector bundles
over a manifold $M$. 

The full symmetries and dynamics of field theories are expanded in terms of the
brackets as follows. 
Given $\lambda \in V_{0}$ and $A\in V_{1}$, the gauge
variations are encoded as the maps $A\mapsto A+\delta_\lambda A$ where
\begin{align} \label{gaugetransfA}
\delta_{\lambda}A=\sum_{n =0}^\infty \, \frac{1}{n!}\, (-1)^{\frac12\,{n\,(n-1)}}\, \ell_{n+1}(\lambda,A,\dots,A) \ ,
\end{align}
where the brackets involve $n$ insertions of the field
$A$. 
The field equations $F_A=0$ are encoded through
\begin{align}\label{EOM} 
F_A=\sum_{n =1}^\infty \, \frac{1}{n!}\, (-1)^{\frac12\,{n\,(n-1)}}\, \ell_{n}(A,\dots,A) \ ,
\end{align}
with the covariant gauge variations
\begin{align} \label{gaugetransfF}
\delta_{\lambda} F_A=\sum_{n =0}^\infty \, \frac{1}{n!}\, (-1)^{\frac12\,{n\,(n-1)}}\, \ell_{n+2}(\lambda,F_A,A,\dots,A) \ .
\end{align}
If the distribution in $TV_1$ spanned by the gauge parameters is
involutive off-shell, that is, when $F_A\neq0$, then the homotopy relations imply that the closure relation for the gauge algebra has the form
\begin{align}\label{eq:closure}
[\delta_{\lambda_1},\delta_{\lambda_2}]_\circ A
= \delta_{[\![\lambda_1,\lambda_2]\!]_A}A \ ,
\end{align}
where $[\delta_{\lambda_1},\delta_{\lambda_2}]_\circ :=\delta_{\lambda_1}\circ\delta_{\lambda_2} - \delta_{\lambda_2}\circ\delta_{\lambda_1}$, and
\begin{align*}
[\![\lambda_1,\lambda_2]\!]_A = -\sum_{n=0}^\infty \, \frac1{n!}\,
(-1)^{\frac12\,{n\,(n-1)}}\, \ell_{n+2}(\lambda_1,\lambda_2,A,\dots,A)
\ .
\end{align*}
The homotopy relations guarantee that the Jacobi
identity is generally satisfied for any triple of maps $\delta_{\lambda_1}$,
$\delta_{\lambda_2}$ and~$\delta_{\lambda_3}$.
The Noether identities are encoded by
\begin{align}\label{eq:Noether}
\dsf_AF_A := \sum_{n=0}^\infty \, \frac1{n!} \, (-1)^{\frac12\, n\,
  (n-1)} \, \ell_{n+1}(F_A,A,\dots,A) = 0 \ ,
\end{align}
which vanishes identically as a consequence of the homotopy relations
${\cal J}_n(A,\dots,A)=0$, for all $n\geq1$, of the
$L_\infty$-algebra~\cite{BVChristian}. 

For a field theory with a Lagrangian formulation, the action functional is encoded
via a symmetric non-degenerate bilinear pairing $\langle -,-\rangle :
V \otimes V\to\FR$ of degree~$-3$, which makes $(V,\langle-,-\rangle)$ into a cyclic
$L_\infty$-algebra. Then it is easy to see that the field equations
$F_A = 0$ follow from varying the action functional defined as 
\begin{align} \label{action}
S(A) := \sum_{n=1}^\infty \, \frac{1}{(n+1)!}\, (-1)^{\frac12\,{n\,(n-1)}}\, \langle A, \ell_{n}(A,\dots,A)\rangle \ ,
\end{align}
since then cyclicity implies $\delta S(A)=\langle F_A,\delta A\rangle$.
Cyclicity also implies
\begin{align}\label{eq:gtNoether}
\delta_\lambda S(A)=\langle F_A,\delta_\lambda A\rangle=-\langle\dsf_AF_A,\lambda\rangle \ ,
\end{align}
so that gauge invariance of the action functional $\delta_\lambda
S(A)=0$ is then equivalent to the Noether identities
$\dsf_AF_A=0$. 
For further details and the connections
to the BV--BRST formalism see \cite{BVChristian}, while details of the
aspects covered in this paper are found in \cite{ECPLinfty}.

\subsection{Example I: Chern--Simons gauge theory}
\label{sec:CStheory}

Let us illustrate how to reconstruct Chern--Simons theory in three
spacetime dimensions in the $L_\infty$-algebra framework, which is the prototypical example of a gauge theory with a simple
bracket structure, see e.g.~\cite{ECPLinfty,BVChristian}. Let $\frg$
be a real Lie algebra with Lie bracket $[-,-]_\frg$ and let $M$ be a
three-manifold. Then the 4-term $L_\infty$-algebra which organises Chern--Simons
gauge theory on $M$ is simply the differential graded Lie algebra whose underlying cochain complex is the
de~Rham complex $(\Omega^\bullet(M,\frg),\dd)$ in three dimensions with
values in the Lie algebra $\frg$, and whose underlying graded Lie algebra $\big(\Omega^\bullet(M,\frg),[-,-]_\frg\big)$ is given by extending the Lie bracket of $\frg$ to
  $\Omega^\bullet(M,\frg):=\Omega^\bullet(M)\otimes\frg$ by the tensor
  product with the exterior multiplication of forms; that is, 
\begin{align*}
  \ell_1(\alpha)=\dd\alpha \qquad \mbox{and} \qquad \ell_2(\alpha,\beta)=-[\alpha,\beta]_\frg
\end{align*}
 for $\alpha,\beta\in V=\Omega^\bullet(M,\frg)$.

 It is easy to check that these brackets reproduce the kinematical
 and dynamical content of classical Chern--Simons gauge theory via the prescription of
 Section~\ref{sec:Linftygft}. The transformation of a gauge field
 $A\in\Omega^1(M,\frg)$ by a gauge parameter
 $\lambda\in\Omega^0(M,\frg)$ is given by
 \begin{align*}
 \delta_\lambda A&=\ell_1(\lambda)+\ell_2(\lambda,A) =
                   \dd\lambda+[A,\lambda]_\frg \ .
 \end{align*}
 As for any
 differential graded Lie algebra, these transformations close a (field-independent)
 gauge algebra
 $[\![\lambda_1,\lambda_2]\!]_A=-\ell_2(\lambda_1,\lambda_2)
 =[\lambda_1,\lambda_2]_\frg$, which in this case is just the Lie
 algebra $\Omega^0(M,\frg)$. The field equations $F_A=0$ are encoded by
 \begin{align*}
 F_A&=\ell_1(A)-\tfrac12\,\ell_2(A,A) = \dd A+\tfrac12\,[A,A]_\frg \ ,
 \end{align*}
 with
 \begin{align*}
 \delta_\lambda F_A &= \ell_2(\lambda,F_A) =
                      [F_A,\lambda]_\frg
 \end{align*}
 expressing the covariance of the
 curvature 2-forms $F_A\in\Omega^2(M,\frg)$ under local gauge
 transformations. Thus the field equations $F_A=0$ are solved by flat
 connections of a trivial principal bundle on the three-manifold
 $M$, which are locally gauge equivalent to the trivial connection $A=0$. Finally
 \begin{align*}
 \dsf_A F_A&=\ell_1(F_A)-\ell_2(A,F_A) = \dd F_A+[A,F_A]_\frg = \dd^AF_A
 \ ,
 \end{align*}
 where $\dd^A$ denotes the usual covariant derivative. Thus the Noether
 identities $\dsf_AF_A=0$ are equivalent to the Bianchi identities
 $\dd^AF_A=0$ in $\Omega^3(M,\frg)$. In this way we reproduce the usual
 space of classical physical states of Chern--Simons theory as the moduli
 space of flat connections on~$M$.

 To reproduce the Chern--Simons action functional in this
 formulation, we need further structure: we assume that $M$ is a closed
 oriented three-manifold and that $\frg$ is a quadratic Lie algebra. We
 may then construct a cyclic $L_\infty$-algebra structure on $V$ by
 defining the pairing of $\frg$-valued forms in complementary degrees: 
 \begin{align} \label{eq:CSpairing}
 \langle\alpha,\beta\rangle := \int_M\,
 \Tr_\frg(\alpha\wedge\beta) \ ,
 \end{align}
where $|\alpha|+|\beta|=3$. This defines a cyclic non-degenerate pairing of
 degree~$-3$, where cyclicity follows from the invariance of the
 quadratic form $\Tr_\frg:\frg\otimes\frg\to\FR$ on $\frg$.
 Then the Chern--Simons action
 functional is reproduced according to the general
 prescription of Section~\ref{sec:Linftygft}:
 \begin{align}\label{eq:CSaction}
 S(A) = \frac12\,\langle A,\ell_1(A)\rangle -\frac1{6}\,\langle
 A,\ell_2(A,A)\rangle =
   \frac12\, \int_M\, \Tr_\frg\Big(A\wedge\dd A + 
 \frac13\, A\wedge [A, A]_\frg\Big) \ .
 \end{align}

 \subsection{Example II: Einstein--Cartan--Palatini gravity}
 \label{sec:ECP3d}

We shall now review the reconstruction of the first-order
formulation of gravity in the $L_\infty$-algebra framework, from the
perspective of the Einstein--Cartan--Palatini (ECP) formalism for
general relativity, 
following~\cite{ECPLinfty}. To set up notation and to illustrate the ideas, here we only
consider the $L_\infty$-algebras in three spacetime dimensions and in Euclidean
signature. The general cases of arbitrary dimensionality and signature
are treated in detail in~\cite{ECPLinfty}, and later we shall also
deal with the four-dimensional case. We comment below on the
differences when one extends these constructions to higher dimensions
(independently of the signature of spacetime, see~\cite{ECPLinfty}).

Let $M$ be a parallelizable three-manifold; for this it suffices to
suppose that $M$ is orientable, a condition we will need anyway later on. The graded vector space
$V$ underlying the $L_\infty$-algebra which organises ECP gravity on $M$ is given
by
\begin{align}\label{eq:ECPvectorspace}
V:= V_{0} \oplus
  V_{1} \oplus V_{2}\oplus
  V_3
\end{align}
where 
\begin{align*}
V_{0}&=\Gamma(TM)\times \Omega^{0}\big(M,\mathfrak{so}(3)\big) \ ,  \\[4pt] 
V_{1}&= \Omega^{1}(M,\FR^{3}) \times
\Omega^{1}\big(M,\mathfrak{so}(3)\big) \ , \\[4pt]
V_{2}&=\Omega^{2}\big(M,\midwedge^{2}\FR^{3}\big)\times
\Omega^{2}\big(M,\FR^{3}\big) \ , \\[4pt]
V_3&= \Omega^{1}\big(M,\Omega^{3}(M)\big)
\times
\Omega^3\big (M,\FR^{3}\big)
\ .
\end{align*} 
In the following we denote elements in degrees~$0$, $1$, $2$ and $3$
respectively by $(\xi,\rho)\in V _0$,
$(e,\omega)\in V _1$, $(E,{\mit\Omega})\in V_2 $,
and $({\CX},{\CP})\in V_3 $. 
This yields a 4-term $L_\infty$-algebra with differential defined by 
\begin{align*}
\ell_{1}(\xi,\rho)=(0,\dd\rho) \ , \quad
  \ell_1 (e,\omega)=(\dd\omega,\dd e) \qquad \mbox{and} \qquad \ell_{1} (E,{\mit\Omega})=(0, \dd
  {\mit\Omega}) \ . 
\end{align*}
The $2$-brackets are defined by
\begin{align}
\ell_{2}\big((\xi_{1},\rho_{1})\,,\,(\xi_{2},\rho_{2})\big)&=\big([\xi_{1},\xi_{2}]_\frv\,,\,-[\rho_{1},\rho_{2}]_{\aso(3)}+\LL_{\xi_1}\rho_{2}
- \LL_{\xi_{2}}\rho_{1}\big) \ , \nn
\\[4pt]
\ell_{2}\big((\xi,\rho)\,,\,(e,\omega)\big)&=\big(-\rho \cdot e
+\LL_{\xi}e \,,\, -[\rho,\omega]_{\aso(3)}+\LL_{\xi}\om\big) \
, \nn \\[4pt]
\ell_{2}\big((\xi,\rho) \,,\, (E,{\mit\Omega})\big)&=\big(-
[\rho, E]_{\aso(3)}+\LL_{\xi}E \,,\, -\rho \cdot
{\mit\Omega}+\LL_{\xi}{\mit\Omega} \big) \ , \nn \\[4pt]
\ell_{2}\big((\xi,\rho)\,,\,({\CX},{\CP})\big)&= \big
(\dd x^\mu\otimes\Tr(\iota_\mu \dd \rho \dwedge
{\CP}) +\LL_{\xi}{\CX}\,,\,-\rho
\cdot {\CP} +\LL_{\xi} {\CP}\big) \ , \nn \\[4pt] 
\ell_{2}\big((e_{1},\omega_{1})\,,\,(e_{2},\omega_{2})\big)&=-\big([\omega_{1},
                                                             \omega_{2}]_{\aso(3)}
\,,\, \omega_{1}\wedge
e_{2} +\omega_{2} \wedge e_{1} \big) \ , \nn \\[4pt]
\ell_{2}\big((e,\om)\,,\,(E,{\mit\Omega})\big)&=\big(\dd x^\mu\otimes\Tr ( \iota_\mu \dd e \dwedge E + \iota_\mu \dd\om \dwedge {\mit\Omega} \nn   -\iota_\mu e \dwedge \dd E - \iota_\mu \om\dwedge \dd {\mit\Omega}) \,, \nn \\ 
&\phantom{{}(\Tr \big( \iota_\mu \dd e \dwedge E +{}}
E\wedge e - \omega \wedge {\mit\Omega}\big) \ , \label{eq:3dbrackets} 
\end{align}
while all the rest of the brackets vanish. As shown
in~\cite{ECPLinfty}, this defines a differential graded Lie algebra.

Let us explain the notation used, as well as the content of these
brackets following from the general prescription of Section~\ref{sec:Linftygft}. The
dynamical fields of ECP gravity on $M$ are pairs $(e,\omega)$
consisting of a coframe field $e \in \Omega^{1}(M,\FR^{3})$
and a spin connection $\om \in \Omega^{1}\big(M,\mathfrak{so}(3)\big)$. Using the isomorphism
$\mathfrak{so}(3)\cong \midwedge^{2} \FR^{3}$, one identifies the
spin connection as an element $\om \in \Omega^{1}(M,\midwedge^{2}
\FR^{3})$; we may then expand $e=e^a\,{\tt E}_a$ and $\omega=\omega^{ab}\,{\tt E}_a\wedge{\tt E}_b$, where $e^a$ and $\omega^{ab}=-\omega^{ba}$ are $1$-forms on $M$, and ${\tt E}_a$ for $a=1,2,3$ form the standard orthonormal basis of $\FR^3$. The $\dwedge$-product separately takes the
exterior products of the `curved' spacetime differential forms on the
manifold $M$
and `flat' spacetime multivectors on the vector space $\FR^{3}$; in particular, $\omega\dwedge e = \omega^{ab}\wedge e^c \, {\tt E}_a\wedge {\tt E}_b\wedge {\tt E}_c$ while $\omega\wedge e = \delta_{bc}\,\omega^{ab}\,\wedge e^c\, {\tt E}_a$. Then
$\Tr:\midwedge^3\FR^{3}\to\FR$ is the Hodge duality
operator on $\FR^3$ endowed with the standard Euclidean metric, defined by $\Tr({\tt E}_a\wedge {\tt E}_b\wedge {\tt E}_c) = \epsilon_{abc}$.

The gauge parameters are pairs $(\xi,\rho)$ consisting of a vector field
$\xi\in\Gamma(TM)$ and a local rotation
$\rho\in\Omega^0(M,\mathfrak{so}(3))$. Following the prescription of
Section~\ref{sec:Linftygft}, their action on fields is given by
\begin{align*}
\delta_{(\xi,\rho)}(e,\om) =
  \ell_1(\xi,\rho)+\ell_2\big((\xi,\rho)\,,\,(e,\om)\big) = (-\rho\cdot e+\LL_\xi
  e\,,\,\dd \rho - [\rho,\om]_{\aso(3)}+\LL_\xi\om) 
\end{align*}
where $\LL_\xi$ is the Lie derivative along the vector field $\xi$, and
$\rho\cdot e = \rho^a{}_b\,e^b \, {\tt E}_a$ denotes matrix multiplication in the fundamental
representation of $\aso(3)$.
Thus $\xi$ generates the standard diffeomorphism symmetry of
general relativity, while $\rho$ generates changes of
orthonormal coframes for a given metric. The gauge algebra
\begin{align*}
[\![(\xi_1,\rho_1)\,,\,(\xi_2,\rho_2)]\!]_{(e,\om)} =
  -\ell_{2}\big((\xi_{1},\rho_{1})\,,\,(\xi_{2},\rho_{2})\big) =\big(-[\xi_{1},\xi_{2}]_\frv\,,\,[\rho_{1},\rho_{2}]_{\aso(3)}-\LL_{\xi_1}\rho_{2}
+ \LL_{\xi_{2}}\rho_{1}\big) 
\end{align*}
is the semi-direct product of the Lie algebras of vector fields on $M$
and of local rotations:
\begin{align}\label{ClassicalSymmetries}
\Gamma(TM)\ltimes \Omega^{0}\big(M,\mathfrak{so}(3)\big) \ .
\end{align}
We denote the Lie bracket of vector fields on $\frv:=\Gamma(TM)$ by
$[-,-]_\frv$.

The field equations are given by
\begin{align*}
F_{(e,\omega)} =\ell_{1}(e,\omega)-\tfrac{1}{2}\,
  \ell_{2}\big((e,\omega)\,,\,(e,\omega)\big) 
=(\dd \om,\dd e)+ \tfrac{1}{2} \,
  ([\omega,\omega]_{\aso(3)},2\,\omega\wedge
                                                        e) = (R,T) \ ,
\end{align*}
where
\begin{align*}
R:=\dd \om +
\tfrac{1}{2}\,[\om,\om]_{\aso(3)} \ \in \ \Omega^2\big(M,\mathfrak{so}(3)\big)
\end{align*}
is the curvature $2$-form of the spin connection $\om$, while the
covariant derivative of the coframe field
\begin{align*}
T:= \dd^{\om}e = \dd e + \om \wedge e \, \, \, \in \, \, \,
  \Omega^{2}(M,\FR^{3})
\end{align*} 
is the torsion $2$-form. The field equations $F_{(e,\om)}=(0,0)$ thus
imply two conditions: Firstly, that the torsion vanishes, $T=0$, which for non-degenerate coframe
fields $e$ can be solved to identify $\omega$ with the Levi--Civita
connection of the metric $\delta^{ab}\,e_a\otimes e_b$. Secondly, that
the curvature vanishes, $R=0$, which is just the vacuum Einstein field
equation in three dimensions, whose solutions correspond to flat
spacetimes which are locally isometric to Euclidean space. In this way
we can recover the usual space of classical physical states of
three-dimensional general relativity (in the absence of matter fields). However,
the ECP theory makes sense for degenerate coframes $e$, and this
extension of general relativity is necessary for its
$L_\infty$-algebra formulation where the space of fields is required
to be a vector (or an affine) space. The covariance of the field
equations is expressed through
\begin{align*}
\delta_{(\xi,\rho)}F_{(e,\om)} =
  \ell_2\big((\xi,\rho)\,,\,F_{(e,\om)}\big) =
  (-[\rho,R]_{\aso(3)}+\LL_\xi
  R,-\rho\cdot
  T+\LL_\xi T) \ .
\end{align*}

Finally,
\begin{align*}
\dsf_{(e,\om)}F_{(e,\om)} &=
                 \ell_1(F_{(e,\om)})-\ell_2\big((e,\om)\,,\,F_{(e,\om)}\big)
  \\[4pt]
  &= \big(\dd x^\mu\otimes\Tr(\iota_\mu e\dwedge\dd
  R-\iota_\mu\dd e\dwedge R +
  \iota_\mu\om\dwedge\dd T -
  \iota_\mu\dd\om\dwedge T)\,,\, \dd^\om
                                    T-R\wedge e\big) \ ,
\end{align*}
where $\iota_\mu$ denotes the contraction with vectors
$\partial_\mu=\frac\partial{\partial x^\mu}$ of the 
local holonomic frame dual to the basis $\{\dd x^\mu\}$ of $1$-forms
in a local coordinate chart on $M$, and
we identify the vector space of $1$-forms valued in $3$-forms $\Omega^1\big(M,\Omega^3(M)\big) $ with
$\Omega^1(M) \otimes\Omega^3(M)$. The Noether identities
$\dsf_{(e,\om)}F_{(e,\om)}=(0,0)$ thus impose two differential
identities among the field equations, the second of which
is equivalent to the first Bianchi identity $\dd^{\om}T=R\wedge e$.

The action functional can be derived in this language by assuming $M$
is closed\footnote{Alternatively, we may work with fields of suitable asymptotic decay.} and defining a
cyclic pairing $\langle -, - \rangle : V\otimes V \rightarrow \FR$ of
degree $-3$ as~\cite{ECPLinfty}
\begin{align} \label{eq:ECPpairing}
\langle (e,\om) \,,\, (E,{\mit\Omega}) \rangle:= 
\int_{M}\, \Tr \big(e\dwedge E+ {\mit\Omega} \dwedge \om \big) \ ,
\end{align}
on $V_{1}\otimes V_{2}$, and as
\begin{align} \label{eq:ECPpairing2}
\langle(\xi,\rho)\,,\,({\CX},{\CP})\rangle := \int_M\,
\iota_\xi{\CX} + \int_M\, \Tr\big(\rho\dwedge{\CP}\big) \ ,
\end{align}
on $V_{0}\otimes V_{3}$, where $\iota_\xi$ denotes contraction with the
vector field $\xi$. Then the prescription of
Section~\ref{sec:Linftygft} yields the standard ECP action functional
in three dimensions:
\begin{align}\label{eq:ECP3daction}
S(e,\omega) &= \tfrac12\, \big\langle (e,\omega)\,,\,\ell_1(e,\omega) \big\rangle -
\tfrac1{6}\, \big\langle (e,\omega)\,,\,\ell_2\big(
              (e,\omega)\,,\,(e,\omega)\big) \big\rangle \nn \\[4pt]
            &= \int_M\, \Tr
\Big(e\dwedge\big(\dd\omega+\tfrac12\, [\omega,\omega]_{\aso(3)}\big)
              \Big) \nn
  \nn \\[4pt]
  &= \int_{M}\,\Tr \big(e\dwedge
    R \big) =\int_{M}\, \epsi_{abc} \,
    e^{a} \wedge R^{bc} \ .
\end{align}
The first Noether identity from $\dsf_{(e,\om)}F_{(e,\om)}=(0,0)$, in
$\Omega^1(M,\Omega^3(M))$, 
corresponds to local diffeomorphism invariance
$\delta_{(\xi,0)} S(e,\om)=0$, while the second Noether identity, in $\Omega^3(M,\FR^3)$, corresponds to the local gauge symmetry
$\delta_{(0,\rho)} S(e,\om)=0$. This construction can be easily
extended to incorporate a cosmological constant~\cite{ECPLinfty}, and
we will do so later on when we study the braided version of this theory.

In higher dimensions $d\geq4$, the $L_\infty$-algebras which organise ECP gravity contain
higher brackets and are no longer simply differential graded Lie
algebras, owing to the higher degree polynomial nature of the theory
in the coframes. We will consider the four-dimensional case in detail
later on, where the theory is again classically equivalent to general
relativity, but now with propagating degrees of freedom.

\section{Drinfel'd twist deformation quantization}
\label{sec:Drinfeldtwist}

\subsection{Drinfel'd twists on manifolds}

We briefly summarise the basic theory of Drinfel'd
twists~\cite{MajidBook}, specialised from the outset to the Lie
algebra of vector fields $\frv:=\Gamma(TM)$ on a manifold $M$, which
generate infinitesimal diffeomorphisms of $M$; here we regard $\frv$
as a Lie algebra over $\FC$ (or more precisely its complexification $\frv\otimes\FC$). The main object of
interest in this theory is the universal enveloping algebra $U\frv$ of
$\frv$, which is the tensor algebra (over $\FC$) of $\frv$,
regarded as the free unital algebra generated by $\frv$,
modulo the two-sided ideal generated by $\xi_1\, \xi_2 -\xi_2\, \xi_1-
[\xi_1,\xi_2]_\frv$ for all $\xi_1,\xi_2\in\frv$.

The enveloping algebra $U\frv$ is naturally a cocommutative Hopf
algebra with coproduct $\Delta:U\frv\to U\frv\otimes U\frv$, counit
$\varepsilon:U\frv\to\FC$ and antipode $S:U\frv\to U\frv$ defined on
generators by
\begin{align*}
\Delta(\xi)&= \xi\otimes 1 + 1\otimes \xi\qquad \mbox{and} \qquad
             \Delta(1)=1\otimes 1 \ , \\[4pt]
  \varepsilon(\xi)&=0 \qquad \mbox{and} \qquad \varepsilon(1)=1 \ ,
  \\[4pt]
S(\xi)&=-\xi \qquad \mbox{and} \qquad S(1)=1 \ ,
\end{align*}
for all $\xi\in \frv$. The maps $\Delta$ and $\varepsilon$ are
extended as algebra homomorphisms, and $S$ as an algebra
antihomomorphism to all of $U\frv$. We adopt the standard
Sweedler notation $\Delta(X)=:X_{\textrm{\tiny(1)}}\otimes
X_{\textrm{\tiny(2)}}$ (with summations understood) to abbreviate the
coproduct of $X\in U\frv$.

From the perspective of field theory, the Hopf algebra $U\frv$ is
generally ``too big'' for the purposes of deformation quantization, as the
ensuing constructions would then only work for
diffeomorphism-invariant field theories. This will actually be the case for the
examples treated in this paper, but to handle other field theories
with a smaller set of spacetime symmetries, one could simply replace
$\frv$ with any Lie subalgebra and apply all of the following
constructions to the corresponding enveloping Hopf subalgebra $\CH\subseteq
U\frv$. This can be interpreted as performing the deformation
quantization along the symmetries of $M$ which leave the field
theory on $M$ invariant, see for instance~\cite{Chaichian:2004za,Wess:2003da,Borowiec:2008uj,DimitrijevicCiric:2018blz} for examples.

To treat the typical examples of twists which arise in physics, we
shall need to introduce formal power series extensions in a
deformation parameter $\hbar$. If $V$ is a complex vector space, we
denote by $V[[\hbar]]$ the vector space of formal power series in
$\hbar$ with coefficients in $V$; it is naturally a module over
$\FC[[\hbar]]$. If $V$ and $W$ are complex vector spaces, then
$V[[\hbar]]\otimes W[[\hbar]]\cong(V\otimes W)[[\hbar]]$, where on the
left-hand side we use the appropriate topological tensor product (see
e.g.~\cite{Barnes:2014ksa} for further details), which for
simplicity we do not distinguish from the usual tensor product of
vector spaces over $\FC$. With these conventions, we denote by
$U\frv[[\hbar]]$ the formal power series extension of the
cocommutative Hopf algebra $U\frv$, with operations applied term by term to the coefficients of series.

A \emph{Drinfel'd twist} is a normalized $2$-cocycle of the Hopf
algebra $U\frv[[\hbar]]$. By this we mean an invertible element
$\CF\in U\frv[[\hbar]]\otimes U\frv[[\hbar]]$ satisfying the cocycle
condition
\begin{align*}
\CF_{12}\,(\Delta\otimes \id)\CF=\CF_{23}\,(\id\otimes \Delta)\CF \ ,
\end{align*}
where $\CF_{12}=\CF\otimes 1$ and $\CF_{23}=1\otimes \CF$, together
with the normalization condition
\begin{align*}
  (\varepsilon\otimes \id)\CF=1=(\id\otimes \varepsilon)\CF \ .
\end{align*}
We write the power series expansion of the twist as $\CF=:\sff^{k}\otimes
\sff_{k}\in U\frv[[\hbar]]\otimes U\frv[[\hbar]]$, with the sum over
$k$ understood. Then the cocycle
condition may be written in Sweedler notation as
\begin{align}\label{eq:cocyclesw}
\sff^k\,\sff^l_{\textrm{\tiny(1)}}\otimes\sff_k\,\sff^l_{\textrm{\tiny(2)}}\otimes\sff_l
  =
  \sff^l\otimes\sff^k\,{\sff_l}_{\textrm{\tiny(1)}}\otimes\sff_k\,{\sff_l}_{\textrm{\tiny(2)}}
  \ ,
\end{align}
and the normalization condition as
\begin{align*}
\varepsilon(\sff^k)\,\sff_k = 1 = \sff^k\,\varepsilon(\sff_k) \ .
\end{align*}
As a consequence, the inverse twist $\CF^{-1}=: \bar{\sff}^{k}\otimes
\bar{\sff}_{k}\in U\frv[[\hbar]]\otimes U\frv[[\hbar]]$ satisfies
similar conditions.

A Drinfel'd twist $\CF$ defines a new Hopf algebra structure on the
universal enveloping algebra $U\frv[[\hbar]]$, which we denote by
$U_\CF\frv$. As algebras, $U_\CF\frv=U\frv[[\hbar]]$ and also the
counit of $U_\CF\frv$ is the same as the counit $\varepsilon$ of
$U\frv[[\hbar]]$. The new coproduct $\Delta_\CF$ and antipode $S_\CF$
of $U_\CF\frv$ are given by
\begin{align*}
\Delta_{\CF}(X):= \CF \, \Delta(X) \, \CF^{-1} \qquad \mbox{and}
  \qquad S_{\CF}(X):=\sff^{k}\,S(\sff_{k})\, S(X)\,
  S(\bar{\sff}^l)\,\bar{\sff}_{l} \ ,
\end{align*}
for all $X\in U\frv[[\hbar]]$. For $X\in U_\CF\frv$, we adopt the Sweedler notation $\Delta_\CF(X) =: X_{\bar\swone}\otimes X_{\bar\swtwo}$ to distinguish the twisted and untwisted coproducts.

This new Hopf algebra is not cocommutative in general,
i.e. $\Delta_\CF\neq\Delta_\CF^{\rm op}:=\tau\circ\Delta_\CF$, where
$\tau$ is the transposition which interchanges the factors in a tensor
product. However, the cocommutativity is controlled up to a
\emph{braiding} given by the invertible $\RR$-matrix $\RR\in
U\frv[[\hbar]]\otimes U\frv[[\hbar]]$ induced by the twist as
\begin{align}
\RR=\CF_{21}\, \CF^{-1}=:\sfR^k\otimes\sfR_k \ ,
\end{align}
where $\CF_{21}=\tau(\CF)=\sff_k\otimes\sff^k$ is the twist with its legs
swapped. Explicitly
\begin{align*}
\Delta_{\CF}^{\rm op}(X)=\RR\,\Delta_{\CF}(X)\,\RR^{-1} \ .
\end{align*}
It is easy to see that the $\RR$-matrix is triangular, that is
\begin{align*}
  \RR_{21} = \RR^{-1} = \sfR_k\otimes\sfR^k \ ,
\end{align*}
  and moreover that 
\begin{align*}
(\Delta_{\CF}\otimes \id) \RR = \RR_{13}\, \RR_{23} \qquad \mbox{and} \qquad
(\id\otimes \Delta_{\CF})\RR= \RR_{13}\, \RR_{12} \ , 
\end{align*}
where $\RR_{13}=\sfR^k\otimes 1 \otimes\sfR_k$, or in Sweedler notation
\begin{align}\label{eq:Rmatrixidsw}
\sfR^{k}_{\bar{\textrm{\tiny(1)}}}\otimes \sfR^{k}_{\bar{\textrm{\tiny(2)}}}\otimes \sfR_{k}= \sfR^{l}\otimes \sfR^{k}\otimes
  \sfR_{l}\, \sfR_{k} \qquad \mbox{and} \qquad \sfR^{k}\otimes
  {\sfR_{k}}_{\bar{\textrm{\tiny(1)}}}\otimes {\sfR_{k}}_{\bar{\textrm{\tiny(2)}}} =
  \sfR^{l}\,\sfR^{k}\otimes \sfR_{k} \otimes \sfR_{l} \ .
\end{align}
The $\RR$-matrix also satisfies the Yang--Baxter equation 
\begin{align}\label{eq:YangBaxter}
\RR_{12} \, \RR_{13} \, \RR_{23}= \RR_{23} \, \RR_{13} \, \RR_{12} \ .
\end{align}

\begin{example}\label{ex:MoyalWeyltwist}
The standard example on $M=\FR^d$ is the
abelian Hermitian Moyal--Weyl twist
\begin{align}
\CF_\theta =
  \exp\big(-\tfrac{\mathrm{i}\,\hbar}2\,\theta^{\mu\nu}\,\partial_\mu\otimes\partial_\nu\big) =:\sff^k_\theta\otimes\sff_{\theta k}
  \ ,
\end{align}
where $(\theta^{\mu\nu})$ is a $d{\times}d$ antisymmetric real-valued
matrix. This twist is based on the enveloping Hopf algebra $\CH$ of
the abelian Lie algebra of infinitesimal translations. In this case the twisted Hopf
algebra is cocommutative (in fact $\Delta_{\CF_\theta}=\Delta$), and the
$\RR$-matrix is given by
\begin{align*}
\RR_\theta = \CF_\theta^{-2} =
  \exp\big(\,\mathrm{i}\,\hbar\,\theta^{\mu\nu}\,\partial_\mu\otimes\partial_\nu\big)
  \ .
\end{align*}
However, our
considerations in the following apply to a more general class of
Drinfel'd twists that we shall specify more precisely later on.
\end{example}

\subsection{Modules}
\label{sec:twistedmodules}

Drinfel'd twist deformation quantization consists in twisting the enveloping
Hopf algebra $U\frv$ to a non-cocommutative Hopf algebra $U_\CF\frv$,
while simultaneously twisting all of its modules~\cite{SpringerBook}. There is a symmetric
monoidal category $\CCM$ whose objects are (left)
$U\frv$-modules and whose morphisms are equivariant maps (see
e.g.~\cite{Barnes:2014ksa}). Since $U\frv$ is a cocommutative Hopf
algebra (equivalently it has a triangular structure with trivial
$\RR$-matrix $1\otimes1$), the braiding isomorphism of $\CCM$ is just the trivial
transposition $\tau$.

A $U\frv$-module algebra is an algebra in the category
$\CCM$. By this we mean an algebra $(\CA,\mu)$ with a
$U\frv$-action $\triangleright: U\frv\otimes \CA
\rightarrow \CA$ which is compatible with the algebra multiplication
via the coproduct $\Delta$, that is
\begin{align*}
X\triangleright\mu(a\otimes b)
  = \mu\big(\Delta(X)\triangleright(a\otimes b)\big)
\end{align*}
for all $X\in U\frv$ and $a,b\in\CA$, where $\mu:\CA\otimes\CA\to\CA$
is the product on $\CA$. We will usually drop the symbol
$\triangleright$ to simplify the notation. This condition means in particular that
vector fields $\xi\in\frv$ act on $(\CA,\mu)$ as derivations:
\begin{align*}
  \xi\big(\mu(a\otimes b)\big) = \mu\big(\xi(
  a)\otimes b\big) + \mu\big(a\otimes\xi( b)\big) \ .
\end{align*}

A Drinfel'd twist $\CF$ defines a functorially equivalent symmetric monoidal
category ${}_\CF\CCM$ of left
$U_\CF\frv$-modules~\cite{Barnes:2014ksa}. The braiding isomorphism of
${}_\CF\CCM$ is now non-trivial and given by composing the transposition
$\tau$ with the action of the inverse of the $\RR$-matrix. Since $\RR$ is
triangular, $\RR_{21}=\RR^{-1}$, the braiding is symmetric,
i.e. the braiding isomorphism squares to the identity
morphism. The action of the inverse $\RR$-matrix on arbitrary tensor products of $U_\CF\frv$-modules $\CV_1\otimes\cdots \otimes\CV_n$ can be computed using the identities \eqref{eq:Rmatrixidsw} for $\RR_{21}=\RR^{-1}$ and their iterations to get
\begin{align}\label{eq:Rmatrixidswn}
\sfR_k(v_1\otimes\cdots\otimes v_{n-1})\otimes\sfR^k( v_n) &= \sfR_{k_1}(v_1)\otimes\cdots\otimes\sfR_{k_{n-1}}(v_{n-1})\otimes \sfR^{k_{n-1}}\cdots\sfR^{k_{1}}(v_n) \ , \nonumber\\[4pt]
\sfR_k(v_1)\otimes\sfR^k(v_2\otimes\cdots\otimes v_n)&= \sfR_{k_1}\cdots\sfR_{k_{n-1}}(v_1)\otimes\sfR^{k_1}(v_2)\otimes\cdots\otimes \sfR^{k_{n-1}}(v_n) \ ,
\end{align}
for $n\geq2$ and $v_i\in \CV_i$. Intuitively, the identities \eqref{eq:Rmatrixidswn} tell us that passing an element ``at once" over many elements $v_i$ is the same as passing successively over ``each one'' individually.

If $(\CA,\mu)$ is a (left) $U\frv$-module algebra, then we
can deform the product $\mu$ on $\CA$ by precomposing it with the
inverse of the twist
$\CF$ to get a new product
\begin{align}\label{eq:mustar}
\mu_\star (a\otimes b) = \mu\circ\CF^{-1}(a\otimes b) =
  \mu\big(\bar\sff^k(a)\otimes\bar\sff_k(b)\big) \ ,
\end{align}
for $a,b\in\CA$, where on the right-hand side we extend $\mu$ to
$\CA[[\hbar]]\otimes\CA[[\hbar]] \cong(\CA\otimes\CA)[[\hbar]]$ by
applying it term by term to the coefficients of a formal power series.
The cocycle condition on $\CF$ guarantees that this
produces an associative star-product $\mu_\star$ on $\CA[[\hbar]]$, and it generally defines a
noncommutative $U_\CF\frv$-module algebra $(\CA[[\hbar]],\mu_\star)$,
that is, an algebra 
in the category of $U_\CF\frv$-modules:
\begin{align*}
X\big(\mu_\star(a\otimes b)\big) =
  \mu_\star\big(\Delta_\CF(X)(a\otimes b)\big) \ ,
\end{align*}
for all $X\in U\frv$ and $a,b\in\CA$.
In the following we denote $(\CA,\mu)$ by $\CA$ and
$(\CA[[\hbar]],\mu_\star)$ by $\CA_\star$ for brevity. 

If the algebra $\CA$ is commutative,
then $\CA_\star$ is braided commutative: the noncommutativity of
$\CA_\star$ is controlled by the $\RR$-matrix as
\begin{align*}
\mu_\star(a\otimes b) = \mu_\star\big(\sfR_k(b)\otimes\sfR^k(a)\big) \ ,
\end{align*}
which is easily proven by recalling that $\RR=\CF_{21}\,
\CF^{-1}$. 

\begin{example}\label{ex:MoyalWeylstar}
We apply this construction to
Example~\ref{ex:MoyalWeyltwist}. Let $\CA$ be the commutative algebra of smooth
complex-valued functions
on $\FR^d$, with the usual pointwise multiplication 
$\mu(f\otimes g)=f\,g$. Then we obtain the noncommutative Moyal--Weyl star-product
\begin{align*}
 f\star_\theta g &=
  \mu\circ\CF_\theta^{-1}(f\otimes g) \\[4pt] &=
                                                {\bar{\sff}_\theta}^k(f)\,{\bar{\sff}_{\theta}}{}_k(g) \\[4pt]
&=f\,g + \sum_{k=1}^\infty\,\frac{(\mathrm{i}\,\hbar)^k}{2^k\,k!}\,
                                                \theta^{\mu_1\nu_1}\cdots\theta^{\mu_k\nu_k}\,
                                                \partial_{\mu_1}\cdots
                                                \partial_{\mu_k}f \,
                                                \partial_{\nu_1}\cdots \partial_{\nu_k}g
\end{align*}
of smooth functions $f$ and $g$ on $\FR^d$.
\end{example}

The tensor algebra of differential forms and vector fields on the manifold $M$ is covariant under the action
of the universal enveloping algebra of infinitesimal diffeomorphisms
$U\frv$. We can construct a noncommutative differential geometry on
$M$ by requiring it to be covariant with respect to the twisted Hopf
algebra $U_\CF\frv$~\cite{SpringerBook, NAGravity}. We illustrate this
below by considering the deformation quantization of the exterior
algebra of differential forms and the (Lie and universal enveloping)
algebra of vector fields. Drinfel'd twist deformation of differential and Cartan calculus on a manifold $M$ has been studied in~\cite{Aschieri:2005zs,NAGravity,Weber:2019ryz,Aschieri:2020ifa}.

\subsection{Differential forms}
\label{sec:twistedforms}

Let $\CA =\big(\Omega^{\bullet}(M),\wedge\big)$ be the exterior
algebra of differential forms on $M$ (or
more precisely its complexification). This is a module over the Lie
algebra of vector fields $\frv$, where the (left) action is given by
the Lie derivative:
\begin{align*}
\xi(\alpha) := \LL_\xi\alpha
\end{align*}
for $\xi\in\frv$ and $\alpha\in\Omega^\bullet(M)$. This action is
extended to a $U\frv$-action via successive applications of the Lie
derivative, viewed as higher order differential operators:
\begin{align*}
\xi_1\cdots\xi_n(\alpha) =
  (\LL_{\xi_1}\circ\cdots\circ\LL_{\xi_n})\alpha 
\end{align*}
for $\xi_1,\dots,\xi_n\in\frv$. This makes
$\big(\Omega^\bullet(M),\wedge\big)$ into a $U\frv$-module
algebra, since
\begin{align*}
\xi (\alpha\wedge \beta) = \LL_{\xi}(\alpha\wedge \beta)=\LL_{\xi} \alpha \wedge \beta+ \alpha \wedge \LL_{\xi} \beta =\wedge \circ \Delta(\xi) (\alpha\otimes \beta)
\end{align*}
for all $\xi\in\frv$ and $\alpha,\beta\in
\Omega^{\bullet}(M)$.

The twisted exterior algebra
$\CA_\star=\big(\Omega^\bullet(M)[[\hbar]],\wedge_\star\big)$ is a
$U_\CF\frv$-module algebra. Following the prescription
\eqref{eq:mustar}, for $\alpha,\beta\in
\Omega^{\bullet}(M)$ we set
\begin{align*}
\alpha\wedge_\star\beta:=
\bar\sff^k(\alpha) \wedge \bar\sff_k(\beta) \ ,
\end{align*}
with (graded) braided commutativity controlled by the $\RR$-matrix as
\begin{align*}
\alpha\wedge_{\star} \beta = (-1)^{|\alpha| \,|\beta|} \ \sfR_{k} (\beta)
  \wedge_{\star} \sfR^{k} (\alpha) \ .
\end{align*}

The standard (undeformed) exterior derivative
$\dd:\Omega^\bullet(M)\to\Omega^\bullet(M)$ is then an ordinary (not
braided!) graded derivation of the twisted exterior algebra $\CA_\star$:
\begin{align*}
\dd(\alpha\wedge_{\star} \beta) = \dd\alpha\wedge_\star\beta +
  (-1)^{|\alpha|}\,\alpha\wedge_\star\dd\beta \ .
\end{align*}
This property is fulfilled because the usual exterior derivative
commutes with the Lie derivatives that enter in the definition of the
star-exterior product. It is also still a differential on $\CA_\star$,
that is, $\dd^2=0$. This twists the de~Rham complex $(\Omega^\bullet(M),\dd)$ of the
manifold $M$ to a noncommutative differential graded algebra~$(\CA_\star,\dd)$. 

In our applications to field theory, we shall be interested in a special
class of Drinfel'd twists. If $M$ is a closed oriented manifold of
dimension $d$, we will require that the usual integral over $M$ is graded
cyclic under the deformed exterior
product $\wedge_\star$:
\begin{align}\label{eq:intcyclic}
\int_M\,\alpha\wedge_\star\beta= (-1)^{|\alpha|\,|\beta|} \,
  \int_M\,\beta\wedge_\star\alpha \ ,
\end{align}
where $|\alpha|+|\beta|=d$. This is guaranteed if the twist $\CF$
satisfies the condition
$S(\bar\sff^k)\,\bar\sff_k=1$, where $S$ is the antipode of
$U\frv$~\cite{AschCast}; for example, this holds for abelian twists\footnote{That is, a Drinfel'd twist generated by an abelian subalgebra of $\frv$, such as the Moyal--Weyl twist of Example \ref{ex:MoyalWeyltwist}.}, whereby $\CF_{21}=\CF^{-1}$. We shall further
require that the twist $\CF$ is Hermitian, that is, it defines a
Hermitian star-product $\wedge_\star$:
\begin{align*}
(\alpha\wedge_\star\beta)^* = (-1)^{|\alpha|\,|\beta|} \,
  \beta^*\wedge_\star\alpha^* \ ,
\end{align*}
where ${}^*$ denotes complex conjugation. This is guaranteed if $\CF$
satisfies the reality condition given by $\bar\sff^k{}^*\otimes\bar\sff_k{}^* =
S(\bar\sff_k)\otimes S(\bar\sff^k)$. 

All of this has an extension to the graded Lie algebra
$\CA=(\Omega^\bullet(M,\frg),[-,-]_\frg)$ of differential forms valued
in a Lie algebra $\frg$, where $[-,-]_\frg$ denotes the tensor product
of exterior multiplication of forms with the Lie bracket of
$\frg$. With the trivial $U\frv$-action on $\frg$, we use \eqref{eq:mustar} to define
\begin{align*}
[\alpha,\beta]_\frg^\star:=[\bar\sff^k(\alpha),\bar\sff_k(\beta)]_\frg \ ,
\end{align*}
for $\alpha,\beta\in\Omega^\bullet(M,\frg)$.
This makes
$\CA_\star=(\Omega^\bullet(M,\frg)[[\hbar]],[-,-]_\frg^\star)$ into a (graded)
\emph{braided Lie algebra}, that is, a graded Lie algebra in the category
${}_\CF\CCM$~\cite{Barnes:2015uxa}. The new bracket 
is now (graded) braided antisymmetric: 
\begin{align*}
[\alpha_{1},\alpha_{2}]_\frg^{\star}=-(-1)^{|\alpha_1|\,|\alpha_2|} \,
  [\sfR_{k}(\alpha_{2}),\sfR^{k}(\alpha_{1})]_\frg^\star 
  \ ,
\end{align*}
and satisfies the (graded) braided Jacobi identity:
\begin{align}\label{eq:braidedJacobiforms}
[\alpha_{1},[\alpha_{2},\alpha_{3}]_\frg^{\star}]_\frg^{\star}
  =[[\alpha_{1},\alpha_{2}]_\frg^{\star},\alpha_{3}]_\frg^{\star} +
  (-1)^{|\alpha_1|\,|\alpha_2|}\,
  [\sfR_{k}(\alpha_{2}),[\sfR^{k}(\alpha_{1}) ,\alpha_{3}]_\frg^{\star}]_\frg^{\star}
  \ ,
\end{align}
for all $\alpha_{1},\alpha_{2},\alpha_{3} \in
\Omega^\bullet(M,\frg)$; this will follow as a special example of Proposition~\ref{prop:braidedfromclassical} below. 
Graded cyclicity can then be formulated
as above by precomposing the integral over $M$ with an invariant
quadratic form $\Tr_\frg$ on the Lie algebra $\frg$.

\subsection{Diffeomorphisms}
\label{sec:braideddiff}

Since we are also interested in diffeomorphism symmetry in our
applications to gravity, we have to consider the Drinfel'd twist
deformation of the algebra of vector fields
$\frv=\Gamma(TM)$. Deforming the Lie bracket following the general
prescription \eqref{eq:mustar}, we obtain the star-bracket of two
vector fields $\xi_1,\xi_2\in\Gamma(TM)$ as
\begin{align*}
[\xi_1,\xi_2]_\frv^\star = [\bar\sff^k(\xi_1),\bar\sff_k(\xi_2)]_\frv =
  \xi_1\star\xi_2 - \sfR_k(\xi_2)\star\sfR^k(\xi_1) \ ,
\end{align*}
where the second equality takes place in the twisting of the universal
enveloping algebra of vector fields $U\frv$. The star-bracket makes
$(\Gamma(TM)[[\hbar]],[-,-]_\frv^\star)$ into the braided Lie algebra
of vector fields. The star-bracket is braided antisymmetric:
\begin{align*}
[\xi_1,\xi_2]_\frv^\star = -[\sfR_k(\xi_2),\sfR^k(\xi_1)]_\frv^\star \ ,
\end{align*}
and it also fulfills the braided Jacobi identity
\begin{align*}
\big[\xi_1,[\xi_2,\xi_3]_\frv^\star\big]_\frv^\star =
  \big[[\xi_1,\xi_2]_\frv^\star,\xi_3\big]_\frv^\star +
  \big[\sfR_k(\xi_2),[\sfR^k(\xi_1),\xi_3]_\frv^\star\big]_\frv^\star
  \ ,
\end{align*}
for all $\xi_1,\xi_2,\xi_3\in\Gamma(TM)$.

Braided diffeomorphisms act via the star-Lie derivative, which is defined following \eqref{eq:mustar} as
\begin{align*}
\LL_\xi^\star T = \LL_{\bar\sff^k(\xi)} \bar\sff_k(T) \ ,
\end{align*}
for any vector field $\xi$ and any tensor field $T$ on $M$. It provides a
representation of the braided Lie algebra of vector fields on
differential forms and tensor fields, since it
satisfies
\begin{align*}
[\LL_{\xi_1}^\star,\LL_{\xi_2}^\star]_\circ^\star :=
  \LL_{\xi_1}^\star\circ\LL_{\xi_2}^\star -
  \LL_{\sfR_k(\xi_2)}^\star\circ \LL_{\sfR^k(\xi_1)}^\star =
  \LL_{[\xi_1,\xi_2]_\frv^\star}^\star \ .
\end{align*}
It is also a \emph{braided derivation}, that is
\begin{align*}
\LL_\xi^\star(T_1\otimes_\star T_2) = \LL_\xi^\star T_1\otimes_\star T_2 +
  \sfR_k(T_1)\otimes_\star \LL_{\sfR^k(\xi)}^\star T_2 \ ,
\end{align*}
for any vector field $\xi$ and any two tensors $T_1,T_2$ on $M$, where
$T_1\otimes_\star
T_2:=\bar\sff^k(T_1)\otimes_{C^\infty(M)}\bar\sff_k(T_2)$. These properties of the star-Lie derivative will be generalized to more general braided Lie algebra actions in Section~\ref{sec:bft}.

\section{Braided $L_\infty$-algebras}
\label{sec:braidedLinfty}

\subsection{$L_\infty$-algebras in braided representation categories}
\label{sec:Linftycat}

Let $(V,\{\ell_n\})$ be an $L_\infty$-algebra as defined in
Section~\ref{sec:Linfty}. The key to defining our notion of a braided
$L_\infty$-algebra consists of reinterpreting all the defining
relations in terms of morphisms in a suitable category, before acting
on elements of $V$. This reformulation also has the advantage of
absorbing the cumbersome sign factors which arise from gradings in the
homotopy relations of an $L_\infty$-algebra.

To illustrate
the idea, let us begin by considering the simplest case where only the $2$-bracket
$\ell_2$ is non-vanishing. Let $\CCV$ be the symmetric monoidal
category of vector spaces (sitting in degree~$0$), whose braiding isomorphism
is given by transposition $\tau$ which exchanges factors of a tensor
product of two vector spaces. Then symmetric and antisymmetric morphisms in $\CCV$ are
defined by acting with the symmetric group $S_n$ of degree~$n$ via
the braiding $\tau$. As an object in $\CCV$, the vector space $V$
carries the structure of a Lie algebra: the morphism $\ell_2:V\otimes
V\to V$ is antisymmetric:
\begin{align*}
\ell_2=-\ell_2\circ\sigma_{(12)}=-\ell_2\circ\tau^{\phantom{\dag}}_{V,V} \ ,
\end{align*}
where $\sigma_{(12)}\in S_2$ is the non-identity permutation of
degree~$2$; on elements this reads
\begin{align*}
  \ell_2(v_1\otimes v_2)=-\ell_2\big(\tau^{\phantom{\dag}}_{V,V}(v_1\otimes
  v_2)\big)=-\ell_2(v_2\otimes v_1) \ .
\end{align*}
  The Jacobi
identity for $\ell_2$ in this framework is $\mathcal{J}_3=0$, where the Jacobiator is
the morphism $\mathcal{J}_3:V\otimes V\otimes V\to V$ in $\CCV$
defined via elements $\sigma\in S_3$ permuting the entries in $V\otimes V\otimes V$ by
\begin{align}\label{eq:Jacobiator}
\mathcal{J}_3:=\ell_2\circ(\ell_2\otimes\id)\circ\big(\sigma_{\id}+\sigma_{(132)}-\sigma_{(23)}\big)
  = \ell_2\circ(\ell_2\otimes\id)\circ\big(\id^{\otimes3}+\tau^{\phantom{\dag}}_{V,V\otimes
  V}-\id\otimes\tau^{\phantom{\dag}}_{V,V}\big) \ .
\end{align}
It is easy to check on elements of $V$ that $\CJ_3(v_1,v_2,v_3)=0$ is
just the standard Jacobi identity for the $2$-bracket $\ell_2$.

Let now $\CCV^\sharp$ be the symmetric monoidal
category of $\RZ$-graded vector spaces, whose braiding
isomorphism is given by transposition times the degree swap
multiplication, which we denote by $\tau^\sharp$. Now symmetric and
antisymmetric morphisms in $\CCV^\sharp$ are defined using the
braiding $\tau^\sharp$ to act with the symmetric group $S_n$, and
yield graded versions. As an object in $\CCV^\sharp$, the vector space
$V$ thus carries the structure of a graded Lie algebra: the morphism
$\ell_2:V\otimes V\to V$ of degree~$0$ is (graded) antisymmetric:
\begin{align*}
\ell_2=-\ell_2\circ\sigma_{(12)} = -\ell_2\circ\tau_{V,V}^\sharp \ ,
\end{align*}
which on homogeneous elements reads
\begin{align*}
  \ell_2(v_1\otimes v_2)=-\ell_2\big(\tau_{V,V}^\sharp(v_1\otimes
  v_2)\big) = -(-1)^{|v_1|\,|v_2|}\,\ell_2(v_2\otimes v_1) \ .
\end{align*}
  The graded
Jacobi identity for $\ell_2$ in this language is analogous to the vanishing of the
Jacobiator \eqref{eq:Jacobiator}, with $\tau$ replaced by
$\tau^\sharp$. We have already encountered a concrete example 
in Section~\ref{sec:twistedforms}: in the present context,
$(\Omega^\bullet(M,\frg),[-,-]_\frg)$ is a Lie algebra in the
representation category $\CCM^\sharp$ of $\RZ$-graded $U\frv$-modules,
with $U\frv$ itself regarded as sitting in degree~$0$.

It should now be clear how to pass to the definition of a braided Lie
algebra: one simply requires a permutation $\sigma\in S_n$ to act via
a (graded) braided transposition in a braided monoidal category. This can be defined in any symmetric
monoidal category, but for concreteness, and with any eye to the
applications we are interested in later on, we specialise immediately
to the twisted representation category ${}_\CF\CCM^\sharp$ of
$\RZ$-graded left $U_\CF\frv$-modules. Similarly to
Section~\ref{sec:twistedforms}, the braiding isomorphism in this
category is determined by the $\RR$-matrix as ${}_{\textrm{\tiny$\CF$}}^{\phantom{\dag}}\tau^\sharp:=\RR^{-1}\circ\tau^\sharp$. Then as an
object in ${}_\CF\CCM^\sharp$, the vector space $V$ carries the
structure of a braided Lie algebra: the morphism $\ell_2:V\otimes V\to
V$ satisfies
\begin{align*}
\ell_2=-\ell_2\circ\sigma_{(12)}=-\ell_2\circ{}_{\textrm{\tiny$\CF$}}^{\phantom{\dag}}\tau_{V,V}^\sharp \ ,
\end{align*}
which on homogeneous elements reads
\begin{align*}
  \ell_2(v_1\otimes v_2)=-\ell_2\big({}_{\textrm{\tiny$\CF$}}^{\phantom{\dag}}\tau_{V,V}^\sharp(v_1\otimes
v_2)\big) = -(-1)^{|v_1|\,|v_2|}\,\ell_2\big(\sfR_k(v_2)\otimes
  \sfR^k(v_1)\big) \ .
\end{align*}
The braided Jacobi identity can be expressed as in
\eqref{eq:Jacobiator} by replacing $\tau$ with ${}_{\textrm{\tiny$\CF$}}^{\phantom{\dag}}\tau^\sharp$. In
Section~\ref{sec:twistedforms} we already encountered the concrete
example of the Lie algebra
$(\Omega^\bullet(M,\frg)[[\hbar]],[-,-]_\frg^\star)$ in the symmetric
monoidal category \smash{${}_\CF\CCM^\sharp$}.

At this stage it is straightforward to pass from Lie algebras to
$L_\infty$-algebras, allowing other brackets $\ell_n$ to be
non-zero for $n\neq2$. An $L_\infty$-algebra consists of an object $V$
in the category $\CCV^\sharp$ together with morphisms $\ell_n:
\text{\Large$\otimes$}^{n} V \rightarrow V$ which are
antisymmetric in $\CCV^\sharp$. The key point is that the homotopy
relations \eqref{eq:calJndef} can now be expressed, without acting
explicitly on elements of $V$, as the vanishing of
morphisms $\CJ_n: \text{\Large$\otimes$}^{n}V\to V$ in $\CCV^\sharp$
defined by
\begin{align}\label{eq:Jnmaps}
{\cal J}_n := \sum^n_{i=1}\, (-1)^{i\,(n-i)} \
\sum_{\sigma\in{\rm Sh}_{i,n-i}} \, \text{sgn}(\sigma) 
\ \ell_{n+1-i}\circ \big(
\ell_i\otimes \id^{\otimes{n-i}}\big) \circ \sigma
\ , 
\end{align}
where the permutations $\sigma$ act via the graded transposition
$\tau^\sharp$ on
$\text{\Large$\otimes$}^{n} V $, and ${\rm sgn}(\sigma)$ is the sign of
$\sigma$. It is easy to check that on elements $v_1,\dots,v_n\in V$
this coincides with \eqref{eq:calJndef}.

This leads us to the central mathematical concept of this paper.
\begin{definition}
A \emph{braided $L_\infty$-algebra} is an $L_\infty$-algebra
$(V,\{\ell_n\})$ in the
symmetric monoidal category ${}_\CF\CCM^\sharp$.
\end{definition}

This definition simply amounts to letting the symmetric
group act in all relations of the $L_\infty$-algebra, written in terms
of morphisms, through the
braiding isomorphism
${}_{\textrm{\tiny$\CF$}}^{\phantom{\dag}}\tau^\sharp$ of
${}_\CF\CCM^\sharp$; we write $_{\textrm{\tiny$\CF$}} \sigma$ to distinguish the action of a permutation $\sigma$ via the braided transposition from the usual one. Explicitly, as in the case of Lie algebras, this
gives the braided antisymmetry of the morphisms $\ell_n:
\text{\Large$\otimes$}^{n}V\to V$ in ${}_\CF\CCM^\sharp$:
\begin{align}\label{TwistedBracket}
\ell_n (\dots, v,v',\dots) = -(-1)^{|v|\,|v'|}\, \ell_n \big(\dots,
  \sfR_k(v'),\sfR^k(v),\dots\big) \ ,
\end{align}
together with \emph{braided homotopy relations}
$\CJ_n(v_1,\dots,v_n)=0$ in ${}_\CF\CCM^\sharp$ which follow from
\eqref{eq:Jnmaps}. The first two relations $\CJ_1(v)=0$ and
$\CJ_2(v_1,v_2)=0$ are unchanged from the classical case of Section~\ref{sec:Linfty}:
\begin{align}
  \ell_1\big(\ell_1(v)\big) &= 0 \ , \nonumber \\[4pt]
  \ell_1\big(\ell_2(v_1,v_2)\big) &= \ell_2\big(\ell_1(v_1),v_2\big) +
  (-1)^{|v_1|}\, \ell_2\big(v_1, \ell_1(v_2)\big) \ , \label{eq:l1l2braided}
\end{align}
hence the map $\ell_1:V\to V$ is still a differential which is a
graded (but not braided!) derivation of the $2$-bracket
$\ell_2:V\otimes V\to V$. The first homotopy relation which differs
from the classical case is
$\CJ_3(v_1,v_2,v_3)=0$, which is a deformation of the classical
identity \eqref{I3} in accordance with the non-trivial braiding
\eqref{TwistedBracket}: 
\begin{align}
& \ell_2\big(\ell_2(v_1,v_2),v_3\big) - (-1)^{|v_2|\,|v_3|}\,
  \ell_2\big(\ell_2(v_1,\sfR_k(v_3)),\sfR^k(v_2)\big) \nn \\ &
                                                               \hspace{2cm} +  (-1)^{(|v_2|+|v_3|)\,|v_1|}\,
  \ell_2\big(\ell_2(\sfR_k(v_2),\sfR_l(v_3)),\sfR^l\sfR^k(v_1)\big) \nn\\[4pt]
& \hspace{3cm} = -\ell_3\big(\ell_1(v_1),v_2,v_3\big) - (-1)^{|v_1|}\,
  \ell_3\big(v_1, \ell_1(v_2), v_3\big) - (-1)^{|v_1|+|v_2|}\,
  \ell_3\big(v_1,v_2, \ell_1(v_3)\big) \nn\\
& \hspace{5cm} -\ell_1\big(\ell_3(v_1,v_2,v_3)\big) \ .
\label{I3braided}\end{align}                                             

\subsection{Drinfel'd twist deformations of $L_\infty$-algebras}
\label{sec:Linftytwist}

A large class of examples of braided $L_\infty$-algebras can be obtained from Drinfel'd twist
deformation quantization, starting from suitable $L_\infty$-algebras
in the classical sense. Let $(V,\{\ell_n\})$ be an $L_\infty$-algebra
in the category $\CCM^\sharp$ of $U\frv$-modules. Concretely, this
means that $V=\bigoplus_{k\in\RZ}\,V_k$ is a $\RZ$-graded (left)
$U\frv$-module and the $n$-brackets $\ell_n:\midoplus^nV\to V$ are
equivariant maps, that is, they all commute with the action of
$\frv=\Gamma(TM)$ on $V$ via the trivial coproduct $\Delta$. Given any 
Drinfel'd twist $\CF\in U\frv[[\hbar]]\otimes U\frv[[\hbar]]$, we can
deform the brackets $\ell_n$ to twisted brackets $\ell_n^\star$ which
are morphisms in the twisted representation category
${}_\CF\CCM^\sharp$ of $\RZ$-graded $U_\CF\frv$-modules, that is, they
commute with the action of $\Gamma(TM)$ on $V[[\hbar]]$ via the
twisted coproduct $\Delta_\CF$. Following the
standard prescription (\ref{eq:mustar}), we set
$\ell_1^\star:=\ell_1$ and
\begin{align}\label{eq:ellnstardef}
	\ell_n^\star(v_1\otimes\cdots\otimes v_n) :=
	\ell_n(v_1\otimes_\star\cdots\otimes_\star v_n)
\end{align}
for $n\geq2$, where $v\otimes_\star v':=\CF^{-1}(v\otimes
v')=\bar\sff^k(v)\otimes\bar\sff_k(v')$ for $v,v'\in V$.

\begin{proposition}\label{prop:braidedfromclassical}
	If $(V,\{\ell_n\})$ is an $L_\infty$-algebra in the category
	$\CCM^\sharp$, then $(V[[\hbar]],\{\ell_n^\star\})$ is a braided
	$L_\infty$-algebra. 
\end{proposition}

\begin{proof}
	Braided antisymmetry of the $n$-brackets $\ell_n^\star$ follows by the
	definition \eqref{eq:ellnstardef}, and the only thing we need to check
	is that the braided homotopy relations for $\ell_n^\star$ follow from
	the (undeformed) homotopy relations for $\ell_n$. The proof follows straightforwardly from well-known abstract
	arguments, see e.g.~\cite{Barnes:2014ksa, Barnes:2015uxa}, and should be obvious to experts, but we spell out the details here for the benefit of the non-expert reader. 
	
	Drinfel'd twist deformation quantization produces a
	symmetric monoidal functor $\CCF:(\CCM^\sharp,\otimes) \to{}(_\CF\CCM^\sharp,\otimes_{\CF})$
	which is an equivalence, where $\otimes_{\CF}$ denotes the tensor product of $U_\CF\frv$-modules acted upon using the twisted coproduct $\Delta_{\CF}$. The functor $\CCF$ acts on objects and morphisms as the
	identity,\footnote{To be precise, we need to precompose $\CCF$ with the
		symmetric monoidal functor which sends an object $V$ in 
		$\CCM^\sharp$ to the object $V[[\hbar]]$ in the symmetric monoidal
		category of (left) $U\frv[[\hbar]]$-modules over $\FC[[\hbar]]$,
		see~\cite{Barnes:2014ksa} for details. We suppress this extra step
		here for clarity of exposition and because we will not need it in our
		applications later on.} with the coherence maps
	$\varphi:\CCF(V)\otimes_{\CF} \CCF(W) \to \CCF(V\otimes W)$ given by 
	\begin{align*}
		\varphi(v\otimes w) = \bar\sff^k(v)\otimes\bar\sff_k(w) =
		v\otimes_\star w \ ,
	\end{align*}
	where $V,W$ are $U\frv$-modules with $v\in V$ and $w\in W$ (the inverse is given by acting
	with the twist $\CF=\sff^k\otimes\sff_k$). These fit into the
	commutative diagram
	\begin{align}\label{diag:braiding}
		\xymatrix{
			\CCF(V)\otimes_\CF \CCF(W)
			\ar[rr]^{{}_{\textrm{\tiny$\CF$}}^{\phantom{\dag}}\tau^\sharp}
			\ar[d]_\varphi & &
			\CCF(W)\otimes_\CF
			\CCF(V) \ar[d]^\varphi \\
			\CCF(V\otimes W) \ar[rr]_{\CCF(\tau^\sharp)} & & \CCF(W\otimes V)
		}
	\end{align}
	in ${}_\CF\CCM^\sharp$. The coherence maps are iterated to arbitrary $n$-fold tensor products of $ U\frv$-modules $V_i$ to get maps $\varphi_n:\CCF(V_1)\otimes_{\CF} \cdots\otimes_\CF\CCF(V_n) \to \CCF(V_1\otimes\cdots\otimes V_n)$ for $n\geq1$ defined by
\begin{align*}
\varphi_n(v_{1}\otimes \cdots \otimes v_{n}) = v_1\otimes_\star \cdots \otimes_\star v_{n} \ ,
\end{align*}
for $v_i\in V_i$, with $\varphi_{2} = \varphi$ and $ \varphi_{1} = \id $.
	
	The morphisms \eqref{eq:ellnstardef} are thus
	defined by precomposing $\CCF(\ell_n)$ with the coherence
	maps $\varphi_n$. We will show that applying this to the morphisms \eqref{eq:Jnmaps} in
	$\CCM^\sharp$ gives morphisms 
	$\CJ_n^\star=\CCF(\CJ_n)\circ\varphi_n$ in ${}_\CF\CCM^\sharp$, where
	\begin{align}\label{eq:Jnmapsstar}
		{\cal J}^\star_n := \sum^n_{i=1}\, (-1)^{i\,(n-i)} \
		\sum_{\sigma\in{\rm Sh}_{i,n-i}} \, \text{sgn}(\sigma) 
		\ \ell^\star_{n+1-i}\circ \big(
		\ell^\star_i\otimes \id^{\otimes{n-i}}\big) \circ {}_{\textrm{\tiny$\CF$}}\sigma
		\ .
	\end{align}
	Since the functor $\CCF$ preserves the braidings by the diagram
	\eqref{diag:braiding}, these are just the braided homotopy
	relations, i.e. $\CJ_n(v_1\otimes\cdots\otimes v_n)=0$ in $\CCM^\sharp$ implies
\begin{align*}
\CJ_n^\star(v_1\otimes\cdots\otimes v_n)=\CJ_n(v_1\otimes_\star\cdots\otimes_\star v_n) =0
\end{align*}
in~${}_\CF\CCM^\sharp$. 
	
	Let us now demonstrate that $\CJ_{n}^\star = \CJ_{n}\circ \varphi_n$ explicitly. Note that this is trivially satisfied for $n=1$, as $\ell_1^\star=\ell_1$ and $\varphi_1=\id$. We will show that
	\begin{align*}
	\ell^\star_{n+1-i}\circ \big(
	\ell^\star_i\otimes \id^{\otimes{n-i}}\big) \circ {}_{\textrm{\tiny$\CF$}}\sigma=\ell_{n+1-i} \circ  \big(\ell_i\otimes \id^{\otimes{n-i}}\big) \circ \,  \sigma \circ \varphi_{n} 
	\end{align*}
for each $n\geq2$, $i=1,\dots,n$ and $\sigma\in{\rm Sh}_{i,n-i}$.  A simple expansion shows
	\begin{align*}
		\ell^\star_{n+1-i}\circ \big(
		\ell^\star_i\otimes \id^{\otimes{n-i}}\big)&= \ell_{n+1-i}\circ \varphi_{n+1-i} \circ \big( (\ell_{i}\circ \varphi_{i})\otimes \id^{\otimes n-i}\big)\\[4pt]
		&= \ell_{n+1-i} \circ \varphi_{n+1-i}\,  \circ (\ell_{i}\otimes \id^{\otimes n-i}) \circ (\varphi_{i}\otimes \id^{\otimes n-i})\\[4pt]
		&= \ell_{n+1-i} \circ (\ell_{i}\otimes \id^{\otimes n-i})\circ  \varphi^{i,n-i}\circ(\id^{\otimes i} \otimes \varphi_{ n-i})\circ (\varphi_{i} \otimes \id^{\otimes n-i}) \ ,
	\end{align*}
	where $\varphi^{i,n-i}(v_{1}\otimes\cdots\otimes v_{n}):=(v_{1}\otimes \cdots \otimes v_{i})\otimes_\star (v_{i+1}\otimes \dots \otimes v_{n}),$ and the third equality follows from 
	\begin{align*}
		\varphi_{n+1-i}\,\circ (\ell_{i}\otimes \id^{\otimes n-i}) (v_{1}\otimes \cdots \otimes v_{n})&= \bar{\sff}^k \big(\ell_{i}(v_{1}\otimes \cdots \otimes v_{i})\big) \otimes \bar{\sff}_k(v_{i+1}\otimes_\star \cdots \otimes_\star v_{n}) \\[4pt]
		&= \ell_{i}\big( \bar{\sff}^k(v_{1}\otimes \cdots \otimes v_{i})\big) \otimes \bar{\sff}_{k}(v_{i+1}\otimes_\star \cdots \otimes_\star v_{n}) \\[4pt]
		&= (\ell_{i}\otimes \id^{\otimes n-i})\circ \varphi^{i,n-i}\circ (\id^{\otimes i}\otimes \varphi_{n-i})(v_{1}\otimes \cdots \otimes v_{n}) \ ,
	\end{align*} 
	where $\ell_{i}$ commutes with $\bar{\sff}^{k}$ since it is a morphism in $\CCM^{\sharp}$.
Since $v_{1} \otimes_\star \cdots \otimes_\star v_{n} = (v_{1}\otimes_\star \cdots \otimes_\star v_{i})\otimes_{\star}(v_{i+1}\otimes_\star \cdots \otimes_\star v_{n})$, the map $\varphi_{n}$ may be equivalently written as 
	\begin{align*}
	\varphi_{n} &= \varphi^{i,n-i}\circ(\id^{\otimes i} \otimes \varphi_{ n-i})\circ (\varphi_{i} \otimes \id^{\otimes n-i}) \ ,
	\end{align*}
and hence
	\begin{align*}
	\ell^\star_{n+1-i}\circ \big(
	\ell^\star_i\otimes \id^{\otimes{n-i}}\big)= \ell_{n+1-i} \circ (\ell_{i}\otimes \id^{\otimes n-i})\circ  \varphi_{n} \ .  
\end{align*}

Next we show that $\varphi_n$ commutes with any transposition via the braiding isomorphism $_{\textrm{\tiny$\CF$}}\tau^\sharp$. That is, for any $n\geq2$ and $j=0,1,\dots,n-2$ one has
\begin{align*}
\varphi_{n}\circ (\id^{\otimes j}\otimes {}_{\textrm{\tiny$\CF$}}\tau^\sharp \otimes \id^{\otimes n-j-2} ) = (\id^{\otimes j}\otimes \tau^\sharp \otimes \id^{\otimes n-j-2}) \circ \varphi_{n} \ .
\end{align*}
Since $v_1\otimes_\star \cdots \otimes_\star v_n=(v_1\otimes_\star \cdots \otimes_\star v_j)\otimes_\star \big( (v_{j+1}\otimes_\star v_{j+2})\otimes_\star (v_{j+3}\otimes_\star \cdots \otimes_\star v_{n})\big)$, the map $\varphi_{n}$ may be also written as 
$$
\varphi_{n}= \varphi^{j,n-j}\circ (\id^{\otimes j}\otimes \varphi^{2,n-j-2})\circ (\varphi_{j}\otimes \id^{\otimes 2}\otimes \varphi_{n-j-2})\circ (\id^{\otimes j}\otimes \varphi\otimes \id^{\otimes n-j-2}) \ .
$$
Since $_{\textrm{\tiny$\CF$}}\tau^\sharp = \varphi^{-1} \circ \tau^\sharp \circ \varphi$ by \eqref{diag:braiding}, we have 
\begin{align*}
\varphi_{n}\circ (\id^{\otimes j}\otimes {}_{\textrm{\tiny$\CF$}}\tau^\sharp \otimes \id^{\otimes n-j-2}) &=  \varphi^{j,n-j}\circ (\id^{\otimes j}\otimes \varphi^{2,n-j-2})\circ (\varphi_{j}\otimes \id^{\otimes 2}\otimes \varphi_{n-j-2}) \\
& \quad \, \circ (\id^{\otimes j} \otimes \tau^\sharp \otimes \id^{\otimes n-j-2})\circ (\id^{\otimes j}\otimes \varphi\otimes \id^{\otimes n-j-2})\\[4pt]
&= \varphi^{j,n-j}\circ (\id^{\otimes j}\otimes \varphi^{2,n-j-2})\circ (\id^{\otimes j} \otimes \tau^\sharp \otimes \id^{\otimes n-j-2}) \\
& \quad \, \circ (\varphi_{j}\otimes \id^{\otimes 2}\otimes \varphi_{n-j-2}) \circ (\id^{\otimes j}\otimes \varphi\otimes \id^{\otimes n-j-2}) \ .
\end{align*}
The map involving the transposition commutes with the remaining maps since 
\begin{align*}
& (\id^{\otimes j}\otimes \varphi^{2,n-j-2})\circ (\id^{\otimes j} \otimes \tau^\sharp \otimes \id^{\otimes n-j-2})(v_{1}\otimes\cdots \otimes v_{n})\\[4pt]
& \hspace{5cm} = v_{1}\otimes \cdots\otimes  v_{j} \otimes \big(\Delta(\bar{\sff}^k)\circ \tau^\sharp\big) (v_{j+1}\otimes v_{j+2})\otimes \bar{\sff}_{k}(v_{j+3}\otimes \cdots \otimes v_{n}) \\[4pt]
 & \hspace{5cm} = v_{1}\otimes \cdots \otimes v_{j} \otimes \big(\tau^\sharp \circ \Delta(\bar{\sff}^k)\big) (v_{j+1}\otimes v_{j+2}) \otimes \bar{\sff}_k(v_{j+3}\otimes \cdots \otimes v_{n})\\[4pt]
 & \hspace{5cm} =(\id^{\otimes j} \otimes \tau^\sharp \otimes \id^{\otimes n-j-2})\circ (\id^{\otimes j}\otimes \varphi^{2,n-j-2})(v_{1}\otimes\cdots \otimes v_{n}) \ ,
\end{align*}
where the second equality follows since the trivial coproduct $\Delta$ is cocommutative, $\Delta^{\rm op}= \Delta$, and so $\Delta \circ \tau^{\sharp} = \tau^{\sharp}\circ \Delta^{\rm op}=\tau^{\sharp}\circ \Delta$. Similarly one has
\begin{align*}
& \varphi^{j,n-j}\circ(\id^{\otimes j} \otimes \tau^\sharp \otimes \id^{\otimes n-j-2})(v_{1}\otimes\cdots \otimes v_{n})\\[4pt]
& \hspace{3cm} =\bar{\sff}^k (v_{1}\otimes\cdots \otimes v_{j})\otimes \bar{\sff}_k\big(\tau^{\sharp}(v_{j+1}\otimes v_{j+2})\otimes v_{j+3}\otimes \cdots \otimes v_{n}\big) \\[4pt]
& \hspace{3cm} =\bar{\sff}^k(v_{1}\otimes\cdots \otimes v_{j})\otimes \bar{\sff}_{k\swone}\big(\tau^{\sharp}(v_{j+1}\otimes v_{j+2})\big) \otimes \bar{\sff}_{k\swtwo}(v_{j+3}\otimes \cdots \otimes v_{n})\\[4pt]
& \hspace{3cm} =\bar{\sff}^k(v_{1}\otimes\cdots \otimes v_{j})\otimes \big(\Delta(\bar{\sff}_{k\swone})\circ \tau^{\sharp}\big)(v_{j+1}\otimes v_{j+2}) \otimes \bar{\sff}_{k\swtwo}(v_{j+3}\otimes \cdots \otimes v_{n})\\[4pt]
& \hspace{3cm} =\bar{\sff}^k(v_{1}\otimes\cdots \otimes v_{j})\otimes\big( \tau^\sharp \circ \Delta(\bar{\sff}_{k\swone})\big)(v_{j+1}\otimes v_{j+2}) \otimes \bar{\sff}_{k\swtwo}(v_{j+3}\otimes \cdots \otimes v_{n})\\[4pt]
& \hspace{3cm} =(\id^{\otimes j} \otimes \tau^\sharp \otimes \id^{\otimes n-j-2})\circ \varphi^{j,n-j} (v_{1}\otimes\cdots \otimes v_{n}) \ .
\end{align*}

Collecting everything together, we have shown that
\begin{align*}
& \ell^\star_{n+1-i}\circ \big(
	\ell^\star_i\otimes \id^{\otimes{n-i}}\big) \circ \, (\id^{\otimes j}\otimes {}_{\textrm{\tiny$\CF$}}\tau^\sharp \otimes \id^{\otimes n-j-2} ) \\[4pt]
&	\hspace{6cm} =\ell_{n+1-i} \circ  \big(\ell_i\otimes \id^{\otimes{n-i}}\big) \circ \,  (\id^{\otimes j}\otimes \tau^\sharp \otimes \id^{\otimes n-j-2} ) \circ \varphi_{n} 
\end{align*}
 for all $n\geq 2$, $i=1,\dots,n$ and $j=0,1,\dots,n-2$. Since any permutation can be written as a composition of transpositions, the result follows immediately. 
\end{proof}

\begin{example}
If $\ell_n=0$ for all $n\geq3$, then Proposition~\ref{prop:braidedfromclassical} is just the statement that a differential graded Lie algebra twist quantizes to a differential (graded) braided Lie algebra. This is precisely what we saw in Section~\ref{sec:twistedforms} for the special case of exterior differential forms valued in a Lie algebra.
\end{example}

\subsection{Cyclic structures}
\label{sec:cyclictwist}

So far we have not discussed cyclic structures on our braided
$L_\infty$-infinity algebras, which as we saw in
Section~\ref{sec:Linftygft} is a crucial ingredient for a classical
field theory to have a Lagrangian formulation in the
$L_\infty$-algebra
framework. Let $(V,\{\ell_n\},\langle-,-\rangle)$ be a cyclic
$L_\infty$-algebra in the category $\CCM^\sharp$ of $\RZ$-graded
$U\frv$-modules. This means that the cyclic pairing
$\langle-,-\rangle:V\otimes V\to\FR$ has to additionally be
$U\frv$-invariant:
\begin{align}\label{eq:paringinv}
\langle X_\swone(v_1),X_\swtwo(v_2)\rangle=0
\end{align}
for all $X\in U\frv$ and $v_1,v_2\in V$; for vector fields
$\xi\in\frv=\Gamma(TM)$ this reads
$\langle\xi(v_1),v_2\rangle=-\langle v_1,\xi(v_2)\rangle$, which is
the natural requirement of diffeomorphism invariance of the pairing.

Following our standard prescription \eqref{eq:mustar}, we can twist
deform the pairing $\langle-,-\rangle$ to a new pairing
$\langle-,-\rangle_\star$ defined by
\begin{align}\label{eq:twistpairing}
\langle v_1,v_2\rangle_\star := \langle
  \bar\sff^k(v_1),\bar\sff_k(v_2) \rangle \ .
\end{align}
Non-degeneracy of the pairing $\langle-,-\rangle_\star:V[[\hbar]]\otimes V[[\hbar]]\to \FR[[\hbar]]$ follows from the  untwisted non-degeneracy and an order by order argument in the formal deformation parameter $\hbar$. In general, graded symmetry of the cyclic pairing $\langle-,-\rangle$
implies that the twisted pairing $\langle-,-\rangle_\star$ is naturally
(graded) braided symmetric, and the natural notion of cyclicity for a
braided $L_\infty$-algebra would of course be that of `braided
cyclicity' defined by inserting suitable factors of the $\RR$-matrix
in the obvious way. However, this level of generality is not 
suitable for applications to field theory, as the loss of {\it strict} graded
symmetry and cyclicity would lead to problems with the variational principle for the corresponding action functionals. We therefore
restrict from the outset the types of Drinfel'd twists that we will
use for deformation quantization of a given cyclic $L_\infty$-algebra.

\begin{definition}\label{def:compatibletwist}
A Drinfel'd twist $\CF\in U\frv[[\hbar]]\otimes U\frv[[\hbar]]$ is
\emph{compatible} with a cyclic structure $\langle-,-\rangle:V\otimes
V\to\FR$ on an $L_\infty$-algebra in $\CCM^\sharp$ if
\begin{align*}
\langle \sfR_k(v_1),\sfR^k(v_2)\rangle_\star = \langle
  v_1,v_2\rangle_\star 
\end{align*}
for all $v_1,v_2\in V$.
\end{definition}

Given the concrete examples of cyclic $L_\infty$-algebra structures
that we considered in Sections~\ref{sec:CStheory} and~\ref{sec:ECP3d},
it should come as no surprise that this definition is an abstraction
of the graded cyclicity property \eqref{eq:intcyclic} for integration
of twisted differential forms which we already discussed in
Section~\ref{sec:twistedforms}, and with it we have

\begin{proposition}\label{prop:compatiblestrict}
Let $(V,\{\ell_n\},\langle-,-\rangle)$ be a cyclic $L_\infty$-algebra
in the category $\CCM^\sharp$, and let $\CF$ be a compatible Drinfel'd
twist. Then
$(V[[\hbar]],\{\ell_n^\star\},\langle-,-\rangle_\star)$ is a strictly
cyclic braided $L_\infty$-algebra.
\end{proposition}

\begin{proof}
For clarity of exposition, in the following we drop the signs arising
from the gradings in the cyclicity conditions of
Section~\ref{sec:Linfty}, as they play no role in the proof since they are the same for the original and the twist deformed $L_\infty$-algebra. We then need to show
that
\begin{align*}
\langle v_0,\ell_n^\star(v_1,v_2,\dots,v_n)\rangle_\star =
  \langle v_n,\ell_n^\star(v_0,v_1,\dots,v_{n-1})\rangle_\star \ .
\end{align*}

For $n=1$, since $\ell_1^\star=\ell_1$ is not deformed, the left-hand side reads
\begin{align*}
\langle v_0, \ell_{1}(v_1)\rangle_{\star}
  &=\langle\bar{\sff}^k (v_0), \ell_{1}(\bar{\sff}_k( v_1))\rangle
  \\[4pt]
  &=\langle\bar{\sff}_k( v_1), \bar{\sff}^k(\ell_{1}( v_0))\rangle \nn \\[4pt]
&=\langle\bar{\sff}^k \sfR_l (v_1), \bar{\sff}_k \sfR^l (\ell_{1}(v_0
                                                                                 ))\rangle
  \\[4pt]
  &=\langle\sfR_l (v_1), \sfR^l (\ell_1 (v_0)) \rangle_\star \nn \\[4pt]
&=\langle v_1, \ell_{1}(v_0)\rangle_\star
\end{align*}
where in the first and fourth equalities we use the definition of the
star-pairing, in the first and second equalities we use commutativity
of $\ell_1$ with the $U\frv$-action, in the second equality we use
classical cyclicity, in the third equality the definition of the
$\RR$-matrix, and in the fifth equality the compatibility of the
twist.

For $n\geq2$, this calculation extends to
\begin{align*}
\langle v_0,\ell_n(v_1\otimes_\star\cdots\otimes_\star v_n)\rangle_\star &= \big\langle
                                                              \bar\sff^k(v_0),\bar\sff_k\big(\ell_n(\bar\sff^l(v_1\otimes_\star\cdots\otimes_\star
                                                              v_{n-1})\otimes\bar\sff_l(v_n))\big)\big\rangle
  \\[4pt]
  &=
    \big\langle\bar\sff^k(v_0),\ell_n\big(\bar\sff_k{}_\swone\bar\sff^l(v_1\otimes_\star\cdots\otimes_\star
    v_{n-1})\otimes\bar\sff_k{}_\swtwo\bar\sff_l(v_n)\big)\big\rangle
  \\[4pt]
  &=
    \big\langle\bar\sff^k_\swone\bar\sff^l(v_0),\ell_n\big(\bar\sff^k_\swtwo\bar\sff_l(v_1\otimes_\star\cdots\otimes_\star
    v_{n-1})\otimes\bar\sff_k(v_n)\big)\big\rangle \\[4pt]
  &=
    \big\langle\bar\sff_k(v_n),\ell_n\big(\bar\sff^k_\swone\bar\sff^l(v_0)\otimes\bar\sff^k_\swtwo\bar\sff_l(v_1\otimes_\star\cdots\otimes_\star
    v_{n-1})\big)\big\rangle \\[4pt]
  &=
    \big\langle\bar\sff_k(v_n),\bar\sff^k\big(\ell_n(\bar\sff^l(v_0)\otimes\bar\sff_l(v_1\otimes_\star\cdots\otimes_\star
    v_{n-1}))\big)\big\rangle \\[4pt]
  &=
    \big\langle\bar\sff^k\sfR_m(v_n),\bar\sff_k\sfR^m\big(\ell_n(\bar\sff^l(v_0)\otimes\bar\sff_l(v_1\otimes_\star\cdots\otimes_\star
    v_{n-1}))\big)\big\rangle \\[4pt]
  &= \big\langle\sfR_m(v_n),\sfR^m\big(\ell_n(v_0\otimes_\star
    v_1\otimes_\star\cdots\otimes_\star v_{n-1})\big)\big\rangle_\star \\[4pt]
  &= \langle v_n,\ell_n(v_0\otimes_\star
    v_1\otimes_\star\cdots\otimes_\star v_{n-1})\rangle_\star
\end{align*}
where we used the same steps as in the $n=1$ case above, and now also
the cocycle property \eqref{eq:cocyclesw} for the inverse of the
Drinfel'd twist and the $U\frv$-invariance \eqref{eq:paringinv} of
the classical pairing in the third equality.
\end{proof}

\section{Braided field theory}
\label{sec:bft}

\subsection{Noncommutative field theories in the braided $L_\infty$-algebra formalism}
\label{sec:ncftLinfty}

We can now build up large classes of noncommutative field theories by
employing a braided version of the `bottoms-up' approach to classical
field theories which we discussed in
Section~\ref{sec:Linftygft}, whereby a classical field theory is defined as a dynamical system by its
4-term $L_\infty$-algebra as the initial input. These can be constructed starting 
from any suitable braided $L_\infty$-algebra in the sense of
Section~\ref{sec:Linftycat}; we call a field theory obtained in this
way a \emph{braided field theory}. That is, in this paper, a braided field theory is
defined as a (formal) dynamical system by its 4-term braided $L_\infty$-algebra as initial input.\footnote{Our formalism may however also be used to treat more general field theories with reducible gauge symmetries, which are captured by $k$-term braided $L_\infty$-algebras with $k>4$.}

Usually these theories will come from
deforming a classical field theory on a manifold $M$ through Drinfel'd
twist deformation quantization of
its corresponding $L_\infty$-algebra, as described in
Section~\ref{sec:Linftytwist}, and indeed this will be the case for
the examples treated in this paper. If $(V,\{\ell_n\})$ is a 4-term
$L_\infty$-algebra with underlying graded vector space
\begin{align*}
V=V_0\oplus V_1\oplus V_2\oplus V_3
\end{align*}
in the representation category $\CCM^\sharp$, then, following the
prescription of Section~\ref{sec:Linftytwist}, we construct the
corresponding 4-term braided $L_\infty$-algebra
$(V[[\hbar]],\{\ell_n^\star\})$ in the twisted representation category
${}_\CF\CCM^\sharp$. 
In simple
instances, a braided field theory obtained in this way through
Drinfel'd twist deformation of a classical Lagrangian field theory will be the
``naive'' deformation, obtained by twisting products of fields in the equations of motion and in
the action functional to star-products in the standard way. However,
we shall see that this is not always the case, a notable example being
the four-dimensional braided ECP theory of gravity that we define
later on. But even in simple cases the notion of braided gauge
symmetry brings with it several novelties and complexities that we
explain in detail below.

We shall now run through the various ingredients
involved in the definition of a noncommutative field theory in this
way for a general 4-term braided $L_\infty$-algebra $(V,\{\ell_n\})$, following the
standard classical prescription of Section~\ref{sec:Linftygft}.

\subsection{Braided Lie algebras of gauge symmetries}
\label{sec:braidedgauge}

Let us first discuss what should generally be meant by a `gauge symmetry' in this setting, and then later pass to its realization in the $L_\infty$-algebra formalism.

\subsubsection*{Braided gauge transformations}

In the context of generalized gauge symmetries, wherein one considers general non-linear field-dependent transformations, a \emph{braided gauge transformation} is a map $V_0\times V_1\to V_1$, denoted as $(\lambda,A)\mapsto \delta_\lambda^\lact A$, which is linear in the gauge parameters $\lambda\in V_0$ and has a potentially arbitrary polynomial dependence on fields $A\in V_1$. This means that the map may be realized as the diagonal of a map $V_0\otimes  \text{\Large$\odot$}^\bullet V_{1}\to V_1$ as an arrow in the category \smash{${}_\CF\CCM^\sharp$}, where the symmetric tensor product is taken with respect to the braided transposition ${}_{\textrm{\tiny$\CF$}}^{\phantom{\dag}}\tau^\sharp$. Explicitly, a braided gauge transformation $\delta_\lambda^\lact A$ may be expanded in tensor powers of the fields $A^{\otimes n}$ as
\begin{align*}
\delta_\lambda^\lact A = \delta_0^\lact(\lambda) + \sum_{n=1}^\infty \, \delta_{\lambda,n}^\lact(A^{\otimes n}) \ ,
\end{align*}
where each polynomial order is now a linear braided symmetric map on $V_1$ in \smash{${}_\CF\CCM^\sharp$}. Similarly to the classical case, the collection of transformations on fields $V_1$ are required to close (in a suitable sense to be made precise later on) under the braided commutator bracket
\begin{align*}
\big[\delta_{\lambda_1}^{\lact}, \delta_{\lambda_2}^{\lact}\big]_\circ^\star :=  \delta_{\lambda_1}^{\lact}\circ \delta_{\lambda_2}^{\lact} - \delta_{\sfR_k(\lambda_2)}^{\lact}\circ \delta_{\sfR^k(\lambda_1)}^{\lact} \ ,
\end{align*}
for $\lambda_1,\lambda_2\in V_0$. 
The braided commutator is a braided Lie bracket, that is, it is braided antisymmetric and satisfies the braided Jacobi identity; this is a special instance of the braided commutator of endomorphisms in the category \smash{${}_\CF\CCM^\sharp$} which makes the braided derivations of a $U_\CF\frv$-module algebra into a braided Lie algebra~\cite{Barnes:2015uxa}. The simplest prototypical case, and in fact the case most relevant for the examples of this paper, is that of a linear braided Lie algebra action, which we discuss in Example~\ref{ex:protoex} below.

By setting $\delta^\lact_\lambda\big|_{V_0}=0$, a braided gauge transformation is extended to a map $\delta_\lambda^\lact:V\to V$ of degree~$0$ as a braided derivation of operations defined on $V\otimes V$. 
It is convenient to formalise the braided derivation property by stating that the action of a braided gauge transformation $\delta^\lact_\lambda$ on a
tensor product is through the twisted coproduct
$$
\triangle_\CF(\delta^\lact_\lambda) = \delta^\lact_\lambda\otimes\id +
\sfR_k\otimes\delta^\lact_{\sfR^k(\lambda)}  \ .
$$
This is also a map $\triangle_\CF(\delta^\lact_\lambda):V\otimes V\to V\otimes V$ of degree~$0$, which encodes the braided Leibniz rule. This map is extended inductively to higher tensor products $\triangle_\CF^n(\delta^\lact_\lambda):\midoplus^nV\to\midoplus^nV$ for $n\geq3$ by
\begin{align}\label{eq:braidedLeibnizleft}
\triangle_\CF^n(\delta^\lact_\lambda) = \delta^\lact_\lambda\otimes\id^{\otimes n-1} + \sum_{i=1}^{n-1}\, \sfR_{k_1}\otimes\cdots\otimes\sfR_{k_i}\otimes\delta^\lact_{\sfR^{k_i}\cdots\sfR^{k_1}(\lambda)}\otimes\id^{\otimes n-i-1} \ ,
\end{align}
and we set $\triangle_\CF^1(\delta_\lambda^{\lact}):=\delta_\lambda^{\lact}$ and $\triangle_\CF^2(\delta_\lambda^\lact):=\triangle_\CF(\delta_\lambda^\lact)$.
Thus if $\mu_n$ is any map on $\midoplus^nV$ in the category ${}_\CF\CCM^\sharp$, we define its braided gauge transformation by
\begin{align}\label{eq:deltaLmudef}
\delta_\lambda^\lact\mu_n := \mu_n\circ\triangle_\CF^n(\delta_\lambda^\lact) \ .
\end{align}
For example, with $n=2$ one has
\begin{align}\label{eq:deltaLmuexpl}
\delta_\lambda^\lact\mu_2(v_1\otimes v_2) = \mu_2(\delta_\lambda^\lact v_1\otimes v_2) + \mu_2\big(\sfR_k(v_1)\otimes \delta^\lact_{\sfR^k(\lambda)}v_2\big)
\end{align}
for all $v_1,v_2\in V$. 

The superscript ${}^\lact$ is used to indicate that this map is regarded as a \emph{left} braided gauge  transformation which ``acts from the left''. Correspondingly, there is a \emph{right} braided gauge transformation $V_1\times V_0\to V_1$, $(A,\lambda)\mapsto \delta_\lambda^\ract A$ which is regarded as ``acting from the right'' and is defined by
\begin{align*}
\delta_\lambda^\ract A:= \delta^\lact_{\sfR_k(\lambda)}\sfR^k(A) = \delta_0^\lact(\lambda) + \sum_{n=1}^\infty \, \delta_{\sfR_k(\lambda),n}^\lact\big(\sfR^k(A^{\otimes n})\big) \ .
\end{align*}
In the classical case, where $\CR=1\otimes1$, the two transformations of course coincide. 
The action of a right braided gauge transformation $\delta_\lambda^\ract$ on tensor products is through the twisted coproduct
\begin{align*}
\triangle_\CF(\delta^\ract_\lambda) = \id\otimes\delta^\ract_\lambda+
\delta^\ract_{\sfR_k(\lambda)}\otimes\sfR^k \ ,
\end{align*}
encoding the braided Leibniz rule from the right, and its higher iterations
\begin{align}\label{eq:braidedLeibnizright}
\triangle_\CF^n(\delta^\ract_\lambda) = \id^{\otimes n-1}\otimes\delta^\ract_\lambda+ \sum_{i=1}^{n-1}\, \id^{\otimes n-i-1} \otimes \delta^\ract_{\sfR_{k_i}\cdots\sfR_{k_1}(\lambda)}\otimes\sfR^{k_1}\otimes\cdots\otimes\sfR^{k_i} \ .
\end{align}

Suppose now that the map $\mu_2$ is braided symmetric, that is, the maps
\begin{align*}
\mu_2 = \mu_2\circ{}_{\textrm{\tiny$\CF$}}^{\phantom{\dag}}\tau^\sharp_{V,V}
\end{align*}
are equal as morphisms in the category ${}_\CF\CCM^\sharp$, where ${}_{\textrm{\tiny$\CF$}}^{\phantom{\dag}}\tau^\sharp_{V,V}:V\otimes V\to V\otimes V$ is the braiding given by \smash{${}_{\textrm{\tiny$\CF$}}^{\phantom{\dag}}\tau^\sharp:=\RR^{-1}\circ\tau^\sharp=\tau^\sharp\circ\RR$} with $\tau^\sharp$ the graded transposition. On homogeneous elements this reads
\begin{align}\label{eq:musym}
\mu_2(v_1\otimes v_2)=(-1)^{|v_1|\,|v_2|}\,\mu_2\big(\sfR_k(v_2)\otimes
  \sfR^k(v_1)\big) \ .
\end{align}
We check that the definition \eqref{eq:deltaLmudef} is compatible with this braided symmetry, that is, that the left braided gauge transformation of the right-hand side of \eqref{eq:musym} is equal to \eqref{eq:deltaLmuexpl}:
\begin{align*}
\delta_\lambda^\lact\big(\mu_2\circ{}_{\textrm{\tiny$\CF$}}^{\phantom{\dag}}\tau^\sharp_{V,V}\big)(v_1\otimes v_2) &:= \big(\mu_2\circ\RR^{-1}\circ\tau^\sharp_{V,V}\circ\triangle_\CF(\delta_\lambda^\lact)\big)(v_1\otimes v_2) \\[4pt]
&= (-1)^{|v_1|\,|v_2|} \, \big(\mu_2\circ\RR^{-1}\big)\big(v_2\otimes\delta_\lambda^\lact v_1 + \delta^\lact_{\sfR^k(\lambda)}v_2\otimes\sfR_k(v_1)\big) \\[4pt]
&= (-1)^{|v_1|\,|v_2|} \, \mu_2\big(\sfR_l(v_2)\otimes\sfR^l(\delta_\lambda^\lact v_1) + \sfR_l(\delta^\lact_{\sfR^k(\lambda)}v_2)\otimes\sfR^l\,\sfR_k(v_1)\big) \\[4pt]
&= \mu_2\big(\delta_\lambda^\lact v_1\otimes v_2 + \sfR_k(v_1)\otimes \delta^\lact_{\sfR^k(\lambda)}v_2\big) \\[4pt]
&= \big(\mu_2\circ\triangle_\CF(\delta^\lact_\lambda)\big)(v_1\otimes v_2) \\[4pt]
&=: \delta_\lambda^\lact\mu_2(v_1\otimes v_2) \ ,
\end{align*}
where in the fourth equality we used braided symmetry \eqref{eq:musym}. 
This consistency check also obviously holds when $\mu_2$ is braided antisymmetric, and it easily extends to braided (anti)symmetric maps $\mu_n$ on higher tensor products $\midoplus^nV$ for $n\geq3$. Similarly, one verifies consistency with right braided gauge transformations. 

Note that here and in the following the ordering between the braiding isomorphism and gauge transformations is important: one must write a map as a composition of morphisms in ${}_\CF\CCM^\sharp$ and precompose them with $\triangle_\CF(\delta_\lambda^{\lact,\ract})$ as in the definition \eqref{eq:deltaLmudef}, including \smash{$\mu_2\circ{}_{\textrm{\tiny$\CF$}}^{\phantom{\dag}}\tau^\sharp_{V,V}$}, because in general the braiding and the twisted coproduct do not commute: \smash{${}_{\textrm{\tiny$\CF$}}^{\phantom{\dag}}\tau^\sharp_{V,V}\circ\triangle_\CF(\delta_\lambda^{\lact,\ract})\neq\triangle_\CF(\delta_\lambda^{\lact,\ract})\circ{}_{\textrm{\tiny$\CF$}}^{\phantom{\dag}}\tau^\sharp_{V,V}$}. In other words
\begin{align*}
\triangle_\CF(\delta_\lambda^{\lact})\big(\sfR_l(v_2)\otimes\sfR^l(v_1)\big) \neq \sfR_l(\delta^\lact_{\sfR^k(\lambda)}v_2)\otimes\sfR^l\,\sfR_k(v_1) +  \sfR_l(v_2)\otimes\sfR^l(\delta_\lambda^\lact v_1) \ ,
\end{align*}
and similarly for right gauge transformations.

\subsubsection*{Drinfel'd twist deformations of gauge symmetries}

If our noncommutative field theory arises from Drinfel'd twist quantization of a classical field theory on a manifold $M$, then a left braided gauge transformation $\delta_\lambda^{\star\lact}:V[[\hbar]]\to V[[\hbar]]$ is induced from a classical gauge transformation $\delta_\lambda:V\to V$ through the usual prescription \eqref{eq:mustar}:
\begin{align*}
\delta_\lambda^{\star\lact}v := \delta_{\bar\sff^k(\lambda)}\bar\sff_k(v)
\end{align*}
for $\lambda\in V_0$ and $v\in V$. The corresponding right braided gauge transformation $\delta_\lambda^{\star\ract}:V[[\hbar]]\to V[[\hbar]]$ is given by
\begin{align*}
\delta_\lambda^{\star\ract}v = \delta^{\star\lact}_{\sfR_k(\lambda)}\sfR^k(v) = \delta_{\bar\sff_k(\lambda)}\bar\sff^k(v) \ ,
\end{align*}
where in the second equality we used $\CF^{-1}\,\RR_{21}=\CF^{-1}\,\RR^{-1} = \CF^{-1}_{21}$.

\begin{example}\label{ex:protoex}
Let us look at the prototypical case that will
apply to all examples in this paper. Let $\frg$ be a Lie algebra with Lie bracket $[-,-]_\frg$, and
let $W$ be a linear representation of $\frg$; we denote the action of $\frg$
on $W$ by $\triangleright:\frg\otimes W\to W$. The action of a gauge
parameter $\lambda\in\Omega^0(M,\frg)$ on a $p$-form field
$\phi\in\Omega^p(M,W)$ valued in $W$ is given by\footnote{The minus sign is the standard convention in gauge theory, so that the action of the corresponding Lie group is a right action.}
\begin{align}\label{eq:repg}
\delta_\lambda\phi=-\lambda\triangleright\phi \ .
\end{align}
In this definition we make no distinction between left and right
actions.

To pass to the twist deformed field theory, we regard the
representation spaces $W$ as trivial $U\frv$-modules and
$\Omega^p(M,W)$ as the tensor product of $U\frv$-modules, similarly to
the graded Lie algebra $\Omega^\bullet(M,\frg)$
(cf. Section~\ref{sec:twistedforms}). We may then define both a
{left} braided representation
\begin{align}\label{eq:lgtphi}
\delta_\lambda^{\star\lact}\phi=-\lambda\triangleright_\star\phi :=-
  \bar\sff^k(\lambda)\triangleright\bar\sff_k(\phi) 
\end{align}
and a {right} braided representation
\begin{align}\label{eq:rgtphi}
\delta_\lambda^{\star\ract}\phi=-\phi\,_{\star\!\!}\triangleleft\lambda :=-
  \sfR_k(\lambda)\triangleright_\star\sfR^k(\phi) =-
  \bar\sff^l\,\sfR_k(\lambda)\triangleright\bar\sff_l\,\sfR^k(\phi) = -\bar\sff_k(\lambda)\triangleright\bar\sff^k(\phi) \ ,
\end{align}
where the first expression encodes the intuition of swapping factors
and then acting on the representation. The twisted coproducts $\triangle_\CF^n(\delta_\lambda^{\star\lact})$ and $\triangle_\CF^n(\delta_\lambda^{\star\ract})$ for $n\geq2$ then define the tensor products of left and right braided representations, respectively. 

Both left and right gauge transformations provide a representation of the braided Lie algebra $\big(\Omega^0(M,\frg)[[\hbar]],[-,-]_\frg^\star\big)$ on $\Omega^p(M,W)[[\hbar]]$:
\begin{align}\label{eq:protobraidedrep}
\big[\delta_{\lambda_1}^{\star\lact,\ract},\delta_{\lambda_2}^{\star\lact,\ract}\big]_\circ^\star = \delta_{[\lambda_1,\lambda_2]_\frg^\star}^{\star\lact,\ract} \ ,
\end{align}
for all $\lambda_1,\lambda_2\in\Omega^0(M,\frg)$. For this, we compute
\begin{align*}
\delta_{\lambda_1}^{\star\lact}\delta_{\lambda_2}^{\star\lact}\phi &= \delta_{\lambda_1}^{\star\lact}\big(-\lambda_2\triangleright_\star\phi \big) \\[4pt]
&= -\sfR_k(\lambda_2)\triangleright_\star\delta_{\sfR^k(\lambda_1)}^{\star\lact}\phi \\[4pt]
&= \bar\sff^m\sfR_k(\lambda_2)\triangleright\big(\bar\sff_{m\swone}\bar\sff^l\sfR^k(\lambda_1)\triangleright\bar\sff_{m\swtwo}\bar\sff_l(\phi)\big) \\[4pt]
&=  \bar\sff^m_{\swone}\bar\sff^l\sfR_k(\lambda_2)\triangleright\big(\bar\sff^m_{\swtwo}\bar\sff_l\sfR^k(\lambda_1)\triangleright\bar\sff_m(\phi)\big) \\[4pt]
&= \bar\sff^m_{\swone}\bar\sff^l\sff^p\bar\sff_k(\lambda_2)\triangleright\big(\bar\sff^m_{\swtwo}\bar\sff_l\sff_p\bar\sff^k(\lambda_1)\triangleright\bar\sff_m(\phi)\big) \\[4pt]
&= \bar\sff^m_{\swtwo}\bar\sff_k(\lambda_2)\triangleright\big(\bar\sff^m_{\swone}\bar\sff^k(\lambda_1)\triangleright\bar\sff_m(\phi)\big) 
\end{align*}
where in the fourth line we used the cocycle property \eqref{eq:cocyclesw} for the inverse of the Drinfel'd twist, in the fifth line we used $\RR_{21}=\RR^{-1}=\CF\,\CF_{21}^{-1}$, and in the last step we used $\CF^{-1}\,\CF=1\otimes1$ along with cocommutativity of the trivial coproduct $\Delta$. With analogous calculations, this correspondingly gives
\begin{align*}
\delta_{\sfR_l(\lambda_2)}^{\star\lact}\delta_{\sfR^l(\lambda_1)}^{\star\lact}\phi &= \bar\sff^m_{\swtwo}\bar\sff_k\sfR^l(\lambda_1)\triangleright\big(\bar\sff^m_{\swone}\bar\sff^k\sfR_l(\lambda_2)\triangleright\bar\sff_m(\phi)\big) = \bar\sff^m_{\swone}\bar\sff^l(\lambda_1)\triangleright\big(\bar\sff^m_{\swtwo}\bar\sff_l(\lambda_2)\triangleright\bar\sff_m(\phi)\big) \ .
\end{align*}
Thus we get
\begin{align*}
\big[\delta_{\lambda_1}^{\star\lact},\delta_{\lambda_2}^{\star\lact}\big]_\circ^\star\phi &= \bar\sff^m_{\swtwo}\bar\sff_l(\lambda_2)\triangleright\big(\bar\sff^m_{\swone}\bar\sff^l(\lambda_1)\triangleright\bar\sff_m(\phi)\big) - \bar\sff^m_{\swone}\bar\sff^l(\lambda_1)\triangleright\big(\bar\sff^m_{\swtwo}\bar\sff_l(\lambda_2)\triangleright\bar\sff_m(\phi)\big) \\[4pt]
&= -[\bar\sff^m_{\swone}\bar\sff^l(\lambda_1),\bar\sff^m_{\swtwo}\bar\sff_l(\lambda_2)]_\frg\triangleright\bar\sff_m(\phi) \\[4pt]
&= \bar\sff^m\big(-[\lambda_1,\lambda_2]_\frg^\star\big)\triangleright\bar\sff_m(\phi) \\[4pt]
&= \delta_{[\lambda_1,\lambda_2]_\frg^\star}^{\star\lact}\phi \ ,
\end{align*}
where in the second line we used the fact that \eqref{eq:repg} is a representation of $\frg$ on $\Omega^p(M,W)$. The proof for right gauge transformations is completely analogous.

For the special case when the field $\phi$ is a gauge connection $a\in\Omega^1(M,\frg)$, $W$ is the adjoint representation of $\frg$, and the left and right braided adjoint actions are given by 
\begin{align*} 
\lambda\triangleright_\star^{\rm ad}a := [\lambda,a]_\frg^\star \qquad \mbox{and} \qquad a\,^{\rm ad}\!\!{}_{\star\!\!}\triangleleft\lambda := [\sfR_k(\lambda),\sfR^k(a)]_\frg^\star = -[a,\lambda]_\frg^\star \ ,
\end{align*}
where in the last equality we used braided antisymmetry of the braided Lie bracket involving gauge parameters. For a matrix Lie algebra $\frg$ one can write
\begin{align*}
[\lambda,a]_\frg^\star = \lambda\star a - \sfR_k(a)\star\sfR^k(\lambda) \ ,
\end{align*}
where the star-product here is the tensor product of matrix multiplication with the twisted product of functions and forms.
The left and right braided adjoint representations induce the left and right braided gauge transformations
\begin{align*}
\delta_\lambda^{\star\lact}a := \dd\lambda - [\lambda,a]_\frg^\star \qquad \mbox{and} \qquad \delta_\lambda^{\star\ract}a := \dd\lambda+[a,\lambda]_\frg^\star \ .
\end{align*}
These give the representations \eqref{eq:protobraidedrep} on $\Omega^1(M,\frg)[[\hbar]]$. This can be shown by using braided antisymmetry and the braided Jacobi identity for the bracket $[-,-]_\frg^\star$, along with the undeformed Leibniz rule for the exterior derivative $\dd$ (cf. Section~\ref{sec:twistedforms}); it will follow as a special case of a more general result in Section~\ref{sec:braidedkinLinfty} below (see Example~\ref{ex:protoLinftygauge}).
\end{example}

The braiding of symmetries illustrated by Example~\ref{ex:protoex} is independent of the particular form of a (classical) Lie algebra $V_0$ of gauge parameters: the only necessary condition is that $V_0$ is a Lie algebra in the category of $U\frv$-modules.

\subsection{Braided versus star-gauge symmetry}
\label{sec:braidedvsstargauge}

Example~\ref{ex:protoex} illustrates an important distinction between braided gauge symmetries and the well-known star-gauge symmetries. These latter gauge transformations are the most prominent in the noncommutative field theory literature, with applications ranging from the Standard Model to gravity, and they are the ones which are consistent with the embeddings of these theories into string theory. In these more conventional settings, one works with a matrix Lie algebra $\frg$ and introduces a star-commutator bracket
\begin{align}\label{eq:starcomm}
[\lambda_1 \!\stackrel{\scriptstyle\star}{\scriptstyle,}\! \lambda_2]_{\frg} := \lambda_1\star\lambda_2 - \lambda_2\star\lambda_1 \ , 
\end{align}
where here the product is the tensor product of star-multiplication with matrix multiplication of $\lambda_1,\lambda_2\in\Omega^0(M,\frg)$. Demanding closure as an ordinary (not braided) Lie algebra is then not possible for gauge algebras beyond the unitary Lie algebras $\frg = \mathfrak{u}(N)$ in the fundamental representation, unless one extends the space of gauge parameters $\Omega^0(M,\frg)$~\cite{AschCast}. This correspondingly introduces new dynamical degrees of freedom into the noncommutative field theory{, which are hard to justify or motivate from a physical standpoint, and moreover cause problems when considering commutative limits of the theories}. One can avoid the introduction of new fields by instead regarding \eqref{eq:starcomm} as an element of the enveloping algebra $U\frg[[\hbar]]$ and using the Seiberg--Witten map~\cite{Jurco:2000ja}; however, the Seiberg--Witten map is not generally known in closed form and so this approach is limited to a description of star-gauge transformations to only the first few orders in the deformation parameter~$\hbar$.

In contrast, in the present setting, the Lie algebra $\big(\Omega^0(M,\frg),[-,-]_\frg\big)$ of gauge parameters is deformed to the braided Lie algebra $\big(\Omega^0(M,\frg)[[\hbar]],[-,-]_\frg^\star\big)$ without the need of extending the vector space $\Omega^0(M,\frg)$ (beyond the usual formal power series extension), for \emph{arbitrary} Lie algebras. Even for a matrix Lie algebra $\frg$, using $\CF_{21}^{-1}=\CF^{-1}\,\RR^{-1}=\CF^{-1}\,\RR_{21}$ one finds that the braided commutator
\begin{align}\label{eq:braidedcomm}
[\lambda_1,\lambda_2]_\frg^\star = \lambda_1\star\lambda_2 - \sfR_k(\lambda_2)\star\sfR^k(\lambda_1)
\end{align}
is not equal to the star-commutator \eqref{eq:starcomm}, even for the unitary Lie algebras and for simple twists like the Moyal--Weyl star-product of Example~\ref{ex:MoyalWeylstar}; a striking example is the Lie algebra $\frg=\mathfrak{u}(1)$ (or any abelian Lie algebra), where the braided commutator \eqref{eq:braidedcomm} always vanishes but the star-commutator \eqref{eq:starcomm} does not.
In our approach the twisting of the gauge algebra and its representations themselves avoids the introduction of new fields altogether. The dynamical sector of the noncommutative field theory is then correspondingly constructed by twist deforming the entire underlying $L_\infty$-algebra structure.

A hybrid definition of `twisted gauge symmetry', which in a sense lies in between the standard star-gauge transformations and our braided gauge transformations, was suggested in~\cite{Oeckl:2000eg,Vassilevich:2006tc,Aschieri:2006ye}. In our setting, this would follow from twist deforming the coproduct of gauge transformations $\triangle(\delta_\lambda)$ by regarding the twist $\CF$ in the universal enveloping algebra of the semi-direct product Lie algebra $\Gamma(TM)\ltimes \Omega^0(M,\frg)$. Then the gauge transformations are just the usual classical gauge transformations, but they act as braided derivations similarly to our braided gauge transformations. Since the actions of $\Omega^0(M,\frg)$ themselves are not deformed, the closure of gauge transformations is the same as in the classical case: the gauge algebra closes in the classical Lie algebra for all Lie algebras $\frg$ and their representations. However, this approach inevitably also involves star-commutators, rather than braided commutators, and so also introduces additional dynamical degrees of freedom~\cite{Aschieri:2006ye}, since the star-commutator does not generally close in the (matrix) Lie algebra $\frg$. This mixing between braided and star-gauge transformations also does not appear to be consistent or natural within the category ${}_\CF\CCM^\sharp$.

\subsection{Braided noncommutative kinematics}
\label{sec:braidedkin}

Continuing Example~\ref{ex:protoex}, we will now define the natural geometric objects which are covariant under braided gauge transformations. These will form the building blocks for the braided gauge-invariant noncommutative field theories that we study in detail later on. Aspects of this story have been studied extensively at the kinematical level by \cite{BraidedGauge}, and also alluded to by \cite{SchenkelThesis}. We will see how this extends to include dynamics by systematically defining the full theory using the braided $L_\infty$-algebra formalism, which constructs braided field theories in a model-independent way.

\subsubsection*{Braided covariant derivatives}

Given a (classical) connection $1$-form $a\in\Omega^1(M,\frg)$, the covariant derivative $\dd^a\phi\in\Omega^{p+1}(M,W)$ of any field $\phi\in\Omega^p(M,W)$ is given classically by
\begin{align*}
\dd^a\phi=\dd\phi + a\triangleright\phi \ ,
\end{align*}
where the action of $\Omega^1(M,\frg)$ on $\Omega^p(M,W)$ is the tensor product of exterior multiplication of forms with the action of $\frg$ on the representation $W$. In the twist deformed field theory, the transformation of $\phi$ under left gauge transformations \eqref{eq:lgtphi} suggests that the correct definition of the \emph{left} covariant derivative is simply by acting with the corresponding left braided representation:
\begin{definition}Given a ({braided}) connection 1-form $a \in \Omega^{1}(M,\frg)$, the \emph{left braided covariant derivative}  of a field $\phi \in \Omega^{p}(M,W)$ is 
\begin{align*}
\dd^a_{\star\lact}\phi := \dd\phi + a\triangleright_\star\phi \ . 
\end{align*}
\end{definition}
This is indeed covariant under left braided gauge transformations:
\begin{align*}
\delta_\lambda^{\star\lact}(\dd_{\star\lact}^a\phi) &= \delta_\lambda^{\star\lact}(\dd\phi + a\triangleright_\star\phi) \\[4pt]
&= \dd\delta_\lambda^{\star\lact}\phi + \delta_\lambda^{\star\lact}a\triangleright_\star\phi + \sfR_k(a)\triangleright_\star\delta_{\sfR^k(\lambda)}^{\star\lact}\phi \\[4pt]
&= -\dd(\lambda\triangleright_\star\phi) + (\dd\lambda-[\lambda,a]_\frg^\star)\triangleright_\star\phi - \sfR_k(a)\triangleright_\star\big(\sfR^k(\lambda)\triangleright_\star\phi\big) \\[4pt]
&= -\dd\lambda\triangleright_\star\phi - \lambda\triangleright_\star\dd\phi + \dd\lambda\triangleright_\star\phi-[\lambda,a]_\frg^\star\triangleright_\star\phi - \sfR_k(a)\triangleright_\star\big(\sfR^k(\lambda)\triangleright_\star\phi\big) \\[4pt]
&= -\lambda\triangleright_\star\dd\phi - \lambda\triangleright_\star(a\triangleright_\star\phi) \\[4pt]
&= -\lambda\triangleright_\star(\dd_{\star\lact}^a\phi) \ ,
\end{align*}
where in the second and fifth lines we used the twisted coproduct $\triangle_\CF(\delta_\lambda^{\star\lact})$ to implement the Leibniz rule from the left.

Similarly, we can show that the left covariant derivative is also covariant under \emph{right} braided gauge transformations \eqref{eq:rgtphi}: 
\begin{align*}
\delta_\lambda^{\star\ract}(\dd_{\star\lact}^a\phi) &=\delta_\lambda^{\star\ract}(\dd\phi+a\triangleright_\star\phi) \\[4pt]
&= \dd\delta_\lambda^{\star\ract}\phi + a\triangleright_\star\delta_\lambda^{\star\ract}\phi + \delta_{\sfR_k(\lambda)}^{\star\ract}a\triangleright_\star\sfR^k(\phi) \\[4pt]
&= \dd(-\phi\,_{\star\!\!}\triangleleft\lambda) + a\triangleright_\star(-\phi\,_{\star\!\!}\triangleleft\lambda) + \big(\dd\sfR_k(\lambda) + [a,\sfR_k(\lambda)]_\frg^\star\big)\triangleright_\star\sfR^k(\phi) \\[4pt]
&= -\dd\phi\,_{\star\!\!}\triangleleft\lambda - \phi\,_{\star\!\!}\triangleleft\dd\lambda - a\triangleright_\star\big(\sfR_k(\lambda)\triangleright_\star\sfR^k(\phi)\big) + \phi\,_{\star\!\!}\triangleleft\dd\lambda + [a,\sfR_k(\lambda)]_\frg^\star\triangleright_\star \sfR^k(\phi) \\[4pt]
&= -\dd\phi\,_{\star\!\!}\triangleleft\lambda - [a,\sfR_k(\lambda)]_\frg^\star\triangleright_\star\sfR^k(\phi) - \sfR_l\sfR_k(\lambda)\triangleright_\star\big(\sfR^l(a)\triangleright_\star\sfR^k(\phi)\big) + [a,\sfR_k(\lambda)]_\frg^\star\triangleright_\star \sfR^k(\phi) \\[4pt]
&= -\dd\phi\,_{\star\!\!}\triangleleft\lambda - \sfR_k(\lambda)\triangleright_\star\sfR^k(a\triangleright_\star\phi) \\[4pt]
&= -(\dd\phi+a\triangleright_\star\phi)\,_{\star\!\!}\triangleleft\lambda \\[4pt] 
&= -(\dd_{\star\lact}^a\phi)\,_{\star\!\!}\triangleleft\lambda \ ,
\end{align*}
where in the second line we used the twisted coproduct $\triangle_\CF(\delta_\lambda^{\star\ract})$ to implement the braided Leibniz rule from the right, in the fifth line we used the twisted coproduct $\triangle_\CF(\delta_\lambda^{\star\lact})$, and in the sixth line we used the second $\RR$-matrix identity of \eqref{eq:Rmatrixidswn}.

As one might have anticipated, we can define analogously the \emph{right} covariant derivative of $\phi\in\Omega^p(M,W)$ by
\begin{align*}
\dd_{\star\ract}^a\phi := \dd\phi + (-1)^p\,\phi\,_{\star\!\!}\triangleleft a = \dd\phi + \sfR_k(a)\triangleright_\star\sfR^k(\phi) \ .
\end{align*}
Similar calculations to those of the left covariant derivative establish that the right covariant derivative is covariant under both left and right braided gauge transformations.

What distinguishes the left and right braided covariant derivatives is their behaviour with respect to the action of the noncommutative algebra of functions on $M$. The vector space of braided fields $\Omega^p(M,W)[[\hbar]]$ carrying a braided representation of $\frg$ is a bimodule over the algebra $\CA_\star=\big(C^\infty(M)[[\hbar]],\star\big)$: one can star-multiply with functions on the left and on the right to produce $p$-forms which transform identically under left and right braided gauge transformations, because functions transform in the trivial representation of the Lie algebra $\frg$. For example, for a function $f\in \CA_\star$ we have
\begin{align*}
\delta_\lambda^{\star\lact}(f\star\phi) = \sfR_k(f)\star\delta_{\sfR^k(\lambda)}\phi = \sfR_k(f)\star\big(-\sfR^k(\lambda)\triangleright_\star\phi\big) = -\lambda\triangleright_\star(f\star\phi) \ ,
\end{align*}
where in the last equality we used the left braided Leibniz rule with $\lambda$ acting via the trivial representation on $f$. Similar results hold for right star-multiplication and right gauge transformations.

However, the left and right covariant derivatives do not satisfy the expected (undeformed) Leibniz rules when viewing $\Omega^p(M,W)[[\hbar]]$ as a bimodule over $\CA_\star$, but only as right and left $\CA_\star$-modules respectively: $\dd_{\star\lact}^a$ is a derivation with respect to right star-multiplication by $\CA_\star$ and analogously $\dd_{\star\ract}^a$ is a left $\CA_\star$-module map. This is where potential confusion for our left/right definitions may arise. Explicitly, it is easy to check the derivation properties
\begin{align*}
\dd_{\star\lact}^a(\phi\star f) = \dd_{\star\lact}^a\phi\star f + (-1)^p\,\phi\wedge_\star\dd f \qquad \mbox{and} \qquad \dd_{\star\ract}^a(f\star\phi) = \dd f\wedge_\star\phi + f\star\dd_{\star\ract}^a\phi \ ,
\end{align*}
but the Leibniz rule does not hold for the other two $\CA_\star$-module actions (here $\wedge_\star$ is simply the twisted exterior product of the form parts). 

In classical differential geometry one can define a connection directly as the covariant derivative map. In noncommutative differential geometry, one has two options which are \emph{a priori} different: a left module connection and a right module connection (see e.g.~\cite{SchenkelThesis,Aschieri:2020ifa}). In contrast, here we have a single gauge connection $a$ which induces both covariant derivatives.

\subsubsection*{Braided curvature}

The left curvature of the gauge field $a$ is now similarly defined with $W$ the adjoint representation and the usual factor of $\frac12$:
\begin{definition}Given a braided gauge field $a\in \Omega^{1}(M,\frg)$, the \emph{left braided curvature} is 
\begin{align}\label{eq:Fadef}
F_a^{\star\lact}:=\dd a + \tfrac12\,[a,a]_\frg^\star \ .
\end{align}
\end{definition}
The right curvature is correspondingly
\begin{align*}
F_a^{\star\ract} := \dd a + \tfrac12\,[\sfR_k(a),\sfR^k(a)]_\frg^\star =\dd a + \tfrac12\,[a,a]_\frg^\star = F_a^{\star\lact} \ ,
\end{align*}
where we used braided symmetry of the bracket $[-,-]_\frg^\star$ on $1$-forms. Hence there is only one curvature $F_a^\star:=F_a^{\star\lact}=F_a^{\star\ract}$. The proof of left and right covariance of the curvature follows from analogous calculations to those above:
\begin{align*}
\delta_\lambda^{\star\lact}F^\star_a = -[\lambda,F^\star_a]_\frg^\star \qquad \mbox{and} \qquad \delta_\lambda^{\star\ract}F^\star_a = [F^\star_a,\lambda]_\frg^\star \ .
\end{align*}
For a matrix Lie algebra $\frg$, one can write
\begin{align*}
F^\star_a = \dd a + \tfrac12\,\big(a\wedge_\star a + \sfR_k(a)\wedge_\star\sfR^k(a)\big) \ ,
\end{align*}
where here the star-product is the tensor product of matrix multiplication with the twisted exterior product; this differs from the usual definition in conventional noncommutative gauge theories, where one would instead write $F_a^\star=\dd a + a\wedge_\star a$.

One of the novelties of the braided curvature $2$-form is that it generally violates the Bianchi identity, that is, it is no longer covariantly constant. The braided deformations of the Bianchi identity are the three-forms in $\Omega^3(M,\frg)[[\hbar]]$ given by
\begin{align}\label{eq:braidedBianchi}
\dd^a_{\star\lact}F_a &= \tfrac12\,[F_a,a]_\frg^\star + \tfrac12\,[a,F_a]_\frg^\star -\tfrac14\,[\sfR_k(a),[\sfR^k(a),a]_\frg^\star]_\frg^\star \ , \nn \\[4pt]
\dd^a_{\star\ract}F_a &= -\tfrac12\,[F_a,a]_\frg^\star - \tfrac12\,[a,F_a]_\frg^\star -\tfrac14\,[\sfR_k(a),[\sfR^k(a),a]_\frg^\star]_\frg^\star \ . 
\end{align}
Both of these relations simply follow by taking the exterior derivative of the expression \eqref{eq:Fadef} and rewriting it in terms of $F_a$ to generate the differential identity
\begin{align*}
\dd F_a = \tfrac12\,[F_a,a]_\frg^\star - \tfrac12\,[a,F_a]_\frg^\star - \tfrac14\,\big([[a,a]_\frg^\star,a]_\frg^\star - [a,[a,a]_\frg^\star]_\frg^\star\big) \ ,
\end{align*}
and using the (odd) braided Jacobi identity \eqref{eq:braidedJacobiforms} on $1$-forms $\Omega^1(M,\frg)[[\hbar]]$ with $\alpha_1=\alpha_2=\alpha_3=a$; contrary to the classical case, this differential identity also involves the gauge field $a$ itself. This violation is well-known in twisted noncommutative differential geometry, see e.g.~\cite{Barnes:2016cjm,NAGravity,Aschieri:2020ifa}. For trivial $\RR$-matrix, the right-hand sides of these covariant derivative equations vanish identically and we recover the 
usual Bianchi identity in classical geometry.

One of the remarkable features of our braided field theory formalism will be that such kinematical violations are simply natural consequences of the braided homotopy relations for the underlying $L_\infty$-algebra, as we discuss further later on. In other words, the violations of the Bianchi identity here will be controlled by a braided $L_\infty$-algebra. Let us finally turn to this formalism.

\subsection{Braided gauge transformations and $L_\infty$-algebras}
\label{sec:braidedkinLinfty}

Generalizing Example~\ref{ex:protoex}, let us now formulate a braided version of the $L_\infty$-algebra
gauge transformations \eqref{gaugetransfA}. Let $(V,\{\ell_n\})$ be a $4$-term braided $L_\infty$-algebra in the category ${}_\CF\CCM^\sharp$ of graded left $U_\CF\frv$-modules. {Following the classical case, we have}

\begin{definition} For a gauge parameter $\lambda\in V_0$, the \emph{left braided gauge transformation} of a field $A\in V_1$ is
\begin{align}\label{eq:LgtL} 
\delta_\lambda^\lact A := \ell_1(\lambda) + \sum_{n =1}^\infty \, \frac{1}{n!}\, (-1)^{\frac12\,{n\,(n-1)}}\, \ell_{n+1}(\lambda,A,\dots,A) 
 \ .
\end{align}
\end{definition}
The right braided gauge transformation is similarly defined by
\begin{align}\label{eq:LgtR}
\delta_\lambda^\ract A := \delta^\lact_{\sfR_k(\lambda)}\sfR^k(A) &= \ell_1(\lambda) + \sum_{n =1}^\infty \, \frac{1}{n!}\, (-1)^{\frac12\,{n\,(n-1)}}\, \ell_{n+1}\big(\sfR_k(\lambda),\sfR^k(A\otimes\dots\otimes A)\big) \nonumber \\[4pt] &= \ell_1(\lambda) + \sum_{n =1}^\infty \, \frac{1}{n!}\, (-1)^{\frac12\,{n\,(n+1)}}\, \ell_{n+1}(A,\dots,A,\lambda) \ ,
\end{align}
where in the last equality we used braided antisymmetry of the $n{+}1$-brackets under interchange of gauge parameters and fields, together with the $\RR$-matrix identities in \eqref{eq:Rmatrixidswn} (and their inverses) to write
\begin{align}
\ell_{n+1}\big(\sfR_k(\lambda),\sfR^k(A\otimes\cdots\otimes A)\big) &=  \ell_{n+1}\big(\sfR_{k_1}\cdots\sfR_{k_n}(\lambda),\sfR^{k_1}(A),\dots,\sfR^{k_n}(A)\big) \nonumber \\[4pt]
&= (-1)^n \, \ell_{n+1}\big(\sfR_{l_1}\sfR^{k_1}(A),\dots,\sfR_{l_n}\sfR^{k_n}(A),\sfR^{l_n}\cdots\sfR^{l_1}\sfR_{k_1}\cdots\sfR_{k_n}(\lambda)\big) \nonumber \\[4pt]
&= (-1)^n \, \ell_{n+1}\big(\sfR_l\sfR^k(A\otimes\cdots\otimes A),\sfR^l\sfR_k(\lambda)\big) \nonumber \\[4pt]
&= (-1)^n \, \ell_{n+1}(A,\dots,A,\lambda) \ . \label{eq:ellRlambdaRA}
\end{align}

\subsubsection*{Braided Lie algebras of gauge transformations}

Whereas a closure statement for the braided gauge $L_\infty$-algebra should follow as in the classical case \eqref{eq:closure}, the general proof is complicated by the complex explicit forms of the braided homotopy relations. For the field theories that we study in this paper, the gauge algebra is always field-independent and closes off-shell, that is, without the need to impose any field equations. In these instances, it is relatively straightforward to establish closure of these left and right braided gauge transformations. For this, as suggested from the theories we shall study, we require that the $4$-term braided $L_\infty$-algebra $(V,\{\ell_n\})$ has the vanishing sets of brackets
\begin{align}\label{eq:vanishingells}
\ell_{n+2}(\lambda_1,\lambda_2,A_1,\dots,A_n)=0 \qquad \mbox{and} \qquad \ell_{n+2}(\lambda_1,\lambda_2,F,A_1,\dots,A_{n-1})=0
\end{align}
for all $n\geq1$, $\lambda_1,\lambda_2\in V_0$, $A_1,\dots,A_n\in V_1$ and $F\in V_2$. If our $L_\infty$-algebra is further equipped with a (braided or strictly) cyclic structure of degree $-3$, as will be the case for all field theories considered in this paper, then the two vanishing conditions in \eqref{eq:vanishingells} are equivalent. 
In this case both left and right braided gauge transformations \eqref{eq:LgtL} and \eqref{eq:LgtR} on fields $V_1$ close the braided Lie algebra
\begin{align}\label{eq:braidedclosure}
\big[\delta_{\lambda_1}^{\textrm{\tiny L,R}},\delta_{\lambda_2}^{\textrm{\tiny L,R}}\big]_\circ^\star = \delta_{-\ell_2(\lambda_1,\lambda_2)}^{\textrm{\tiny L,R}} \ .
\end{align}

Let us indicate why the braided commutation relations \eqref{eq:braidedclosure} hold. Using \eqref{eq:LgtL} and the braided Leibniz rule \eqref{eq:braidedLeibnizleft} we compute the composition of two left braided gauge transformations evaluated on a field $A\in V_1$ to get
\begin{align*}
\delta^\lact_{\lambda_1}\delta^\lact_{\lambda_2}A &= \sum_{n=1}^\infty\, \frac{(-1)^{\frac12\,n\,(n-1)}}{n!} \, \ell_{n+1}\big(\sfR_l(\lambda_2),\triangle^n_\CF(\delta_{\sfR^l(\lambda_1)}^\lact)(A\otimes\cdots \otimes A)\big) \\[4pt]
&= \sum_{n=1}^\infty\, \frac{(-1)^{\frac12\,n\,(n-1)}}{n!} \, \Big( \ell_{n+1}\big(\sfR_l(\lambda_2),\delta_{\sfR^l(\lambda_1)}^\lact A,A^{\otimes n-1}\big) \\ & \quad \, \hspace{1cm} + \sum_{i=1}^{n-1} \, \ell_{n+1}\big(\sfR_l(\lambda_2),\sfR_{k_1}(A),\dots,\sfR_{k_i}(A),\delta^\lact_{\sfR^{k_i}\cdots\sfR^{k_1}\sfR^l(\lambda_1)}A,A^{\otimes n-i-1}\big)\Big) \\[4pt]
&= \sum_{n,m=0}^\infty \, \frac{(-1)^{\frac12\,n\,(n+1) + \frac12\,m\,(m-1)}}{(n+1)!\,m!} \, \Big(\ell_{n+2}\big(\sfR_l(\lambda_2),\ell_{m+1}(\sfR^l(\lambda_1),A^{\otimes m}),A^{\otimes n}\big) \\
& \quad \, \hspace{1cm} + \sum_{i=1}^n \, \ell_{n+2}\big(\sfR_l(\lambda_2),\sfR_{k_1}(A),\dots,\sfR_{k_i}(A),\ell_{m+1}(\sfR^{k_i}\cdots\sfR^{k_1}\sfR^l(\lambda_1),A^{\otimes m}),A^{\otimes n-i}\big) \Big) \\[4pt]
&= \sum_{n,m=0}^\infty \, \frac{(-1)^{\frac12\,n\,(n+1) + \frac12\,m\,(m-1)}}{(n+1)!\,m!} \ \sum_{i=0}^n \, \ell_{n+2}\big(\sfR_l(\lambda_2),\sfR_k(A^{\otimes i}),\ell_{m+1}(\sfR^k\sfR^l(\lambda_1),A^{\otimes m}),A^{\otimes n-i}\big)
\end{align*}
where in the last step we used the first identity from \eqref{eq:Rmatrixidswn}. Reordering the double sum we obtain 
\begin{align*}
\delta^\lact_{\lambda_1}\delta^\lact_{\lambda_2}A = \sum_{n=0}^\infty \, \frac{(-1)^{\frac12\,n\,(n-1)}}{n!} \ \sum_{j=0}^n \, & \frac{(-1)^{(j-1)\,(n-j)}}{n-j+1} \, \binom{n}{j} \\ & \times \ \sum_{i=0}^{n-j} \, \ell_{n-j+2}\big(\sfR_l(\lambda_2),\sfR_k(A^{\otimes i}),\ell_{j+1}(\sfR^k\sfR^l(\lambda_1),A^{\otimes j}),A^{\otimes n-i-j}\big) \ .
\end{align*}
For the corresponding composition $\delta^\lact_{\sfR_p(\lambda_2)}\delta^\lact_{\sfR^p(\lambda_1)}A$, we use
\begin{align*}
& \ell_{n-j+2}\big(\sfR_l\sfR^p(\lambda_1),\sfR_k(A^{\otimes i}),\ell_{j+1}(\sfR^k\sfR^l\sfR_p(\lambda_2),A^{\otimes i}),A^{\otimes n-i-j}\big) \\[4pt]
& \hspace{6cm} = \ell_{n-j+2}\big(\lambda_1,\sfR_k(A^{\otimes i}),\ell_{j+1}(\sfR^k(\lambda_2),A^{\otimes j}),A^{\otimes n-i-j}\big) \ .
\end{align*}

Collecting terms order by order in tensor powers $A^{\otimes n}$, we see that the closure formula \eqref{eq:braidedclosure} for left gauge transformations evaluated on a field $A\in V_1$ is satisfied if
\begin{align}\label{eq:closurecond}
\sum_{j=0}^n \ \sum_{i=0}^{n-j} \, \frac{(-1)^{(j-1)\,(n-j)}}{n-j+1} \, \binom{n}{j} \, \Big( & \, \ell_{n-j+2}\big(\lambda_1,\sfR_k(A^{\otimes i}),\ell_{j+1}(\sfR^k(\lambda_2),A^{\otimes j}),A^{\otimes n-i-j}\big) \nonumber \\ & -\ell_{n-j+2}\big(\sfR_l(\lambda_2),\sfR_k(A^{\otimes i}),\ell_{j+1}(\sfR^k\sfR^l(\lambda_1),A^{\otimes j}),A^{\otimes n-i-j}\big)\Big) \nonumber \\[4pt] & \hspace{4cm} = \ell_{n+1}\big(\ell_2(\lambda_1,\lambda_2),A^{\otimes n}\big) 
\end{align}
for each $n\geq0$. This is a consequence of the braided homotopy relations $\CJ_{n+2}(\lambda_1,\lambda_2,A^{\otimes n})=0$ from \eqref{eq:Jnmaps} and the imposed restrictions \eqref{eq:vanishingells}. To illustrate the complicated $\RR$-matrix gymnastics that are involved in the proof, even with our simplifying assumptions \eqref{eq:vanishingells}, we demonstrate this explicitly for the first two orders. 

For $n=0$ the left-hand side of \eqref{eq:closurecond} evaluates to
\begin{align*}
\ell_2\big(\lambda_1,\ell_1(\lambda_2)\big) - \ell_2\big(\sfR_l(\lambda_2),\ell_1(\sfR^l(\lambda_1))\big) &= \ell_2\big(\lambda_1,\ell_1(\lambda_2)\big) - \ell_2\big(\sfR_l(\lambda_2),\sfR^l(\ell_1(\lambda_1))\big) \\[4pt]
&= \ell_2\big(\lambda_1,\ell_1(\lambda_2)\big) + \ell_2\big(\ell_1(\lambda_1),\lambda_2\big) \\[4pt]
&= \ell_1\big(\ell_2(\lambda_1,\lambda_2)\big) 
\end{align*}
as required, where in the first equality we used $U\frv$-equivariance of the map $\ell_1$, in the second equality we used braided antisymmetry of the $2$-bracket $\ell_2$, and in the last step we used the derivation property of the homotopy identities \eqref{eq:l1l2braided}. 

For $n=1$ the left-hand side of \eqref{eq:closurecond} is
\begin{align*}
&-\tfrac12\,\ell_3\big(\lambda_1,\ell_1(\lambda_2),A\big) -\tfrac12\,\ell_3\big(\lambda_1,\sfR_k(A),\ell_1(\sfR^k(\lambda_2))\big) + \ell_2\big(\lambda_1,\ell_2(\lambda_2,A)\big) \\ & \quad \, 
+ \tfrac12\,\ell_3\big(\sfR_l(\lambda_2),\ell_1(\sfR^l(\lambda_1)),A\big) + \tfrac12\,\ell_3\big(\sfR_l(\lambda_2),\sfR_k(A),\ell_1(\sfR^k\sfR^l(\lambda_1)\big) - \ell_2\big(\sfR_l(\lambda_2),\ell_2(\sfR^l(\lambda_1),A)\big) \\[4pt]
& \hspace{1cm} = -\tfrac12\,\ell_3\big(\lambda_1,\ell_1(\lambda_2),A\big) - \tfrac12\,\ell_3\big(\lambda_1,\sfR_k(A),\sfR^k(\ell_1(\lambda_2))\big) -\ell_2\big(\ell_2(\sfR_{k\bar\swone}(\lambda_2),\sfR_{k\bar\swtwo}(A)),\sfR^k(\lambda_1)\big) \\ & \hspace{1cm} \quad \, + \tfrac12\,\ell_3\big(\sfR_l(\lambda_2),\sfR^l(\ell_1(\lambda_1)),A\big) + \tfrac12\,\ell_3\big(\sfR_l(\lambda_2),\sfR_k(A),\sfR^k\sfR^l(\ell_1(\lambda_1))\big) \\ & \hspace{1cm} \quad \, + \ell_2\big(\ell_2(\sfR_{k\bar\swone}\sfR^l(\lambda_1),\sfR_{k\bar\swtwo}(A)),\sfR^k\sfR_l(\lambda_2)\big) \\[4pt]
& \hspace{1cm} = -\ell_3\big(\lambda_1,\ell_1(\lambda_2),A\big) - \ell_2\big(\ell_2(\sfR_k(\lambda_2),\sfR_l(A)),\sfR^l\sfR^k(\lambda_1)\big) - \tfrac12\,\ell_3\big(\ell_1(\lambda_1),\lambda_2,A\big) \\
& \hspace{1cm} \quad \, + \tfrac12\,\ell_3\big(\sfR_{k\bar\swone}(\lambda_2),\sfR_{k\bar\swtwo}(A),\sfR^k(\ell_1(\lambda_1))\big) + \ell_2\big(\ell_2(\sfR_k\sfR^l(\lambda_1),\sfR_m(A)),\sfR^m\sfR^k\sfR_l(\lambda_2)\big) \\[4pt]
& \hspace{1cm} = -\ell_3\big(\lambda_1,\ell_1(\lambda_2),A\big) - \tfrac12\,\ell_3\big(\ell_1(\lambda_1),\lambda_2,A\big) - \tfrac12\,\ell_3\big(\sfR_m\sfR_l\sfR^k(\ell_1(\lambda_1)),\sfR^m\sfR_{k\bar\swone}(\lambda_2),\sfR^l\sfR_{k\bar\swtwo}(A)\big) \\
& \hspace{1cm} \quad \, -\ell_2\big(\ell_2(\sfR_k(\lambda_2),\sfR_l(A)),\sfR^l\sfR^k(\lambda_1)\big) + \ell_2\big(\ell_2(\lambda_1,\sfR_m(A)),\sfR^m(\lambda_2)\big) \\[4pt]
& \hspace{1cm} = -\ell_3\big(\lambda_1,\ell_1(\lambda_2),A\big) - \tfrac12\,\ell_3\big(\ell_1(\lambda_1),\lambda_2,A\big) - \tfrac12\,\ell_3\big(\sfR_l\sfR^k(\ell_1(\lambda_1)),\sfR^l_{\bar\swone}\sfR_{k\bar\swone}(\lambda_2),\sfR^l_{\bar\swtwo}\sfR_{k\bar\swone}(A)\big) \\
& \hspace{1cm} \quad \, -\ell_2\big(\ell_2(\sfR_k(\lambda_2),\sfR_l(A)),\sfR^l\sfR^k(\lambda_1)\big) + \ell_2\big(\ell_2(\lambda_1,\sfR_m(A)),\sfR^m(\lambda_2)\big) \\[4pt]
& \hspace{1cm} = -\ell_3\big(\lambda_1,\ell_1(\lambda_2),A\big) - \ell_3\big(\ell_1(\lambda_1),\lambda_2,A\big) \\
& \hspace{1cm} \quad \,-\ell_2\big(\ell_2(\sfR_k(\lambda_2),\sfR_l(A)),\sfR^l\sfR^k(\lambda_1)\big) + \ell_2\big(\ell_2(\lambda_1,\sfR_m(A)),\sfR^m(\lambda_2)\big) \\[4pt]
& \hspace{1cm} = \ell_2\big(\ell_2(\lambda_1,\lambda_2),A\big)
\end{align*}
as required. In the first equality we used $U\frv$-equivariance and braided antisymmetry of $\ell_2$. In the second equality we used braided antisymmetry of $\ell_3$ and the second $\RR$-matrix identity \eqref{eq:Rmatrixidswn}. In the third equality we used braided symmetry on degree~1 elements followed by braided antisymmetry involving degree~0 elements. In the fourth equality we used the first $\RR$-matrix identity from \eqref{eq:Rmatrixidswn}.
In the final step we used the braided homotopy relation \eqref{I3braided} together with the vanishing brackets $\ell_3(\lambda_1,\lambda_2,A)=0$ and $\ell_3(\lambda_1,\lambda_2,\ell_1(A))=0$ by our requirements \eqref{eq:vanishingells}. Similarly, one checks explicitly the leading orders of the closure relation \eqref{eq:braidedclosure} evaluated on a field $A\in V_1$ for right braided gauge transformations \eqref{eq:LgtR}. A complete proof, including the general statement of the closure relation, is beyond the scope of this paper.
\begin{example}\label{ex:protoLinftygauge}
Returning to Example~\ref{ex:protoex},
let us now consider a $4$-term braided $L_\infty$-algebra whose kinematical sector is given by the homogeneous subspaces
\begin{align*}
V_0=\Omega^0(M,\frg) \qquad \mbox{and} \qquad V_1=\Omega^p(M,W)\times\Omega^1(M,\frg) \ .
\end{align*}
Comparing the left and right braided gauge transformations with the prescriptions \eqref{eq:LgtL} and \eqref{eq:LgtR} we read off the differential as
\begin{align*}
\ell_1^\star(\lambda) = (0,\dd\lambda) \ ,
\end{align*}
while the only non-zero brackets involving gauge parameters and fields are given by
\begin{align*}
\ell_2^\star\big(\lambda,(\phi,a)\big) = \big(-\lambda\triangleright_\star\phi,-[\lambda,a]_\frg^\star\big) \qquad \mbox{and} \qquad \ell_2^\star\big((\phi,a),\lambda\big) = \big(\phi\,_{\star\!\!}\triangleleft\lambda, -[a,\lambda]_\frg^\star\big) \ .
\end{align*}
The brackets in this sector satisfy \eqref{eq:vanishingells}, and comparing \eqref{eq:protobraidedrep} with \eqref{eq:braidedclosure} we find the non-zero $2$-bracket among gauge parameters
\begin{align*}
\ell_2^\star(\lambda_1,\lambda_2) = -[\lambda_1,\lambda_2]_\frg^\star \ .
\end{align*}
\end{example}

Example~\ref{ex:protoLinftygauge} illustrates the typical features of braided gauge symmetries in most field theories of interest in physics, and in particular those treated in the present paper. Given the complexity of the calculations illustrated above using braided homotopy relations even at the lowest orders, for simplicity and ease of exposition, in the following we will impose the further vanishing conditions
\begin{align}\label{eq:physgt}
\ell_{n+1}(\lambda,A_1,\dots,A_n)=0 \qquad \mbox{and} \qquad \ell_{n+1}(\lambda,F,A_1,\dots,A_{n-1})=0 \ ,
\end{align}
for all $n\geq2$, $\lambda\in V_0$, $A_1,\dots,A_n\in V_1$ and $F\in V_2$.
The first condition implies that the braided gauge transformations are at most linear in the field $A$, that is, $\delta_\lambda^\lact A=\ell_1(\lambda)+\ell_2(\lambda,A)$ and $\delta_\lambda^\ract A=\ell_1(\lambda)-\ell_2(A,\lambda)$. The second condition is again equivalent to the first for (braided or strictly) cyclic $L_\infty$-algebras. With these assumptions, the lowest orders calculations presented above constitute a complete proof of the closure condition \eqref{eq:braidedclosure}.

\subsection{Braided covariant dynamics}
\label{sec:braidedEOM}

For a field $A\in V_1$, the braided analog of \eqref{EOM} which dictates the field equations follows exactly as in the classical case.
\begin{definition}
For $A\in V_1$, the \emph{braided field equations} $F_A=0$ in $V_2$ are given by
\begin{align}\label{eq:braidedeom}
F_A := \sum_{n =1}^\infty \, \frac{1}{n!}\, (-1)^{\frac12\,{n\,(n-1)}}\, \ell_{n}(A,\dots,A) \ .
\end{align} 
\end{definition}
\subsubsection*{Left gauge covariance}

The expression \eqref{eq:braidedeom} is covariant under the left gauge transformations discussed in Section~\ref{sec:braidedkinLinfty}, namely
\begin{align}\label{eq:lefteomcov}
\delta_\lambda^\lact F_A = \sum_{n =0}^\infty \, \frac{1}{(n+1)!}\, (-1)^{\frac12\,{n\,(n-1)}} \ \sum_{i=0}^{n} \, (-1)^i \, \ell_{n+2}\big(\lambda,A^{\otimes i},F_A,A^{\otimes n-i}\big) \ ,
\end{align}
for all gauge parameters $\lambda\in V_0$. In the classical case, where $\RR=1\otimes1$, the expression \eqref{eq:lefteomcov} agrees with \eqref{gaugetransfF}.

Let us indicate why \eqref{eq:lefteomcov} is true. We compute the gauge transformation of \eqref{eq:braidedeom} using \eqref{eq:LgtL} and similar calculations to those in Section~\ref{sec:braidedkinLinfty}:
\begin{align}\label{eq:deltaLFA}
\delta_\lambda^\lact F_A &= \sum_{n =1}^\infty \, \frac{1}{n!}\, (-1)^{\frac12\,{n\,(n-1)}}\, \ell_{n}\big(\triangle_\CF^n(\delta^\lact_\lambda) (A\otimes\dots\otimes A)\big) \nonumber \\[4pt]
&= \sum_{n,m=0}^\infty \, \frac{(-1)^{\frac12\,n\,(n+1)+\frac12\,m\,(m-1)}}{(n+1)!\,m!} \ \sum_{i=0}^n \,  \ell_{n+1}\big(\sfR_k(A^{\otimes i}),\ell_{m+1}(\sfR^k(\lambda),A^{\otimes m}),A^{\otimes n-i}\big) \\[4pt]
&= \sum_{n=0}^\infty \frac{(-1)^{\frac12\,n\,(n-1)}}{n!} \, \sum_{j=0}^n \, \sum_{i=0}^{n-j} \, \frac{(-1)^{(j-1)\,(n-j)}}{n-j+1} \, \binom{n}{j} \, \ell_{n-j+1}\big(\sfR_k(A^{\otimes i}),\ell_{j+1}(\sfR^k(\lambda),A^{\otimes j}),A^{\otimes n-i-j}\big) \nonumber
\end{align}
where in the last step we reorganized the double sum. On the other hand, using \eqref{eq:braidedeom} the right-hand side of \eqref{eq:lefteomcov} is equal to
\begin{align*}
\sum_{n,m=0}^\infty \, & \frac{(-1)^{\frac12\,n\,(n-1) + \frac12\,m\,(m-1)}}{(n+1)!\,m!} \ \sum_{i=0}^n \, (-1)^i \, \ell_{n+2}\big(\lambda,A^{\otimes i},\ell_{m}(A^{\otimes m}),A^{\otimes n-i}\big) \\[4pt] & \hspace{1cm} = \sum_{n=0}^\infty \, \frac{(-1)^{\frac12\,n\,(n-1)}}{n!} \ \sum_{j=0}^n \ \sum_{i=0}^{n-j} \, \frac{(-1)^{j\,(n-j)+i}}{n-j+1} \, \binom{n}{j} \, \ell_{n-j+2}\big(\lambda,A^{\otimes i},\ell_j(A^{\otimes j}),A^{\otimes n-i-j}\big)
\end{align*}
where we set $\ell_0:=0$.

Thus, collecting terms order by order in tensor powers of the fields $A^{\otimes n}$, we see that \eqref{eq:lefteomcov} is satisfied if
\begin{align}
\sum_{j=0}^n \ \sum_{i=0}^{n-j}\, \frac{(-1)^{j\,(n-j)}}{n-j+1} \, \binom{n}{j} \, \Big( (-1)^{n-j} \, & \, \ell_{n-j+1}\big(\sfR_k(A^{\otimes i}),\ell_{j+1}(\sfR^k(\lambda),A^{\otimes j}),A^{\otimes n-i-j}\big) \label{eq:leftcovcond} \\ & \hspace{2cm} - (-1)^i \, \ell_{n-j+2}\big(\lambda,A^{\otimes i},\ell_j(A^{\otimes j}),A^{\otimes n-i-j}\big)\Big) = 0 \nonumber
\end{align}
for each $n\geq0$. This now follows by the braided homotopy relations $\CJ_{n+1}(\lambda,A^{\otimes n})=0$ from \eqref{eq:Jnmaps}. 
For $n=0$ the expression \eqref{eq:leftcovcond} reads $\ell_1(\ell_1(\lambda))=0$, which is the differential identity in \eqref{eq:l1l2braided}. For $n=1$ the left-hand side of the expression \eqref{eq:leftcovcond} reads
\begin{align*}
& -\tfrac12\,\ell_2\big(\ell_1(\lambda),A\big) - \tfrac12\, \ell_2\big(\sfR_k(A),\ell_1(\sfR^k(\lambda))\big) + \ell_1\big(\ell_2(\lambda,A)\big) - \ell_2\big(\lambda,\ell_1(A)\big) \\[4pt]
& \hspace{8cm} = \ell_1\big(\ell_2(\lambda,A)\big) - \ell_2\big(\ell_1(\lambda),A\big) - \ell_2\big(\lambda,\ell_1(A)\big) = 0 
\end{align*}
as required, where we used $U\frv$-equivariance of $\ell_1$ and braided symmetry of $\ell_2$ on degree~$1$ elements in the first step, and the derivation identity of \eqref{eq:l1l2braided} in the last step. For $n=2$, with similar manipulations one can reduce the left-hand side of \eqref{eq:leftcovcond} to
\begin{align*}
& \tfrac13\,\ell_3\big(\ell_1(\lambda),A,A\big) + \tfrac13\,\ell_3\big(\sfR_k(A),\ell_1(\sfR^k(\lambda)),A\big) + \tfrac13\,\ell_3\big(\sfR_{k\bar\swone}(A),\sfR_{k\bar\swtwo}(A),\ell_1(\sfR^k(\lambda))\big) \\ & \quad \, +\ell_2\big(\ell_2(\lambda,A),A\big) - \ell_2\big(\sfR_k(A),\ell_2(\sfR^k(\lambda),A)\big) + \ell_3\big(\lambda,\ell_1(A),A) - \ell_3\big(\lambda,A,\ell_1(A)\big) \\ & \quad \, +\ell_1\big(\ell_3(\lambda,A,A)\big) - \ell_2\big(\lambda,\ell_2(A,A)\big) \\[4pt]
& \hspace{1cm} = \ell_3\big(\ell_1(\lambda),A,A\big) + \ell_3\big(\lambda,\ell_1(A),A\big) - \ell_3\big(A,A,\ell_1(A)\big) + \ell_1\big(\ell_3(\lambda,A,A)\big) \\ & \hspace{2cm} + \ell_2\big(\ell_2(\lambda,A),A\big) + \ell_2\big(\ell_2(\lambda,\sfR_k(A)),\sfR^k(A)\big) + \ell_2\big(\ell_2(\sfR_k(A),\sfR_l(A)),\sfR^l\sfR^k(\lambda)\big) = 0 
\end{align*}
as required, where in the final step we used \eqref{I3braided}.

For $n\geq3$, we will only give the proof under the vanishing conditions \eqref{eq:physgt} that we motivated in Section~\ref{sec:braidedkinLinfty}; that is, we show $\delta_\lambda^\lact F_A = \ell_2(\lambda,F_A)$. In this case, all summands of the second line of \eqref{eq:leftcovcond} vanish except for the term with $j=n$, which is $-\ell_2\big(\lambda,\ell_n(A^{\otimes n})\big)$. In the first line, only the summands with $j=0,1$ are non-zero. For $j=0$ the contributions to the sum over $i$ are
\begin{align*}
\ell_{n+1}\big(\sfR_k(A^{\otimes i}),\ell_1(\sfR^k(\lambda)),A^{\otimes n-i}\big) = \ell_{n+1}\big(\ell_1(\lambda),A^{\otimes n}\big)
\end{align*}
for each $i=0,1,\dots,n$, where we used $U\frv$-equivariance of the bracket $\ell_1$ and braided symmetry of the brackets $\ell_{n+1}$ on degree~$1$ elements. For $j=1$ the contributions to the sum over $i$ can be written in the form
\begin{align*}
\ell_n\big(\sfR_k(A^{\otimes i}),\ell_2(\sfR^k(\lambda),A),A^{\otimes n-1-i}\big) &= \ell_n\big(\sfR_l(\ell_2(\sfR^k(\lambda),A),\sfR^l\sfR_k(A^{\otimes i}),A^{\otimes n-1-i}\big) \\[4pt]
&= \ell_n\big(\ell_2(\lambda,\sfR_k(A)),\sfR^k(A^{\otimes i}),A^{\otimes n-1-i}\big) 
\end{align*}
for each $i=0,1,\dots,n-1$,
where in the first equality we used braided symmetry of $\ell_n$ on degree~$1$ elements, and in the last step we used the first $\RR$-matrix identity from \eqref{eq:Rmatrixidswn} when multiplied on the right by $\sfR^k\otimes1\otimes\sfR_k$ and acted on by the braided antisymmetric map $\ell_2$ of the first two factors. 

Collecting everything together, the condition \eqref{eq:leftcovcond} in this way reads
\begin{align}\label{eq:leftcovJn+1}
(-1)^n \, \ell_{n+1}\big(\ell_1(\lambda),A^{\otimes n}\big) = \ell_2\big(\lambda,\ell_n(A^{\otimes n})\big) - \sum_{i=1}^n \, \ell_n\big(\ell_2(\lambda,\sfR_k(A)),\sfR^k(A^{\otimes i-1}),A^{\otimes n-i}\big) \ .
\end{align}
But \eqref{eq:leftcovJn+1} is just the braided homotopy relation $\CJ_{n+1}(\lambda,A^{\otimes n})=0$ from \eqref{eq:Jnmaps} in this case: terms in the sum over $i$ in $\CJ_{n+1}(\lambda,A^{\otimes n})$ with $i\neq1,2,n$ all vanish because each one involves a bracket from \eqref{eq:physgt}. After removing the terms involving the vanishing brackets from \eqref{eq:physgt}, the remaining terms in $\CJ_{n+1}(\lambda,A^{\otimes n})=0$ with $i=1,2,n$ are precisely the relations \eqref{eq:leftcovJn+1} for all $n\geq3$.

\subsubsection*{Right gauge covariance}

The covariance of \eqref{eq:braidedeom} under right gauge transformations is similarly given by
\begin{align*}
\delta_\lambda^\ract F_A = \delta_{\sfR_k(\lambda)}^\lact \sfR^k(F_A) = \sum_{n =0}^\infty \, \frac{1}{(n+1)!}\, (-1)^{\frac12\,{n\,(n+1)}} \ \sum_{i=0}^n \, (-1)^{i+1} \, \ell_{n+2}\big(A^{\otimes i},F_A,A^{\otimes n-i},\lambda\big) \ .
\end{align*}
The proof is completely analogous to that above and will not be repeated.

Because of this, in the remainder of this paper we will mostly only mention left gauge transformations $\delta_\lambda^\lact$ explicitly, with the corresponding statements for right gauge transformations being the obvious modifications involving suitable rearrangements of arguments and sign changes. Indeed left and right braided gauge transformations are not independent, as the right transformations are defined in terms of left transformations (cf. \eqref{eq:LgtR}). This will become apparent in Section~\ref{sec:braidedNoether}, where we will also see that they both induce the same interdependence among the field equations.

\subsubsection*{Space of physical states}

One of the novelties of braided field theory is the extent to which its symmetries can be regarded as genuine gauge symmetries from a dynamical perspective. Recall that in the classical case of an ordinary field theory on a manifold $M$, by a (infinitesimal) gauge symmetry one usually means an infinite-dimensional symmetry parametrised by sections of a vector bundle over $M$; in the case of Lagrangian field theories these are symmetries of the action functional. As such, they induce redundancies in the degrees of freedom which are solutions of the field equations. For example, suppose that $M$ is a globally hyperbolic spacetime. Choosing a Cauchy surface $\Sigma\subset M$, along with values of the fields and their jets on $\Sigma$ representing initial data, for any solution of the field equations which extends the values of the fields on $M$, one may use the gauge symmetries to produce new solutions with the same initial data. Generally, by Noether's second theorem, gauge symmetries are in a one-to-one correspondence with differential identities among the field equations, which further exhibits that not all degrees of freedom are independent: the space of \emph{physical} states is the moduli space of classical solutions modulo gauge transformations. 

In a braided field theory, the situation is markedly different. In this case there is no properly defined moduli space of classical solutions, as braided gauge transformations do not generally act on the subspace of solutions $A\in V_1$ to the field equations $F_A=0$ and one cannot take the quotient by braided gauge transformations. In other words, braided gauge symmetries \emph{do not} produce new solutions of the field equations. To understand this point, let us consider the covariance of the field equations. Let $A\in V_1$ be a solution of 
\begin{align*}
F_A = 0 \ .
\end{align*}
From covariance \eqref{eq:lefteomcov} it then follows that
\begin{align*}
\delta_\lambda^\lact F_A = 0
\end{align*}
on solutions of $F_A=0$, for any gauge parameter $\lambda\in V_0$. However, computing the first order variation in the infinitesimal gauge parameter we find
\begin{align*}
\left. \frac\dd{\dd t}\right|_{t=0} \, F_{A+t\,\delta_\lambda^\lact A} &= \left. \frac\dd{\dd t}\right|_{t=0} \ \sum_{n=1}^\infty \, \frac1{n!} \, (-1)^{\frac12\,n\,(n-1)} \, \ell_n(A+t\,\delta_\lambda^\lact A,\dots,A+t\,\delta_\lambda^\lact A) \\[4pt]
&= \sum_{n=1}^\infty \, \frac1{n!} \, (-1)^{\frac12\,n\,(n-1)} \ \sum_{i=0}^n \, \ell_n\big(A^{\otimes i},\delta_\lambda^\lact A, A^{\otimes n-1-i}\big) \\[4pt]
&= \sum_{n,m=0}^\infty \, \frac{(-1)^{\frac12\,n\,(n+1) + \frac12\,m\,(m-1)}}{(n+1)! \, m!} \ \sum_{i=0}^n \, \ell_{n+1}\big(A^{\otimes i},\ell_{m+1}(\lambda,A^{\otimes m}),A^{\otimes n-i}\big) \ .
\end{align*}
This differs from the second line of \eqref{eq:deltaLFA} by the $\RR$-matrix factors in the sum over~$i$, and we have shown that in general
\begin{align*}
F_{A+\delta_\lambda^\lact A} \neq F_A + \delta_\lambda^\lact F_A \ ,
\end{align*}
{even at order $\mathcal{O}(\lambda^1)$}, except in the classical case where $\RR=1\otimes 1$ and the expected equality holds. This implies that a braided gauge transformation of a solution $A$ to the field equations $F_A=0$ may result in a field with $F_{A+\delta_\lambda^\lact A}\neq0${, even at order $\mathcal{O}(\lambda^1)$}.

What is then the meaning of a braided gauge symmetry, and why should one even speak of them as ``gauge'' transformations? The answer comes from a braided version of Noether's second theorem: Even though the solution generating aspect of gauge transformations is lost, they still induce interdependence of the field equations, and hence gauge redundancies in this sense. These identities follow, analogously to the classical case, from our braided $L_\infty$-algebra formalism, and is our next step in the construction of a braided field theory in this formulation.

\subsection{Braided Noether identities}
\label{sec:braidedNoether}

The most complicated part of the story, and yet the most important in light of our discussion about the significance of braided gauge symmetries from Section~\ref{sec:braidedEOM}, concerns the braided version of the Noether identities. They are relations for the field equations $F_A$ which hold off-shell, that is, when $F_A\neq0$, and are a consequence of the braided homotopy relations $\CJ_n(A^{\otimes n})=0$ for all $n\geq1$. In Section~\ref{sec:braidedaction} we will discuss how they correspond to braided gauge symmetries.

Working under the same assumptions \eqref{eq:physgt}, we will prove the general vanishing identity
\begin{align}\label{braidedNoether}
\dsf_AF_A & := \ell_1(F_A) + \frac12\,\big(\ell_2(F_A,A) - \ell_2(A,F_A)\big) - \sum_{n=3}^\infty \, \frac1{n!} \, (-1)^{\frac12\,n\,(n-1)} \, \ell_1\big(\ell_{n}(A,\dots,A)\big) \nonumber \\[4pt]
& \quad \, -\sum_{n=2}^\infty \, \frac1{2\,n!} \, (-1)^{\frac12\,n\,(n-1)} \, \Big(\ell_2\big(\ell_n(A,\dots,A),A\big) - \ell_2\big(A,\ell_n(A,\dots,A)\big)\Big) \\[4pt]
&= 0 \ . \nonumber
\end{align}
We comment on the general case below. Note that, in contrast to the classical case \eqref{eq:Noether}, the operation $\dsf_A F_A$ is no longer linear in the field equations $F_A$ and contains inhomogeneous terms involving brackets of the fields $A$ themselves. This is related to the violations of the Bianchi identities which we discusssed at the end of Section~\ref{sec:braidedkin}, and we will come back to this point in our explicit examples later on. 

To prove the identity \eqref{braidedNoether}, we first note that, for any $4$-term braided $L_\infty$-algebra with a (braided or strict) cyclic structure of degree~$-3$, the second of the conditions in \eqref{eq:physgt} is equivalent to
\begin{align}\label{eq:physNoethercond}
\ell_{n+1}(F,A_1,\dots,A_n)=0 \ ,
\end{align}
for all $n\geq2$, $A_1,\dots,A_n\in V_1$ and $F\in V_2$. This vanishing condition significantly simplifies the braided homotopy relation $\CJ_n(A^{\otimes n})=0$ from \eqref{eq:Jnmaps}: in this case the terms with $\ell_{n+1-i}$ acting then vanish for $i\leq n-2$, and the sum over $i$ receives non-zero contributions only from terms with $i=n-1,n$.

The derivation of \eqref{braidedNoether} follows by considering a particular weighted sum of braided homotopy relations $\CJ_n(A^{\otimes n})=0$, which vanishes term by term, expanding each $\CJ_n(A^{\otimes n})$ in brackets, and manipulating the expressions algebraically to make the field equations $F_A$ appear ``maximally'', as many times as possible. We start from
\begin{align}\label{eq:weightedhomotopysum}
0 &= \sum_{n=1}^{\infty}\, \frac{1}{n!}\,(-1)^{\frac{1}{2}\, n\,(n-1)} \, \CJ_{n}(A^{\otimes n})\nn \\[4pt]
&= \sum_{n=1}^{\infty}\, \frac{1}{n!}\, (-1)^{\frac{1}{2}\, n\,(n-1)} \ \sum_{i=1}^{n} \, (-1)^{i\,(n-i)} \ \sum_{\sigma\in{\rm Sh}_{i,n-i}} \, \text{sgn}(\sigma) \, \ell_{n+1-i} \circ (\ell_{i}\otimes \id^{\otimes n-i})\circ {}_{\textrm{\tiny$\CF$}}\sigma (A^{\otimes n})\nn \\[4pt]
	&= \ell_{1}\big(\ell_{1}(A)\big)-\frac{1}{2}\,\Big(\ell_{1}\big(\ell_{2}(A,A)\big)-\ell_{2}\big(\ell_{1}(A),A\big)+\ell_{2}\big(A,\ell_{1}(A)\big)\Big) \\
	& \quad \,+ \sum_{n=3}^{\infty} \, \frac{1}{n!} \, (-1)^{\frac{1}{2}\, n\,(n-1)} \,\bigg((-1)^{n-1} \, \sum_{\sigma\in{\rm Sh}_{n-1,1}} \, \text{sgn}(\sigma)\, \ell_{2} \circ (\ell_{n-1}\otimes \id)\circ {}_{\textrm{\tiny$\CF$}}\sigma\nn \\& \hspace{10cm} + \sum_{\sigma\in{\rm Sh}_{n,0}} \, \text{sgn}(\sigma)\, \ell_{1} \circ \ell_{n}\circ {}_{\textrm{\tiny$\CF$}}\sigma\bigg) (A^{\otimes n}) \ , \nn
\end{align}
where in the third equality we expanded out the $n=1,2$ terms, while using the simplification of $\CJ_n(A^{\otimes n})$ for $n\geq 3$. Next we notice the last sum over $(n,0)$-shuffled permutations contains only the identity permutation, thus combining the first term, the second term and the $\ell_{1}\circ \ell_{n}$ terms of the sum:
\begin{align*}
0 &= \ell_{1}\Big(\ell_{1}(A)-\frac{1}{2}\,\ell_{2}(A,A) + \sum_{n=3}^{\infty}\, \frac{1}{n!} \, (-1)^{\frac{1}{2}\,n\,(n-1)} \, \ell_{n}(A^{\otimes n})\Big) + \frac{1}{2} \, \Big(\ell_{2}\big(\ell_{1}(A),A\big) - \ell_{2}\big(A,\ell_{1}(A)\big)\Big) \nn \\
& \quad \, + \sum_{n=3}^\infty \, \frac{1}{n!} \, (-1)^{\frac{1}{2}\,n\,(n-1)+n-1} \ \sum_{\sigma\in{\rm Sh}_{n-1,1}} \, \text{sgn}(\sigma) \, \ell_{2} \circ (\ell_{n-1}\otimes \id)\circ {}_{\textrm{\tiny$\CF$}}\sigma (A^{\otimes n}) \nn \\[4pt]
&= \ell_{1}(F_A)+\frac{1}{2}\,\Big(\ell_{2}\big(\ell_{1}(A),A\big) - \ell_{2}\big(A,\ell_{1}(A)\big)\Big) \\ & \quad \, + \sum_{n=2}^\infty \, \frac{1}{(n+1)!} \, (-1)^{\frac{1}{2}\,n\,(n-1)} \ \sum_{i=0}^n\, \ell_2\big(\ell_n(A^{\otimes n-i},\sfR_k(A^{\otimes i})),\sfR^k(A)\big) \ ,
\end{align*}
where we cancelled the sign factor $\text{sgn}(\sigma)$ with the same sign coming out of the graded (braided) permutation after acting with $_{\textrm{\tiny$\CF$}} \sigma$ on $A^{\otimes n}$, since $A$ is of degree $1$. Next we separate out the terms with $i=0,n$ from the sum over $i$, write $\frac{1}{(n+1)!}=\frac{1}{2\,n!}-\frac{n-1}{2\,(n+1)!}$, and use braided antisymmetry for the $i=n$ term to get
\begin{align}\label{eq:leftNoetherprelim}
0 &= \ell_{1}(F_A) +\frac{1}{2}\,\Big(\ell_{2}\big(\ell_{1}(A),A\big) - \ell_{2}\big(A,\ell_{1}(A)\big)\Big) \nn \\ & \quad \, +\sum_{n=2}^\infty \, \bigg(\frac{1}{2\,n!}-\frac{n-1}{2\,(n+1)!}\bigg)\,(-1)^{\frac{1}{2}\,n\,(n-1)}\,\Big(\ell_{2}\big(\ell_{n}(A^{\otimes n}),A\big) -\ell_{2}\big(A,\ell_{n}(A^{\otimes n})\big)  \Big) \nn  \\ & \quad \, +\sum_{n=2}^\infty \, \frac{1}{(n+1)!}\,(-1)^{\frac{1}{2}\,n\,(n-1)} \ \sum_{i=1}^{n-1} \, \ell_2\big(\ell_n(A^{\otimes n-i},\sfR_k(A^{\otimes i})),\sfR^k(A)\big)
\nn \\[4pt]
&= \ell_{1}(F_A)+\frac{1}{2}\,\big( \ell_{2}( F_A,A) - \ell_{2}(A, F_A)\big) -\frac{1}{2}\,\sum_{n=2}^\infty\,\frac{n-1}{(n+1)!}\,\Big(\ell_{2}\big(\ell_{n}(A^{\otimes n}),A\big) -\ell_{2}\big(A,\ell_{n}(A^{\otimes n})\big)\Big) \nn \\
& \quad \, + \sum_{n=2}^\infty \, \frac{1}{(n+1)!}\,(-1)^{\frac{1}{2}\,n\,(n-1)} \ \sum_{i=1}^{n-1} \, \ell_2\big(\ell_n(A^{\otimes n-i},\sfR_k(A^{\otimes i})),\sfR^k(A)\big) \ ,
\end{align}
where in the second equality we combine the second and third terms with the $\frac{1}{2\,n!}$ terms of the first sum.

The expression \eqref{eq:leftNoetherprelim} is one form of the braided Noether identities. Written in this way, the reduction to the classical case is immediate: when $\RR=1\otimes 1$ the terms within the sum over~$i$ are all equal and so the two sums cancel each other by strict antisymmetry of bracket $\ell_2$, while the second and third terms likewise combine, thus reducing the identity \eqref{eq:leftNoetherprelim} to the classical Noether identity \eqref{eq:Noether} (when \eqref{eq:physNoethercond} is applied), as expected. The identity \eqref{eq:leftNoetherprelim} can be written in a slightly simpler fashion by using $\CJ_{n+1}(A^{\otimes n+1})=0$ to rewrite the sum over $i$ as
\begin{align*}
\sum_{i=1}^{n-1} \, \ell_2\big(\ell_n(A^{\otimes n-i},\sfR_k(A^{\otimes i})),\sfR^k(A)\big) &= \ell_{2}\big(A,\ell_{n}(A^{\otimes n})\big)-\ell_2\big(\ell_{n}(A^{\otimes n}),A\big) - (-1)^{n} \, \ell_{1} \big(\ell_{n+1}(A^{\otimes n+1})\big) \ .
\end{align*}
Substituting this in \eqref{eq:leftNoetherprelim}, after a short manipulation of the prefactors we arrive at \eqref{braidedNoether}.

It is clear from this calculation how the braided Noether identities follow for any braided field theory. One simply adds up all the braided homotopy relations $\CJ_n(A^{\otimes n})=0$ and makes the field equations appear within the other brackets. The field equations $F_A$ will appear within brackets up to as high order as the brackets which appear in the braided gauge transformation formula \eqref{eq:LgtL}, like they did from the restriction \eqref{eq:physNoethercond} in the present case; many more extra brackets involving solely the fields $A$ will appear in the most general case. The details of the general case are beyond the scope of this paper and will be left for future investigations, as the present result is sufficient for all the theories we shall consider.

The correspondence with braided gauge symmetries is immediately apparent in an action formalism for the braided field theory, where the braided Noether identities $\dsf_AF_A$ are `dual' to the braided gauge transformations $\delta_\lambda^\lact A$ and $\delta_\lambda^\ract A$ via a cyclic structure on the braided $L_\infty$-algebra, as in the classical case from Section~\ref{sec:Linftygft}. Indeed this is how Noether identities are usually derived in classical field theories, and we will discuss this further in Section~\ref{sec:braidedaction} below. But we stress that this is unnnecessary and the Noether identities follow from the known braided homotopy relations in the $L_\infty$-algebra approach to field theories.

\subsection{Action formulation}
\label{sec:braidedaction}

The variational formulation of the braided covariant field equations helps to elucidate the meaning of braided gauge symmetries as discussed previously, while also demonstrating some unexpected subtleties and differences in the action principle as compared to the classical case. Although all of the dynamics of a braided field theory can be encoded in the braided $L_\infty$-algebra structure as we have thus far described, the action formalism will make more clear to what extent particular examples of braided field theories differ from their classical counterparts. 

The field equations $F_A=0$, given by \eqref{eq:braidedeom}, can be described as the stationary locus of an action functional $S:V_1\to\FR$ which is constructed by endowing our $4$-term braided $L_\infty$-algebra $(V,\{\ell_n\})$ with a non-degenerate \emph{strictly} cyclic pairing $\langle-,-\rangle:V\otimes V\to\FR$ of degree~$-3$, that is
\begin{align}\label{eq:strictsympair}
\langle v_0,v_1\rangle = \langle\sfR_k(v_0),\sfR^k(v_1)\rangle =  \langle v_1,v_0\rangle \ ,
\end{align}
and 
\begin{align}\label{eq:strictcyclic}
\langle v_0,\ell_n(v_1,v_2,\dots,v_n)\rangle = (-1)^{(|v_0|+1)\,n} \ \langle v_n,\ell_n(v_0,v_1,\dots,v_{n-1})\rangle \ ,
\end{align}
for all $n\geq1$ and $v_0,v_1,\dots,v_n\in V$. We can interpret the strict symmetry \eqref{eq:strictsympair} as saying that we focus on pairings which are invariant under the natural action of the twisted Hopf algebra. Then the action functional is given exactly as in the classical case.
\begin{definition}
The \emph{braided action functional} $S:V_{1} \rightarrow \FR$ is 
\begin{align}\label{eq:braidedaction}
S(A) := \sum_{n=1}^\infty \, \frac{1}{(n+1)!}\, (-1)^{\frac12\,{n\,(n-1)}}\, \langle A, \ell_{n}(A,\dots,A)\rangle \ ,
\end{align}
on dynamical fields $A\in V_{1}$.
\end{definition}

\subsubsection*{Variational principle}

The strictness requirement on the cyclic structure ensures that the critical locus of the action functional \eqref{eq:braidedaction} coincides with solutions of the field equations from Section~\ref{sec:braidedEOM}. For a general field variation $\delta A$ on the tangent space $TV_1$ to field space, the variation of the action functional (or any other functional of $A\in V_1$) is defined in the usual way by
\begin{align*}
\delta S(A) := \left. \frac\dd{\dd t}\right|_{t=0} \, S(A+t\,\delta A) \ .
\end{align*}
This defines $\delta$ as a strict derivation of degree~$0$, that is, it obeys an undeformed Leibniz rule with the trivial coproduct $\triangle(\delta)=\delta\otimes\id + \id\otimes\delta$. We further require that $\delta$ commutes with the $\RR$-matrices of the braiding, analogously to the action of the exterior derivative $\dd$ on twisted differential forms, but in contrast to braided gauge transformations. Similarly to Section~\ref{sec:braidedgauge}, we iterate this coproduct to act on maps $\mu_n$ defined on $\midoplus^nV$ by $\delta\mu_n:=\mu_n\circ\triangle^n(\delta)$; in particular, we define $\delta\langle-,-\rangle := \langle-,-\rangle\circ\triangle(\delta)$. 

Then the variation $\delta S$ of the action functional \eqref{eq:braidedaction} proceeds identically to the classical case, where both the cyclicity requirement and the derivation properties of $\delta$ are the same: we use
\begin{align*}
\delta\langle A,\ell_n(A^{\otimes n})\rangle = \langle\delta A,\ell_n(A^{\otimes n})\rangle + \sum_{i=0}^{n-1} \, \langle A,\ell_n(A^{\otimes i},\delta A,A^{\otimes n-i-1})\rangle = (n+1) \, \langle\delta A,\ell_n(A^{\otimes n})\rangle \ ,
\end{align*}
where in the last step we used strict cyclicity to set all terms in the sum equal to the first term. Using \eqref{eq:braidedeom} this shows that
\begin{align}\label{eq:braidedvar}
\delta S(A) = \langle \delta A,F_A\rangle
\end{align}
for arbitrary variations $\delta A$ of the fields. Hence by non-degeneracy of the pairing, the extrema of the action functional \eqref{eq:braidedaction} are precisely the solutions $A\in V_1$ of the field equations $F_A=0$.

\subsubsection*{Braided gauge invariance}
\label{sec:braidedgaugeinv}

The problem of invariance of the action functional \eqref{eq:braidedaction} under braided gauge transformations \eqref{eq:LgtL} is much more involved and does not simply follow the classical route. A potentially problematic new feature which immediately arises is the following. Recalling our discussion from Section~\ref{sec:braidedgauge} regarding the consistency with braided symmetric maps, the braided Leibniz rule is not generally compatible with the strict symmetry \eqref{eq:strictsympair} of the cyclic pairing $\langle-,-\rangle$. Strict symmetry means \smash{$\langle-,-\rangle=\langle-,-\rangle\circ\tau_{V,V}^\sharp$}, but this is not the consistent notion of symmetry for a morphism in the category \smash{${}_\CF\CCM^\sharp$}, because it only involves the graded transposition $\tau^\sharp$ which is \emph{not} the braiding isomorphism in ${}_\CF\CCM^\sharp$ unless $\RR=1\otimes1$. As a result $\delta_\lambda^\lact\langle-,-\rangle\neq\delta_\lambda^\lact(\langle-,-\rangle\circ\tau^\sharp)$; explicitly
\begin{align}\label{eq:braidedincons}
\delta_\lambda^\lact\langle v_1,v_2\rangle = \langle\delta^\lact_\lambda v_1,v_2\rangle + \langle\sfR_k(v_1),\delta^\lact_{\sfR^k(\lambda)}v_2\rangle = \langle\delta^\lact_{\sfR^k(\lambda)}v_2,\sfR_k(v_1)\rangle + \langle v_2,\delta^\lact_\lambda v_1\rangle \neq  \delta_\lambda^\lact\langle v_2,v_1\rangle \ ,
\end{align}
which is not generally consistent with $\langle v_1,v_2\rangle = \langle v_2,v_1\rangle$ unless $\RR=1\otimes 1$ in which case equality holds.\footnote{Strictly speaking, the gauge transformations $\delta_\lambda^\lact$ are not even defined on maps which do not live in \smash{${}_\CF\CCM^\sharp$}, so formally this discussion is not even well-defined from the categorical point of view.}

Despite this evident inconsistency, the gauge transformation of the action functional \eqref{eq:braidedaction} \emph{is} nevertheless well-defined, because the cyclicity condition \eqref{eq:strictcyclic} ensures compatibility with strict symmetry of the pairing order by order in tensor powers of the fields $A^{\otimes n}$ for $n\geq 1$: Using \eqref{eq:LgtL}, and repeatedly applying strict symmetry \eqref{eq:strictsympair}, cyclicity \eqref{eq:strictcyclic} and the first $\RR$-matrix identity from \eqref{eq:Rmatrixidswn}, we compute
\begin{align}\label{eq:deltaLsym}
\delta_\lambda^\lact\big\langle A,\ell_n\big(A^{\otimes n}\big)\big\rangle &= \big\langle\delta^\lact_\lambda A,\ell_n\big(A^{\otimes n}\big)\big\rangle + \big\langle \sfR_l(A) , \ell_{n}\big(\triangle_\CF^n(\delta_{\sfR^l(\lambda)}^\lact)(A^{\otimes n})\big)\big\rangle \nonumber \\[4pt]
&= \big\langle A,\ell_n\big(\delta_\lambda^\lact A,A^{\otimes n-1}\big)\big\rangle + \sum_{i=0}^{n-2} \, \big\langle\sfR_l(A),\ell_n\big(\sfR_k(A^{\otimes i}),\delta_{\sfR^k\sfR^l(\lambda)}^\lact A,A^{\otimes n-i-1}\big)\big\rangle \nonumber \\
& \hspace{6cm} + \big\langle\sfR_l(A),\ell_n\big(\sfR_k(A^{\otimes n-1}),\delta_{\sfR^k\sfR^l(\lambda)}^\lact A\big)\big\rangle \nonumber \\[4pt]
&= \big\langle\ell_n\big(\delta_\lambda^\lact A,A^{\otimes n-1}\big),A\big\rangle + \sum_{i=0}^{n-2} \, \big\langle\ell_n\big(\sfR_l(A),\sfR_k(A^{\otimes i}),\delta_{\sfR^k\sfR^l(\lambda)}^\lact A,A^{\otimes n-i-2}\big),A\big\rangle \nonumber \\
& \hspace{6cm} + \big\langle\ell_n\big(\sfR_l(A),\sfR_k(A^{\otimes n-1})\big),\delta_{\sfR^k\sfR^l(\lambda)}^\lact A\big\rangle \nonumber \\[4pt]
&= \big\langle\ell_n\big(\delta_\lambda^\lact A,A^{\otimes n-1}\big),A\big\rangle + \sum_{i=0}^{n-2} \, \big\langle\ell_n\big(\sfR_k(A^{\otimes i+1}),\delta_{\sfR^k(\lambda)}^\lact A,A^{\otimes n-i-2}\big),A\big\rangle \nonumber \\
& \hspace{6cm} + \big\langle\ell_n\big(\sfR_k(A^{\otimes n})\big),\delta_{\sfR^k(\lambda)}^\lact A\big\rangle \nonumber \\[4pt]
&= \big\langle\ell_n\big(\triangle_\CF^n(\delta_\lambda^\lact)(A^{\otimes n})\big),A\big\rangle + \big\langle\sfR_k\big(\ell_n(A^{\otimes n})\big),\delta_{\sfR^k(\lambda)}^\lact A\big\rangle \nonumber \\[4pt]
&= \delta_\lambda^\lact\big\langle\ell_n\big(A^{\otimes n}\big),A\big\rangle \ ,
\end{align}
where in the fourth equality we used
\begin{align*}
\ell_n\big(\sfR_l(A),\sfR_k(A^{\otimes i}),\ell_{m+1}(\sfR^k\sfR^l(\lambda),A^{\otimes m}),A^{\otimes n-i-2}\big) = \ell_n\big(\sfR_k(A^{\otimes i+1}),\ell_{m+1}(\sfR^k(\lambda),A^{\otimes m}),A^{\otimes n-i-2}\big)
\end{align*}
and similarly
\begin{align*}
\big\langle\ell_n\big(\sfR_l(A),\sfR_k(A^{\otimes n-1})\big),\ell_{m+1}\big(\sfR^k\sfR^l(\lambda),A^{\otimes m}\big)\big\rangle = \big\langle\ell_n\big(\sfR_k(A^{\otimes n})\big),\ell_{m+1}\big(\sfR^k(\lambda),A^{\otimes m}\big)\big\rangle \ ,
\end{align*}
for all $m\geq0$. Hence $\delta_\lambda^\lact S(A)$ is well-defined. 

Recall that in the classical case, the exact same cyclicity argument as that used to arrive at \eqref{eq:braidedvar} would show that a general gauge transformation of the action functional can be written as $\delta_\lambda S(A) = \langle \delta_\lambda A,F_A\rangle$ (cf. \eqref{eq:gtNoether}). In this case, a classical gauge transformation can be identified with a specific direction in the tangent space $TV_1$ to field space, as a special instance of a general variation of the action functional. However, while this identification is essentially tautological in the classical case, it ceases to be true in our braided setting, because general variations and braided gauge transformations are treated differently: a transformation $\delta_\lambda^\lact A$ of a field $A\in V_1$ along a gauge parameter $\lambda\in V_0$ is no longer a special case of a general variation $\delta A$ of the field, because the gauge transformations are now defined as \emph{braided} derivations, in contrast to the general variations. In other words, in the braided case the first equality in \eqref{eq:gtNoether} simply no longer holds in general:
\begin{align}\label{eq:deltalambdapair}
\delta_\lambda^\lact S(A) \neq \langle \delta_\lambda^\lact A,F_A\rangle \ .
\end{align}
This is explicitly evident in the calculation of \eqref{eq:deltaLsym}.

With these preliminary considerations out of the way, we will now establish braided gauge invariance of the action functional:
\begin{align}\label{eq:braidedgaugeinv}
\delta_\lambda^\lact S(A)=0
\end{align}
for all $\lambda\in V_0$ and $A\in V_1$. Because of \eqref{eq:deltalambdapair}, we have to work much harder than in the classical case. By using steps analogous to those in \eqref{eq:deltaLFA}, we compute the gauge variation of the action functional \eqref{eq:braidedaction} using \eqref{eq:LgtL} to write
\begin{align*}
\delta_\lambda^\lact S(A) &= \sum_{n,m=0}^\infty \, \frac{(-1)^{\frac12\,n\,(n+1)+\frac12\,m\,(m-1)}}{(n+2)!\,m!} \, \Big( \big\langle\ell_{m+1}\big(\lambda,A^{\otimes m}\big),\ell_{n+1}\big(A^{\otimes n+1}\big)\big\rangle \\ & \hspace{4cm} + \sum_{i=0}^n \,  \big\langle \sfR_l(A), \ell_{n+1}\big(\sfR_k(A^{\otimes i}),\ell_{m+1}(\sfR^k\sfR^l(\lambda),A^{\otimes m}),A^{\otimes n-i}\big)\big\rangle \Big) \ .
\end{align*}
Next we use cyclicity as well as the by now familiar manipulations with the $\RR$-matrix identities \eqref{eq:Rmatrixidswn} to isolate the gauge parameter $\lambda$ in the left entries of the pairings. This yields
\begin{align*}
\big\langle\ell_{m+1}\big(\lambda,A^{\otimes m}\big),\ell_{n+1}\big(A^{\otimes n+1}\big)\big\rangle = (-1)^{m+1} \, \big\langle\lambda,\ell_{m+1}\big(A^{\otimes m},\ell_{n+1}(A^{\otimes n+1})\big)\big\rangle \ ,
\end{align*}
and
\begin{align*}
& \big\langle \sfR_l(A), \ell_{n+1}\big(\sfR_k(A^{\otimes i}),\ell_{m+1}(\sfR^k\sfR^l(\lambda),A^{\otimes m}),A^{\otimes n-i}\big)\big\rangle \\[4pt]
& \hspace{5cm} = \big\langle A,\ell_{n+1}\big(\sfR_l(A),\sfR_k(A^{\otimes i}),\ell_{m+1}(\sfR^k\sfR^l(\lambda),A^{\otimes m}),A^{\otimes n-i-1}\big)\big\rangle \\[4pt]
& \hspace{5cm} = \big\langle A,\ell_{n+1}\big(\sfR_l(A^{\otimes i+1}),\ell_{m+1}(\sfR^l(\lambda),A^{\otimes m}),A^{\otimes n-i-1}\big)\big\rangle \\[4pt]
& \hspace{5cm} = \big\langle A,\ell_{n+1}\big(\sfR_k(\ell_{m+1}(\sfR^l(\lambda),A^{\otimes m})),\sfR^k\sfR_l(A^{\otimes i+1}),A^{\otimes n-i-1}\big)\big\rangle \\[4pt]
& \hspace{5cm} = \big\langle A,\ell_{n+1}\big(\ell_{m+1}(\lambda,\sfR_k(A^{\otimes m})),\sfR^k(A^{\otimes i+1}),A^{\otimes n-i-1}\big)\big\rangle \\[4pt]
& \hspace{5cm} = \big\langle \ell_{m+1}\big(\lambda,\sfR_k(A^{\otimes m})\big),\ell_{n+1}\big(\sfR^k(A^{\otimes i+1}),A^{\otimes n-i}\big)\big\rangle \\[4pt]
& \hspace{5cm} = (-1)^{m+1} \, \big\langle\lambda,\ell_{m+1}\big(\sfR_k(A^{\otimes m}),\ell_{n+1}(\sfR^k(A^{\otimes i+1}),A^{\otimes n-i})\big)\big\rangle \ ,
\end{align*}
for $i=0,1,\dots,n$, where in the first equality we used cyclicity, in the second equality we used the first $\RR$-matrix identity from \eqref{eq:Rmatrixidswn}, in the third equality we used braided symmetry of the bracket $\ell_{n+1}$, in the fourth equality we used the first $\RR$-matrix identity from \eqref{eq:Rmatrixidswn} again, and in the last two equalities we applied cyclicity again. Thus the gauge transformation is given by
\begin{align}\label{eq:isolatelambda}
\delta_\lambda^\lact S(A) &= - \sum_{n,m=0}^\infty \, \frac{(-1)^{\frac12\,n\,(n+1)+\frac12\,m\,(m+1)}}{(n+2)!\,m!} \ \sum_{i=0}^{n+1} \, \big\langle\lambda,\ell_{m+1}\big(\sfR_k(A^{\otimes m}),\ell_{n+1}(\sfR^k(A^{\otimes i}),A^{\otimes n-i+1})\big)\big\rangle \nonumber \\[4pt]
&= - \sum_{n=0}^\infty \frac{(-1)^{\frac12\,n\,(n+1)}}{n!} \ \sum_{j=0}^n \ \sum_{i=0}^{n-j+1} \, \frac{(-1)^{j\,(n-j)}}{(n-j+1)\,(n-j+2)} \, \binom{n}{j} \\[4pt]
& \hspace{6cm} \times \, \big\langle\lambda,\ell_{j+1}\big(\sfR_k(A^{\otimes j}),\ell_{n-j+1}(\sfR^k(A^{\otimes i}),A^{\otimes n-i-j+1})\big)\big\rangle \ , \nonumber
\end{align}
for all $\lambda\in V_0$.

By collecting terms order by order in tensor powers $A^{\otimes n+1}$, and using non-degeneracy of the cyclic pairing, we see that \eqref{eq:braidedgaugeinv} is satisfied if
\begin{align}\label{eq:gaugeinvn+1}
\sum_{j=0}^n \ \sum_{i=0}^{n-j+1} \, \frac{(-1)^{j\,(n-j)}}{(n-j+1)\,(n-j+2)} \, \binom{n}{j} \, \ell_{j+1}\big(\sfR_k(A^{\otimes j}),\ell_{n-j+1}(\sfR^k(A^{\otimes i}),A^{\otimes n-i-j+1})\big) = 0
\end{align}
for each $n\geq0$. This is a consequence of the fact that this expression can be brought to the form of the braided homotopy relations $\CJ_{n+1}(A^{\otimes n+1})=0$ from \eqref{eq:Jnmaps}. The proof follows in exactly the same vein as the proof of braided covariance of the field equations from Section~\ref{sec:braidedEOM}, and will not be repeated here; the details are left for the interested reader to fill in. 

\subsubsection*{Braided Noether identities}

We can now make a more precise statement about the correspondence between the braided Noether identities discussed in Section~\ref{sec:braidedNoether} and braided gauge symmetries of our field theory. Despite the fact that the first equality of \eqref{eq:gtNoether} does not carry over to the braided case (cf. \eqref{eq:deltalambdapair}), one can still write the gauge transformation of the action functional in the form
\begin{align*}
\delta_\lambda^\lact S(A) = -\langle\lambda,\dsf_A F_A\rangle \ ,
\end{align*}
by using cyclicity to isolate the gauge parameter $\lambda$, along with the braided homotopy relations, while making the field equations $F_A$ appear. In this sense the operation $\sfd_A$ is the `braided adjoint' to $\delta^\lact_\lambda$ with respect to the cyclic pairing, and the braided Noether identities $\dsf_A F_A=0$ then follow from braided gauge invariance $\delta_\lambda^\lact S(A)=0$ of the action functional \eqref{eq:braidedaction}, as discussed above. This is the way that Noether identities of field theories are usually derived in the classical case. 

The isolation of $\lambda$ was done in the calculation \eqref{eq:isolatelambda}, from which an explicit formula for $\dsf_A F_A$ can be extracted: the weighted sum of braided homotopy relations \eqref{eq:weightedhomotopysum} is exactly the combination that appears in \eqref{eq:isolatelambda}. What remains to be done is to massage the expression into a form involving $F_A$ explicitly. However, this is not necessary as we saw in Section~\ref{sec:braidedNoether}, as the calculation follows from the braided homotopy relations: the most effective and concise way to derive the braided Noether identities is through the braided $L_\infty$-algebra framework which we have developed in this section. Similarly, cyclicity of the pairing implies that right braided gauge invariance $\delta_\lambda^\ract S(A)=0$ induces the same braided Noether identities $\sfd_AF_A=0$.

\subsubsection*{Drinfel'd twist deformations of action functionals}

In our concrete examples later on, the pairing in a braided field theory used to define an action functional will be obtained by twist deformation quantization of the pairing used in a classical field theory. Note, however, that the resulting action functional from the braided $L_\infty$-algebra formulation is not generally the same as that obtained by naively twist deforming the Lagrangian of a classical field theory, rather we twist only the structure maps of the cyclic $L_\infty$-algebra structure used for the construction of the classical field theory. We will see explicitly in Section~\ref{sec:NCgravity4D} that the two prescriptions do not generally define the same noncommutative field theory.  

\subsection{Example: Braided Chern--Simons theory}
\label{sec:braidedCStheory}

In the remainder of this paper we apply the general framework developed in this section to concrete examples of (diffeomorphism-invariant) braided field theories. We start with the simplest illustrative example to close off the present section. Similarly to the classical case, our formalism provides a vast generalization of braided noncommutative Chern--Simons gauge theory on a three-dimensional manifold $M$ based on a Lie algebra $\frg$, which is the prototypical example of a braided field theory. Starting from the differential graded Lie algebra of $\frg$-valued exterior differential forms on $M$ from Section~\ref{sec:CStheory}, we use Proposition~\ref{prop:braidedfromclassical} to construct the $4$-term differential (graded) braided Lie algebra $\big(\Omega^\bullet(M,\frg)[[\hbar]],\ell_1^\star,\ell_2^\star\big)$ from Section~\ref{sec:twistedforms} with the brackets
\begin{align*}
\ell_1^\star(\alpha) = \dd \alpha \qquad \mbox{and} \qquad \ell_2^\star(\alpha,\beta) = -[\alpha,\beta]_\frg^\star \ ,
\end{align*}
for $\alpha,\beta\in V=\Omega^\bullet(M,\frg)$.

Using \eqref{eq:LgtL} we define the (left) braided gauge transformation of a gauge field $A\in \Omega^1(M,\frg)$ by a gauge parameter $\lambda\in\Omega^0(M,\frg)$ as
\begin{align}\label{eq:CSstargt}
\delta_{\lambda}^{\star\lact}A := \ell_1^\star(\lambda) + \ell_2^\star(\lambda,A) = \dd\lambda - [\lambda,A]_\frg^\star \ .
\end{align}
This is just the braided gauge transformation of a connection from Examples~\ref{ex:protoex} and~\ref{ex:protoLinftygauge}; in particular, they close the braided Lie algebra \eqref{eq:braidedclosure} on $\Omega^0(M,\frg)[[\hbar]]$ with bracket $\ell_2^\star(\lambda_1,\lambda_2)=-[\lambda_1,\lambda_2]_\frg^\star$. Using \eqref{eq:braidedeom}, the field equations $F_A^\star=0$ in $\Omega^2(M,\frg)[[\hbar]]$ are encoded by 
\begin{align}\label{eq:CSstarEOM}
F_A^\star := \ell_1^\star(A) - \tfrac12\,\ell_2^\star(A,A) = \dd A + \tfrac12\, [A,A]_\frg^\star \ ,
\end{align}
which is just the braided curvature of the gauge field $A$ from Section~\ref{sec:braidedkin}; in particular, from \eqref{eq:lefteomcov} we obtain the anticipated braided gauge covariance
\begin{align*}
\delta_\lambda^{\star\lact}F_A^\star = \ell_2^\star(\lambda,F_A^\star) = -[\lambda,F_A^\star]_\frg^\star \ .
\end{align*}
Thus the braided Chern--Simons field equations state that the braided connection $A$ is flat, which is just the obvious deformation of the corresponding statement from the classical case.

At this point, however, the classical and braided field theories diverge. Recalling from Section~\ref{sec:braidedEOM} that \smash{$F^\star_{A+\delta_\lambda^{\star\lact} A}\neq F^\star_A + \delta_\lambda^{\star\lact}F^\star_A$}, there is no moduli space of flat connections modulo braided gauge transformations, and the space of physical states cannot be described as in the classical case. This deviation is accounted for by the different Noether identity in the braided case, which is no longer simply the Bianchi identity for the braided curvature, which is generally violated as we discussed in Section~\ref{sec:braidedkin}. Using \eqref{braidedNoether} and the braided homotopy relations, the braided Noether identity in $\Omega^3(M,\frg)[[\hbar]]$ is given by
\begin{align}\label{eq:CSstarNoether}
\dsf_A^\star F_A^\star :=& \, \ell_1^\star(F_A^\star) - \tfrac12\,\big(\ell_2^\star(A,F_A^\star) - \ell_2^\star(F_A^\star,A)\big) + \tfrac14\,\ell_2^\star\big(\sfR_k(A),\ell_2^\star(\sfR^k(A),A)\big) \nonumber \\[4pt]
=& \, \dd F_A^\star + \tfrac12\,[A,F_A^\star]_\frg^\star - \tfrac12\,[F_A^\star,A]_\frg^\star + \tfrac14\,[\sfR_k(A),[\sfR^k(A),A]_\frg^\star]_\frg^\star \nonumber \\[4pt]
=& \, \tfrac12\,\big(\dd_{\star\lact}^AF_A^\star + \dd_{\star\ract}^AF_A^\star\big) + \tfrac14\,[\sfR_k(A),[\sfR^k(A),A]_\frg^\star]_\frg^\star \ = \ 0 \ ,
\end{align}
in agreement with the deformed Bianchi identities \eqref{eq:braidedBianchi}.
Thus $\sfd_A^\star$ does not coincide with the symmetrized braided covariant derivative $\frac12\,\big(\dd_{\star\lact}^A+\dd_{\star\ract}^A\big)$.

Finally, to formulate the braided Chern--Simons action functional, we assume that the three-manifold $M$ is compact and oriented, that $\frg$ is a quadratic Lie algebra, and that the Drinfel'd twist $\CF\in U\frv[[\hbar]]\otimes U\frv[[\hbar]]$ is Hermitian and compatible with the cyclic structure \eqref{eq:CSpairing}; in this case the compatibility conditions were discussed in Section~\ref{sec:twistedforms}. By Proposition~\ref{prop:compatiblestrict}, the pairing defined by
\begin{align*}
\langle\alpha,\beta\rangle_\star := \int_M \, \Tr_\frg(\alpha\wedge_\star\beta) \ ,
\end{align*}
with $|\alpha|+|\beta| = 3$, thus determines a strictly cyclic structure on the underlying braided $L_\infty$-algebra $\big(\Omega^\bullet(M,\frg)[[\hbar]],\ell_1^\star,\ell_2^\star\big)$. 

Then we can construct a braided gauge invariant action functional $S_\star:\Omega^1(M,\frg)\to\FR[[\hbar]]$, for which the braided Chern--Simons field equation $F_A^\star=0$ describes its critical locus, according to the general prescription of \eqref{eq:braidedaction}:
\begin{align}\label{eq:braidedCSaction}
S_\star(A) = \frac12\,\big\langle A,\ell_1^\star(A)\big\rangle_\star - \frac16 \, \big\langle A,\ell_2^\star(A,A)\big\rangle_\star = \frac12\, \int_M\, \Tr_\frg\Big(A\wedge_\star\dd A + 
 \frac13\, A\wedge_\star [A, A]^\star_\frg\Big) \ .
\end{align}
This is just the expected braided deformation of the classical Chern--Simons action functional \eqref{eq:CSaction} (with the braided Lie bracket).
Using graded cyclicity of the integral over $M$ with respect to the twisted exterior
product $\wedge_\star$, in the sense of Section~\ref{sec:twistedforms}, along with invariance of the quadratic form $\Tr_\frg$ on $\frg$ and the usual $\RR$-matrix identities \eqref{eq:Rmatrixidswn}, one may check explicitly that this action functional is invariant under the braided gauge transformations \eqref{eq:CSstargt}, that it gives the field equations defined by \eqref{eq:CSstarEOM} under arbitrary variations of the gauge fields $\delta A$, and that it yields the braided Noether identity \eqref{eq:CSstarNoether} along the lines explained in Section~\ref{sec:braidedaction} using braided gauge invariance of $S_\star(A)$ and cyclicity of the pairing.

We stress that the constraints imposed in this last step on the manifold $M$, the Lie algebra $\frg$, and the twist $\CF$ are only needed if one wishes to view the field equations as the stationary locus of an action functional. The braided gauge symmetries and dynamics of the braided field theory can be formulated entirely without these restrictions if one is satisfied with regarding it as a non-Lagrangian field theory, which for many (but not all) considerations is sufficient at the classical level.

\section{Noncommutative gravity}
\label{sec:NCgravity}

\subsection{Noncommutative deformations of general relativity}
\label{sec:NCGR}

Over the last $20$ years there has been an ongoing effort 
to construct consistent noncommutative theories of gravity, motivated from several research programmes. On one hand, it is argued that they should serve as low-energy limits of any fully fledged theory of quantum gravity, retaining certain aspects of quantum-induced noncommutativity. On the other hand, T-duality and the conjectural noncommutativity of certain string theoretic backgrounds suggest that gravity ought to be formulated in such terms as well. These noncommutative theories of gravity rely on 
the notions of noncommutative spacetime and noncommutative geometry, and in appropriate 
limits they reduce to classical general relativity. One of the main problems encountered in these 
approaches is the breaking of the diffeomorphism symmetry of general relativity. 
In most theories the diffeomorphism symmetry, or at 
least a part of it, is broken and one 
needs to make sense of this breaking and of the remaining symmetries, if any survive. To put our theories into proper context with previous literature on the subject, we start by briefly reviewing some previous theories of noncommutative gravity based on deformation quantization in the various flavours that are similar in spirit to the approach taken in the present paper: geometric quantization, star-product deformations (equivalently quantization of Poisson structures), and Drinfel'd twist deformations. 

In the second order or metric formalism, noncommutative theories of gravity were originally introduced via Drinfel'd twist deformation quantization in the seminal works~\cite{TwistApproach,Aschieri:2005zs}. In these theories one starts from the Hopf algebra of the classical diffeomorphism symmetry of general relativity and applies a Drinfel'd twist to generate the twisted Hopf algebra of diffeomorphisms, exactly as we have done in the present paper. The noncommutative differential geometry is then covariant with respect to the twisted diffeomorphism symmetry~\cite{PaoloAlex}. Consistent noncommutative deformations of the (vacuum) Einstein equations can be written down, and some 
particular solutions based on twists constructed from Killing or semi-Killing vector fields have been obtained in~\cite{TwistSolutions,Aschieri:2009qh}, as twist deformations of classical solutions and their symmetries.

In the first order or Cartan formalism, noncommutative gravity can be formulated as a 
noncommutative gauge theory of the Lorentz or (A)dS group using the enveloping algebra approach and Seiberg--Witten maps~\cite{Jurco:2000ja}, as discussed in Section~\ref{sec:braidedvsstargauge}. A noncommutative deformation of the Einstein--Cartan--Palatini action was discussed in~\cite{AschCast,Chamseddine:2000si,Cardella:2002pb,Aschieri:2012in}, and the coupling to matter fields (fermion and gauge fields) was discussed in~\cite{Aschieri:2014xka,Gocanin:2017lxl,Aschieri:2012vf,Ciric:2018ygk}; a noncommutative extension of Palatini--Holst theory was considered by~\cite{deCesare:2018cjr}. A common feature of these models is that the first order correction to the classical theory (in the deformation parameter $\hbar$) vanishes if the corresponding action functional is real. The first non-vanishing correction is then at order $\hbar^2$. A solution of the field equations corresponding to a deformation of Minkowski spacetime via the Moyal--Weyl twist was found in \cite{Ciric:2016isg}. Some of these theories also have a twisted diffeomorphism symmetry realized via the action of the star-Lie derivative from Section~\ref{sec:braideddiff} (see~\cite{Aschieri:2014xka,Aschieri:2012vf}). A relation between the Seiberg--Witten map approach and noncommutative theories of gravity via Fedosov deformation quantization of endomorphism bundles was developed in~\cite{Dobrski:2015emm}. Deformed dispersion relations in a $\kappa$-deformation of Minkowski space defined via a Jordanian twist were derived in~\cite{Aschieri:2020yft}.

The main goal of this final section is to apply the formalism of braided field theory that we have developed in Section~\ref{sec:bft} to construct new braided theories of noncommutative gravity. Our approach has many advantages compared to previous treatments. In particular, our theories are symmetric under braided diffeomorphisms, like the theories of~\cite{TwistApproach,Aschieri:2005zs,Aschieri:2014xka,Aschieri:2012vf}. However, in contrast to the first order formalism of~\cite{Aschieri:2014xka,Aschieri:2012vf}, our approach retains braided versions of all symmetries with good classical limits, without the need of introducing new spurious degrees of freedom that have no interpretation in the classical world (see the discussion in Section~\ref{sec:braidedvsstargauge}). Our theories have an action formulation and the coupling to fermion fields is straightforward. In both the metric and Cartan approaches, the failure of twisted diffeomorphisms in generating new classical solutions has always been prevalent (though seemingly not noticed before), and our framework clarifies this unconventional feature (see the discussion in Section~\ref{sec:braidedEOM}). We proceed now to our new braided versions of noncommutative gravity, treating separately the three-dimensional and four-dimensional theories in turn.

\subsection{Braided gravity in three dimensions}
\label{sec:NCgravity3D}

The braided noncommutative deformation of Einstein--Cartan--Palatini gravity in three dimensions follows a similar path to the braided Chern--Simons theory constructed in Section~\ref{sec:braidedCStheory}. It serves to introduce as well some of the ingredients that will appear later in the more complicated four-dimensional version of the theory. Our first observation is that the differential graded Lie algebra of Section~\ref{sec:ECP3d} is an $L_\infty$-algebra in the category $\CCM^\sharp$, that is, the vector space \eqref{eq:ECPvectorspace} is a graded $U\frv$-module and the brackets $\ell_1,\ell_2$ are $U\frv$-equivariant; this stems from the diffeomorphism symmetry of general relativity. We proceed to twist deform it using Proposition~\ref{prop:braidedfromclassical} to a differential (graded) braided Lie algebra, starting from the same graded vector space \eqref{eq:ECPvectorspace} and with the same notation. We use the same non-zero brackets $\ell_1$ and $\ell_2$, except that now we include an extra term which will accomodate the presence of a cosmological constant $\Lambda\in\FR$; this modifies only the $2$-brackets $\ell_2\big((e_1,\omega_1)\,,\,(e_2,\omega_2)\big)$ between fields in \eqref{eq:3dbrackets} by shifting the first slot with the term $2\,\Lambda\,e_1\dwedge e_2$~\cite{ECPLinfty}.

\subsubsection*{Braided $\boldsymbol{L_\infty}$-algebra}
\label{sec:3dLinfty}

The $4$-term braided $L_\infty$-algebra underlying three-dimensional noncommutative gravity has differential given by
\begin{align*}
\ell_{1}^\star(\xi,\rho)=(0,\dd\rho) \ , \quad
  \ell_1^\star(e,\omega)=(\dd\omega,\dd e) \qquad \mbox{and} \qquad \ell_{1}^\star(E,{\mit\Omega})=(0, \dd
  {\mit\Omega}) \ ,
\end{align*}
which as always coincides with the classical differential. The $2$-brackets are modified non-trivially from the classical case to
\begin{align}
\ell_{2}^{\star}\big((\xi_{1},\rho_{1})\,,\,(\xi_{2},\rho_{2})\big)&=\big([\xi_{1},\xi_{2}]_\frv^\star\,,\,-[\rho_{1},\rho_{2
}]_{\aso(3)}^\star+\LL_{\xi_1}^{\star}\rho_{2}
- \LL_{\sfR_k(\xi_{2})}^\star\sfR^k(\rho_{1})\big) \ , \nn
\\[4pt]
\ell_{2}^{\star}\big((\xi,\rho)\,,\,(e,\omega)\big)&=\big(-\rho \star e
+\LL_{\xi}^\star e \,,\, -[\rho,\omega]_{\aso(3)}^\star+\LL_{\xi}^\star\om\big) \
, \nn \\[4pt]
\ell_{2}^{\star}\big((\xi,\rho) \,,\, (E,{\mit\Omega})\big)&=\big(-
[\rho , E]_{\aso(3)}^\star+\LL_{\xi}^\star E \,,\, -\rho \star
{\mit\Omega}+\LL_{\xi}^\star{\mit\Omega} \big) \ , \nn \\[4pt]
\ell_{2}^{\star}\big((\xi,\rho)\,,\,(\CX,\CP)\big)&= 
                                                \Big(\dd x^\mu\otimes\Tr\big(\iota_\mu\dd\,
                                                \bar{\sff}^{k} (\rho)
                                                \dwedge \bar{\sff}_{k}
                                                (\CP)\big)+\LL_{\xi}^\star\CX\,,\,-\rho
                                                \star \CP
                                                +\LL_{\xi}^\star
                                                \CP\Big) \ , \label{eq:braided3dbrackets} \\[4pt] 
\ell_{2}^{\star}\big((e_{1},\omega_{1})\,,\,(e_{2},\omega_{2})\big)&=-\big(\,[\omega_{1}, \omega_{2}]_{\aso(3)}^{\star} + 
2\,\Lambda\, e_1\dwedge_\star e_2
\,,\, \omega_{1}\wedge_\star
e_{2} +\sfR_k(\omega_{2}) \wedge_\star \sfR^k (e_{1}) \big) \
                                                                     , \nn
  \\[4pt]
\ell_{2}^{\star}\big((e,\om)\,,\,(E,{\mit\Omega})\big)&=\Big(\dd x^{\mu}\otimes \Tr  \big( \iota_\mu\dd \,
\bar{\sff}^{k}(e) \dwedge \bar{\sff}_{k}(E) + \iota_\mu\dd\,\bar{\sff}^k( \om) \dwedge \bar{\sff}_{k}({\mit\Omega}) \nn   
-\iota_\mu\bar{\sff}^k (e) \dwedge  \dd\, \bar{\sff}_{k}( E)  \nn \\ 
&\hspace{3cm} - \iota_\mu\bar{\sff}^{k}( \om )\dwedge \dd\,
  \bar{\sff}_k ({\mit\Omega}) \big)\, ,\,  \sfR_k
  (E)\wedge_\star \sfR^k (e) - \omega \wedge_\star {\mit\Omega}
  \Big)  \ . \nn 
\end{align}

In \eqref{eq:braided3dbrackets}, the notation $\rho\star e := \rho\,\triangleright_\star e=\rho^a{}_b\,\star e^b \, {\tt E}_a$ is shorthand for the left braided fundamental representation of $\rho\in\Omega^0(M,\aso(3))$ on $e\in\Omega^1(M,\FR^3)$, as a special instance of the general definition in \eqref{eq:lgtphi}; similarly the action of $\om\in\Omega^1(M,\aso(3))$ on $e\in \Omega^1(M,\FR^3)$ is $\om\wedge_\star e := \om\triangleright_\star e$, where here the star-exterior product includes matrix multiplication in the fundamental representation of $\aso(3)$. For a vector field $\xi\in\Gamma(TM)$, the star-Lie derivative $\LL_\xi^\star$ is defined in Section~\ref{sec:braideddiff}. The twisted product $\dwedge_\star$ is the tensor product of the star-exterior product of differential forms on $M$ with the undeformed exterior product of multivectors on the vector space $\FR^3$; in particular $e_1\dwedge_\star e_2 = \sfR_k(e_2)\dwedge_\star\sfR^k(e_1)$. Note that in the fourth and sixth brackets, for generic twists it is not possible to write the first slots entirely in terms of twisted products $\dwedge_\star$ or $\otimes_\star$, as the various terms do not commute appropriately; we will come back this point below.

In the following we use some calculational identities which are deformations of those used in~\cite{ECPLinfty}. One has
\begin{align*}
(\rho\star e)\dwedge_\star E & = -\rho\dwedge_\star \big(\sfR_k (E) \wedge_\star \sfR^k (e)\big) \ , \nn \\[4pt]
(\rho\star e)\dwedge_\star e\dwedge_\star e & = -\rho\dwedge_\star \big(\sfR_k (e\dwedge_\star e) \wedge_\star 
\sfR^k (e) \big) \ .
\end{align*}
In the classical case, the second line vanishes, but in the noncommutative theory 
\begin{align}\label{eq:eeenon0}
(e\dwedge_{\star} e )\wedge_\star e \neq 0 
\end{align}
because $1$-forms no longer strictly anticommute with each other.

\subsubsection*{Braided Lie algebra of gauge symmetries}
\label{sec:3dgauge}

From \eqref{eq:LgtL}, the left braided gauge transformation of a field $A=(e,\om)$ by a gauge parameter $\lambda=(\xi,\rho)$ is encoded by
\begin{align*}
\delta_{(\xi,\rho)}^{\star\lact} (e,\om) &= \ell_{1}^{\star}(\xi,\rho) + \ell_{2}^{\star}\big((\xi,\rho)\,,\,(e,\om)\big) = \big(-\rho\star e + \LL_\xi^\star e\,,\, \dd\rho-[\rho,\omega]_{\aso(3)}^\star+\LL_\xi^\star\om\big) \ .
\end{align*}
This combines the braided gauge transformations from Examples~\ref{ex:protoex} and~\ref{ex:protoLinftygauge}, for $\frg=\aso(3)$, $W=\FR^3$ and $p=1$, with the braided diffeomorphisms from Section~\ref{sec:braideddiff} which abide by the rules of braided gauge transformations spelled out in Section~\ref{sec:braidedgauge}.\footnote{Analogously to our discussion of left and right braided gauge transformations from Section~\ref{sec:braidedgauge}, there are correspondingly a left and a right star-Lie derivative. The star-Lie derivative appearing here acts from the left, but for simplicity we do not indicate this explicitly in the notation $\LL_\xi^\star$.}

According to \eqref{eq:braidedclosure}, these two individual braided Lie algebras of gauge symmetries combine into the single braided Lie algebra with bracket given by
\begin{align*}
-\ell_2^\star\big((\xi_1,\rho_1)\,,\,(\xi_2,\rho_2)\big) = \big(-[\xi_1,\xi_2]_\frv^\star\,,\,[\rho_1,\rho_2]_{\aso(3)}^\star + \LL_{\sfR_k(\xi_2)}^\star\sfR^k(\rho_1) - \LL_{\xi_1}^\star\rho_2\big) \ .
\end{align*}
We identify this combined gauge algebra as the braided semi-direct product 
\begin{align*}
\Gamma(TM)[[\hbar]]\ltimes_\star\Omega^0\big(M,\aso(3)\big)[[\hbar]]
\end{align*}
of the braided Lie algebra of vector fields from Section~\ref{sec:braideddiff}, and the braided Lie algebra from Examples~\ref{ex:protoex} and~\ref{ex:protoLinftygauge} with $\frg=\aso(3)$, $W=\FR^3$ and $p=1$. Thus our twisting of both the diffeomorphism and local rotational symmetries in a democratic manner preserves the semi-direct product structure of the classical Lie algebra \eqref{ClassicalSymmetries}, and simultaneously avoids the introduction of new degrees of freedom into the theory, in contrast to some previous approaches to noncommutative gravity (cf.\ Section~\ref{sec:braidedvsstargauge}).

\subsubsection*{Braided curvature and torsion}
\label{sec:3deom}

From \eqref{eq:braidedeom}, the field equations $F^\star_{(e,\om)}=(0,0)$ for the three-dimensional braided ECP theory are encoded through
\begin{align*}
F^\star_{(e,\om)} = \ell_1^\star(e,\om) - \tfrac12\,\ell_2^\star\big((e,\om)\,,\,(e,\om)\big) =: (F^\star_e,F^\star_\om) \ .
\end{align*}
The first entry is given by
\begin{align}\label{eq:Fe3d}
\begin{tabular}{|l|}\hline\\
$\displaystyle
F^\star_e = R^\star + \Lambda\,e\dwedge_\star e
$\\\\\hline\end{tabular}
\end{align}
where
\begin{align}\label{eq:curvature}
R^\star = \dd \om +\tfrac12\,[\om,\om]_{\aso(3)}^\star = \dd \om + \tfrac12\,\big(\om\wedge_\star\om + \sfR_k(\om)\wedge_\star\sfR^k(\om)\big)
\end{align}
is the braided curvature $2$-form of the spin connection $\om$, as an example of braided curvature from Section~\ref{sec:braidedkin}. The second entry is given by
\begin{align*}
F^\star_\om = \dd e + \tfrac12\,\om\wedge_\star e + \tfrac12\,\sfR_k(\om)\wedge_\star\sfR^k(e) \ .
\end{align*}
To understand the meaning of this expression, recall from Section~\ref{sec:braidedkin} that there are two separate definitions of braided covariant derivative, and hence there are correspondingly two separate consistent definitions of torsion: A \emph{left} torsion $2$-form
\begin{align}\label{eq:torsionleft}
T^\star_\lact := \dd_{\star\lact}^\om e = \dd e + \om\wedge_\star e \ ,
\end{align}
and a \emph{right} torsion $2$-form
\begin{align}\label{eq:torsionright}
T^\star_\ract := \dd_{\star\ract}^\om e = \dd e + \sfR_{k}(\om)\wedge_\star \sfR^k(e) \ .
\end{align}
Hence the corresponding field equation is given by the symmetrized torsion
\begin{align}\label{eq:Fom3d}
\begin{tabular}{|l|}\hline\\
$\displaystyle
F^\star_\om = \tfrac12\,(T^\star_\lact + T^\star_\ract)
$\\\\\hline\end{tabular}
\end{align}

According to \eqref{eq:leftcovcond}, the braided covariance of the field equations is encoded by
\begin{align}\label{eq:covell2}
\delta_{(\xi,\rho)}^{\star\lact}F^\star_{(e,\om)} = \big(\delta_{(\xi,\rho)}^{\star\lact}F^\star_e,\delta_{(\xi,\rho)}^{\star\lact}F^\star_\om\big) = \ell_2^\star\big((\xi,\rho)\,,\,(F^\star_e,F^\star_\om)\big) \ ,
\end{align}
where
\begin{align}\label{eq:covcondFeom}
\delta_{(\xi,\rho)}^{\star\lact}F^\star_e = -[\rho,F_e^\star]_{\aso(3)}^\star + \LL_\xi^\star F_e^\star \qquad \mbox{and} \qquad \delta_{(\xi,\rho)}^{\star\lact}F^\star_\om =-\rho\star F_\om^\star + \LL_\xi^\star F_\om^\star \ .
\end{align}
This is in complete harmony with the discussion from Section~\ref{sec:braidedkin}, which shows that the braided curvature is covariant, and that the left and right torsions are both covariant with respect to left braided gauge transformations. 

Thus the geometrical meaning of the braided field equations is just a deformation of that from the classical case. The equation $F_\om^\star=0$ is simply the balanced torsion-free condition\footnote{This improves the interpretation of the torsion equation from~\cite{Ciric:2020eab}, where only the left torsion $2$-form was considered.} $T^\star_\lact+T^\star_\ract=0$, and it would be interesting to see if this can be solved uniquely and algebraically for the spin connection $\om$ in terms of the coframe $e$, as in the classical case where it would yield the Levi--Civita connection. Conditions for retrieving unique connections in similar noncommutative geometry settings have been suggested before (see e.g.~\cite{Aschieri:2020ifa,BeggsMajidbook}), and it would be interesting to compare these conditions with our balanced torsion-free condition. The equation $F_e^\star=0$ is the braided analog of the classical Einstein equation with cosmological constant in our formalism (before inserting the Levi--Civita connection). We stress that our field equations are different from those previously suggested for the first order formulation of noncommutative gravity~\cite{AschCast}, already at the kinematic level: in our approach the internal degrees of freedom are braided and hence there are no new dynamical fields introduced from having to extend the underlying Lie algebra $\aso(3)$. 

Our remarks from Section~\ref{sec:braidedEOM} concerning the space of classical solutions of course also apply here. In particular, a braided diffeomorphism $\delta^{\star\lact}_{(\xi,0)}(e,\om)$ can map a classical solution into a field configuration which no longer solves the field equations; for example, for $\Lambda=0$, a {braided} flat noncommutative spacetime can be braided diffeomorphic to a {braided} curved noncommutative spacetime. This peculiar feature appears not only in our formalism, but also in the many earlier treatments of noncommutative gravity which used an underlying symmetry principle based on twisted diffeomorphisms (cf. Section~\ref{sec:NCGR}), which was discussed in Section~\ref{sec:braideddiff}. As far as we are aware, this salient point has not been mentioned or appreciated before. According to the general framework of braided field theory that we have developed in the present paper, the meaning of the local symmetries in noncommutative gravity are captured by braided Noether identities.

\subsubsection*{Braided Noether identities}
\label{sec:3dNoether}

By \eqref{braidedNoether}, the gauge redundancies of braided gravity in three dimensions are encoded by the braided Noether identities
\begin{align}\label{eq:braidedNoether3d}
\dsf_{(e,\om)}^\star F^\star_{(e,\om)} &= \ell_1^\star\big(F_{e}^\star,F_\om^\star\big) - \frac12\,\Big(\ell_2^\star\big((e,\om)\,,\,(F_e^\star,F_\om^\star)\big) - \ell_2^\star\big((F_e^\star,F_\om^\star)\,,\,(e,\om)\big)\Big) \\ & \hspace{6cm} + \frac14\,\ell_2^\star\big(\sfR_k(e,\om) \, , \, \ell_2^\star(\sfR^k(e,\om),(e,\om))\big) \ = \ (0,0) \ , \nn
\end{align}
where $\sfR_k(e,\om):=\big(\sfR_k(e),\sfR_k(\om)\big)$, and we used the slightly simpler form \eqref{eq:leftNoetherprelim}.
They lead to a pair of differential identities among the field equations \eqref{eq:Fe3d} and \eqref{eq:Fom3d}.

The Noether identity corresponding to braided local $\aso(3)$ rotations is given by the second slot of \eqref{eq:braidedNoether3d}. Compared to the classical case~\cite{ECPLinfty}, it now acquires a contribution from the cosmological constant due to the braiding, in addition to the usual inhomogeneous terms:
\begin{align}\label{eq:Noetherso3}
\tfrac12\,\big(\dd_{\star\lact}^\om F^\star_{\om} + \dd_{\star\ract}^\om & F_\om^\star- \sfR_k( F^\star_{e}) \wedge_{\star} \sfR^k(e) - F^\star_{e} \wedge_{\star} e \big) \nn \\
&+\tfrac{1}{4}\, \big( -[\om,\sfR_{k}(\om)]_{\aso(3)}^{\star} \wedge_\star
  \sfR^k( e) +\sfR_k( \om) \wedge_\star \sfR^k(\om)
  \wedge_\star e + \sfR_k (\om \wedge_\star \om) \wedge_\star
  \sfR^k( e) \nn \\
&\hspace{2cm} -2\,\Lambda\, (e\dwedge_{\star} e )\wedge_\star e -2\,\Lambda\, \sfR_k(e\dwedge_{\star} e 
)\wedge_\star \sfR^k( e)\big) \ =  \ 0 \ .
\end{align}
The last two terms here vanish in the classical case, but they are generally non-vanishing in the noncommutative case due to \eqref{eq:eeenon0}. This identity can also be derived directly from the braided deformations of the first Bianchi identity:
\begin{align*}
\dd_{\star\lact}^\om T_\lact^\star &= R^\star\wedge_\star e + \tfrac12\,\omega\wedge_\star\omega\wedge_\star e - \tfrac12\,\sfR_k(\om)\wedge_\star\sfR^k(\om)\wedge_\star e \ , \\[4pt]
\dd_{\star\ract}^\om T_\ract^\star &= \sfR_k(R^\star)\wedge_\star\sfR^k(e) - \tfrac12\,\sfR_k(\om)\wedge_\star\sfR_l(\om)\wedge_\star\sfR^l\sfR^k(e) + \tfrac12\,\sfR_m\sfR_l(\omega)\wedge_\star\sfR^m\sfR_k(\om)\wedge_\star\sfR^l\sfR^k(e) \ ,
\end{align*}
which follow easily by taking the exterior derivatives of \eqref{eq:torsionleft} and \eqref{eq:torsionright}, respectively.
In the classical case, the last two terms on the right-hand sides of both equations cancel and we recover the expected identity $\dd^\om T=R\wedge e$, which is equivalent to \eqref{eq:Noetherso3} in that case.

The Noether identity corresponding to braided diffeomorphisms is given by the first slot of \eqref{eq:braidedNoether3d}. Given the form of the classical map $\ell_2:V_1\otimes V_2\to V_3$ from \eqref{eq:3dbrackets}, we have to use the defining expression $\ell^\star_{2}\big((e,
\om), (\CF_{e},\CF_\om)\big) := \ell_{2}\big( \bar{\sff}^k(e,\om), \bar{\sff}_k(\CF_{e},\CF_{\om})\big)$ directly in terms of the undeformed classical bracket. Expanding \eqref{eq:braidedNoether3d} using classical antisymmetry of brackets in the third line, we find
\begin{align}\label{eq:diffNoether3d}
\dd x^\mu \otimes \Tr\big(N^\star_\mu(F_e^\star,F_\om^\star,e,\om)\big) \ = \ 0 \ ,
\end{align}
where the $3$-form $N^\star_\mu(F_e^\star,F_\om^\star,e,\om)\in\Omega^3(M,\midwedge^3\FR^3)[[\hbar]]$ is given by 
\begin{align*}
N^\star_\mu &= \frac12\,\big(\iota_\mu  \dd \,\bar{\sff}^k (e) \dwedge \bar{\sff}_k (F^\star_{e} )+ \iota_\mu\dd \,\bar{\sff}^k (\om )\dwedge \bar{\sff}_k(F^\star_{\om}) - \iota_\mu\bar{\sff}^k(e)\dwedge \dd \,\bar{\sff}_k (F^\star_{e}) - \iota_\mu \bar{\sff}^k (\om) \dwedge \dd \,\bar{\sff}_k (F^\star_{\om}) \nn \\
& \hspace{1cm} + \iota_\mu  \dd \,\bar{\sff}_k (e) \dwedge \bar{\sff}^k (F^\star_{e}) + \iota_\mu\dd \,\bar{\sff}_k (\om )\dwedge \bar{\sff}^k(F^\star_{\om}) - \iota_\mu\bar{\sff}_k (e)\dwedge \dd \,\bar{\sff}^k (F^\star_{e}) - \iota_\mu \bar{\sff}_k (\om) \dwedge \dd \,\bar{\sff}^k( F^\star_{\om}) \big) \\
& \quad \, - \frac14\,\Big( \iota_\mu \dd \,\bar{\sff}^k \sfR_l (e)\dwedge \bar{\sff}_k\big([\sfR^l (\om),\om]_{\aso(3)}^\star + 2\,\Lambda\,\sfR^l(e)\dwedge_\star e\big) \\
& \hspace{2.5cm} + \iota_\mu \dd \,\bar{\sff}^k \sfR_l (\om)\dwedge \bar{\sff}_k \big(\sfR^l(\om)\wedge_\star e - \sfR_m (\om) \wedge_\star \sfR^m \sfR^l (e) \big)\nn \\
& \hspace{4cm} -\iota_\mu \bar{\sff}^k \sfR_l (e)\dwedge \dd \,\bar{\sff}_k\big([\sfR^l (\om),\om]_{\aso(3)}^\star + 2\,\Lambda\,\sfR^l(e)\dwedge_\star e\big) \\\
& \hspace{5.5cm} -\iota_\mu \bar{\sff}^k \sfR_l (\om)\dwedge \dd \,\bar{\sff}_k \big(\sfR^l(\om)\wedge_\star e -  \sfR_m (\om) \wedge_\star \sfR^m \sfR^l (e) \big) \Big) \ .
\end{align*}
This expression does not immediately involve the star-contraction defined by $\iota_{\xi}^{\star} \alpha = \iota_{\bar{\sff}^{k}(\xi)} \bar{\sff}_{k}(\alpha)$ for $\xi\in\Gamma(TM)$ and $\alpha\in \Omega^{\bullet}(M)$. Nevertheless, the expression \eqref{eq:diffNoether3d} is globally well-defined on $M$ as an element of $\Omega^1(M)\otimes \Omega^3(M)[[\hbar]]$. In some special cases where the legs of the twist $\CF$ commute with the local basis vector fields,  we may write the braided diffeomorphism Noether identity in a slightly more compact form; for example, this holds for abelian twists locally around every point of $M$ outside a set of measure zero~\cite{Aschieri:2009qh,SchenkelThesis}.

\begin{example}\label{ex:MoyalNoether3d}
Consider the Moyal--Weyl twist $\CF_\theta = \sff^k_\theta\otimes\sff_{\theta k}$ on $M=\FR^3$ from Examples~\ref{ex:MoyalWeyltwist} and~\ref{ex:MoyalWeylstar}. Then $\iota_\mu \circ \bar{\sff}_\theta^k = \bar{\sff}_\theta^k \circ \iota_\mu$ since $\iota_\mu\circ \LL_{\partial_\nu}= \LL_{\partial_\nu}\circ\iota_\mu + \iota_{[\partial_\mu,\partial_\nu]_\frv}= \LL_{\partial_\nu}\circ\iota_\mu$ by the Cartan structure equations, and $\iota_\mu^{\star_\theta}=\iota_\mu$ since $\LL_{\partial_\nu}\partial_\mu=[\partial_\nu,\partial_\mu]_\frv=0$ in a holonomic basis. In this case we have
\begin{align*}
\iota_\mu  \dd \,\bar{\sff}_\theta^k( e) \dwedge \bar{\sff}_{\theta k} (F^{\star_\theta}_{e}) = \iota^{\star_\theta}_{\mu}\dd e \dwedge_{\star_\theta} F^{\star_\theta}_{e}
\end{align*}
and
\begin{align*}
\iota_\mu^{\star_\theta} \dd \,\bar{\sff}_{\theta k} (e) \dwedge \bar{\sff}_\theta^k (F^{\star_\theta}_{e}) = \bar{\sff}_\theta^k (F^{\star_\theta}_{e}) \dwedge \bar{\sff}_{\theta k} (\iota_\mu^{\star_\theta} \dd e )= F^{\star_\theta}_{e}\dwedge_{\star_\theta} \iota^{\star_\theta}_\mu \dd e \ ,
\end{align*}
with analogous expressions for the rest of the terms. Since also $\LL_{\partial_\nu}\dd x^\mu=0$, it follows that the expression for the braided diffeomorphism Noether identity \eqref{eq:diffNoether3d} simplifies slightly to an expression in $\Omega^1(\FR^3)\otimes_{\star_\theta}\Omega^3(\FR^3)[[\hbar]]$ given by
\begin{align}\label{eq:MoyalNoether3d}
\dd x^\mu\otimes_{\star_\theta}\Tr\big(N_\mu^{\star_\theta}(F_e^\star,F_\om^\star,e,\om)\big) = 0
\end{align}
with
\begin{align*}
N_\mu^{\star_\theta} &= \frac12\, \big( \iota^{\star_\theta}_\mu  \dd  e \dwedge_{\star_\theta} F^{\star_\theta}_{e} + \iota^{\star_\theta}_\mu\dd \om \dwedge_{\star_\theta} F^{\star_\theta}_{\om} - \iota^{\star_\theta}_\mu e\dwedge_{\star_\theta} \dd F^{\star_\theta}_{e} - \iota^{\star_\theta}_\mu \om \dwedge_{\star_\theta} \dd F^{\star_\theta}_{\om} \nn \\
	& \hspace{1cm} + F^{\star_\theta}_{e}\dwedge_{\star_\theta} \iota^{\star_\theta}_\mu  \dd e + F^{\star_\theta}_{\om} \dwedge_{\star_\theta} \iota^{\star_\theta}_\mu\dd \om  - \dd F^{\star_\theta}_{e} \dwedge_{\star_\theta} \iota^{\star_\theta}_\mu e  - \dd F^{\star_\theta}_{\om} \dwedge_{\star_\theta} \iota^{\star_\theta}_\mu\om \big) \\
& \quad \, -\frac14 \, \Big( \iota^{\star_\theta}_\mu \dd \sfR_l( e)\dwedge_{\star_\theta} [\sfR^l (\om),\om]_{\aso(3)}^{\star_\theta} + 2\,\Lambda\,e\dwedge_{\star_\theta}\iota_\mu^{\star_\theta}\dd e\dwedge_{\star_\theta}e \\
& \hspace{2cm} + \iota^{\star_\theta}_\mu \dd \sfR_l( \om)\dwedge_{\star_\theta} \big(\sfR^l(\om)\wedge_{\star_\theta} e - \sfR_k (\om) \wedge_{\star_\theta} \sfR^k \sfR^l (e) \big)\nn \\
	& \hspace{4cm} -\iota^{\star_\theta}_\mu \sfR_l (e)\dwedge_{\star_\theta} \dd [\sfR^l (\om),\om]_{\aso(3)}^{\star_\theta} + 2\,\Lambda\,e\dwedge_{\star_\theta}\iota_\mu^{\star_\theta}\dd e\dwedge_{\star_\theta}e \\
	& \hspace{6cm} -\iota^{\star_\theta}_\mu \sfR_l (\om)\dwedge_{\star_\theta} \dd  \big(\sfR^l(\om)\wedge_{\star_\theta} e -  \sfR_k (\om) \wedge_{\star_\theta} \sfR^k \sfR^l (e) \big) \Big) \ .
\end{align*}
\end{example}

\subsubsection*{Action functional}
\label{sec:3daction}

The variational principle for three-dimensional braided gravity follows, similarly to the braided Chern--Simons theory of Section~\ref{sec:braidedCStheory}, from choosing a Hermitian Drinfel'd twist $\CF\in U\frv[[\hbar]\otimes U\frv[[\hbar]]$ which is compatible with the cyclic structure of the classical ECP theory defined by \eqref{eq:ECPpairing} and~\eqref{eq:ECPpairing2}. Using Proposition~\ref{prop:compatiblestrict}, we can then turn $(V[[\hbar]],\ell_1^\star,\ell_2^\star)$ into a strictly cyclic braided $L_\infty$-algebra with cyclic pairing defined by
\begin{align}\label{eq:braidedpairing1}
\langle (e,\om) \,,\, (E,{\mit\Omega}) \rangle_{\star}:=
\int_{M}\, \Tr \big(e\dwedge_\star E+ {\mit\Omega} \dwedge_\star \om \big) \ ,
\end{align}
on $V_{1}\otimes V_{2}$ and
 \begin{align}\label{eq:braidedpairing2}
 \langle(\xi,\rho)\,,\,({\CX},{\CP})\rangle_{\star} := \int_M\,
 \iota_{\xi}^{
\star}{\CX} + \int_M\, \Tr\big(\rho\dwedge_\star{\CP}\big) \ ,
 \end{align}
 on $V_{0}\otimes V_{3}$.
 
According to \eqref{eq:braidedaction}, the action functional is then given by 
 \begin{align*}
 S_{\star}(e,\om)&= \tfrac{1}{2}\, \big\langle(e,\om)\,,\, \ell_{1}^{\star}(e,\om) \big\rangle_\star -\tfrac{1}{6}\,\big\langle(e,\om)\,,\,\ell_{2}^{\star}\big((e,\om)\,,\,(e,\om)\big)\big\rangle_\star \\[4pt]
 &= \int_{M}\,\Tr\Big(e \dwedge_{\star} R^\star + \frac\Lambda3\, e\dwedge_\star e\dwedge_\star e\Big) = \int_M\, \epsi_{abc}\, \Big( e^{a}\wedge_{\star} R^{\star bc} + 
\frac{\Lambda}{3}\,e^{a}\wedge_{\star}e^{b}\wedge_{\star}e^{c} \Big) \ .
 \end{align*}
In the second line we used invariance of a top-vector under the left and right braided actions by $\aso(3)$-valued $1$-forms, analogously to the classical case~\cite{ECPLinfty} order by order in the formal power series expansion of the twist in $\hbar$; for example
\begin{align*}
0=(\om\dwedge_\star e)\,_{\star\!\!}\triangleleft \om = -\om \dwedge_\star \big(\sfR_k( \om) \wedge_{\star} \sfR^k(e)\big) + [\sfR_k\sfR_l( \om), \sfR^k(\om)]_{\aso(3)}^\star \dwedge_{\star} \sfR^l(e) \ ,
\end{align*}
and so on. This is just the expected deformation of the classical action functional \eqref{eq:ECP3daction}; it is the unique action functional constructed from the braided pairing, and the corresponding polynomial orders in the coframe field and the curvature, with the correct classical form when $\RR=1\otimes1$.

Using the $\RR$-matrix identities \eqref{eq:Rmatrixidswn}, the cyclicity of integration \eqref{eq:intcyclic} with respect to the star-product $\wedge_\star$, and the invariance of a top-vector under braided $\aso(3)$ rotations, one can directly check that varying this action functional with respect to $e$ gives the field equation $F_e^\star=0$ defined in \eqref{eq:Fe3d}, while varying with respect to $\om$ gives the field equation $F_\om^\star=0$ defined in \eqref{eq:Fom3d}, as was guaranteed by the general theory of Section~\ref{sec:braidedaction}. Also in agreement with the general theory, one checks invariance of the action functional \smash{$\delta_{(0,\rho)}^{\star\lact}S_\star(e,\om)=0$} under braided $\aso(3)$ rotations using braided covariance \smash{$\delta_\rho^{\star\lact}R^\star = -[\rho,R^\star]_{\aso(3)}^\star$} and $\aso(3)$ top-vector invariance such as
\begin{align*}
\rho\star(e\dwedge_\star R^\star) = \bar\sff^k(\rho)\cdot\bar\sff_k(e\dwedge_\star R^\star)=0 \ ,
\end{align*}
as well as invariance \smash{$\delta_{(\xi,0)}^{\star\lact}S_\star(e,\om)=0$} under braided diffeomorphisms using the classical Cartan formula for the Lie derivative and Stokes' theorem order by order in $\hbar$. These two braided gauge invariances further imply respectively the braided Noether identities \eqref{eq:Noetherso3} along the lines discussed in Section~\ref{sec:braidedaction}, though again the most effective and concise way to do this is through the braided $L_\infty$-algebra framework of this section.

\subsection{Braided gravity in four dimensions}
\label{sec:NCgravity4D}

Both the braided Chern--Simons theory from Section~\ref{sec:braidedCStheory} and the three-dimensional braided ECP theory from Section~\ref{sec:NCgravity3D} are based on differential (graded) braided Lie algebras, while their field equations and variational principle could have been more or less guessed from the outset by simply replacing ordinary products of fields with star-products. This does not preclude, however, the different kinematical sector of these theories due to the braided Lie brackets, as well as the novel general nature of classical solutions in braided field theory and the corresponding modified forms of gauge redundancies through the braided Noether identities. We will now discuss a highly non-trivial example which is based on an $L_\infty$-algebra which is not simply a differential graded Lie algebra. It does not follow the naive deformations of the three-dimensional theories studied thus far, and it illustrates the novelty of our approach in generating new noncommutative field theories with interesting symmetry structures. Indeed, we shall now derive a new noncommutative theory of gravity in four dimensions by applying the braided $L_\infty$-algebra framework to the classical four-dimensional ECP theory.

\subsubsection*{Braided $\boldsymbol{L_\infty}$-algebra}

Let $M$ be an oriented parallelizable four-manifold. The $4$-term $L_\infty$-algebra underlying the classical Einstein--Cartan--Palatini formulation of general relativity on $M$ with cosmological constant $\Lambda\in\FR$ was developed in detail by~\cite{ECPLinfty} for an arbitrary signature pseudo-Riemannian structure on $M$; for definiteness, we focus on the case of Euclidean signature, but everything we do holds identically for arbitrary signature $(p,q)$ (with $p+q=4$): we simply replace the Lie algebra $\mathfrak{so}(4)$ and its fundamental representation $\FR^4$ by their indefinite versions $\mathfrak{so}(p,q)$ and $\FR^{p,q}$. As in the three-dimensional case, this defines an $L_\infty$-algebra in the category $\CCM^\sharp$ of $\RZ$-graded $U\frv$-modules. Then the underlying graded vector space is given by
\begin{align*}
V:= V_{0} \oplus V_{1} \oplus V_{2} \oplus V_3
\end{align*}
where 
\begin{align*} 
V_{0}&=\Gamma(TM)\times \Omega^{0}\big(M,\mathfrak{so}(4)\big) \ ,
       \nn \\[4pt] 
V_{1}&= \Omega^{1}(M,\FR^{4}) \times
       \Omega^{1}\big(M,\mathfrak{so}(4) \big) \ , \\[4pt]
V_{2}&=\Omega^{3}\big(M,\midwedge^{3}\FR^{4}\big) \times
       \Omega^{3}\big(M,\midwedge^{2}\FR^{4}\big) \ , \\[4pt]
V_3&=\Omega^1\big(M,\Omega^4(M)\big) \times \Omega^4\big(M,\midwedge^2\FR^{4}\big)
     \ .
\nn
\end{align*} 
We use the same conventions and notation for fields as we did in the case of three-dimensional gravity from Sections~\ref{sec:ECP3d} and~\ref{sec:NCgravity3D}.
Following the prescription of Section~\ref{sec:Linftytwist}, we twist deform the classical brackets written down in~\cite[Section~8.2]{ECPLinfty}, which leads to a $4$-term braided $L_\infty$-algebra $(V[[\hbar]],\ell_1^\star,\ell_2^\star,\ell_3^\star)$ with the following non-vanishing braided brackets.

The differential $\ell_1^\star$ is given by
\begin{align}\label{eq:ell14d}
\ell_{1}^\star(\xi,\rho)=(0,\dd\rho) \ , \quad
	\ell_{1}^\star(e,\omega)=(0,0) \qquad \mbox{and} \qquad \ell_{1}^\star(E,{\mit\Omega})=(0,-\dd {\mit\Omega}) \ .
\end{align}
The $2$-bracket $\ell_2^\star$ is defined by
\begin{align}\label{eq:ell24d}
\ell_{2}^\star\big((\xi_{1},\rho_{1})\,,\,(\xi_{2},\rho_{2})\big)&=
	\big([\xi_{1},\xi_{2}]_\frv^\star\,,\,-[\rho_{1},\rho_{2}]_{\aso(4)}^\star+\LL_{\xi_{1}}^{\star}\rho_{2}
	- \LL_{\sfR_k(\xi_{2})}^{\star}\sfR^k(\rho_{1}) \big) \ , \nn \\[4pt]
\ell^\star_{2}\big((\xi,\rho)\,,\,(e,\omega)\big)&=\big(-\rho \star e
	+\LL_{\xi}^{\star}e\,,\, -[\rho,\omega]_{\aso(4)}^\star+\LL_{\xi}^{\star}\om\big) \ , \nn \\[4pt]
	\ell_{2}^\star\big((\xi,\rho)\,,\,(E,{\mit\Omega}) \big)&=\big(-
	\rho \star E+\LL_{\xi}^{\star}E \,,\, -[\rho
	, {\mit\Omega}]_{\aso(4)}^\star+\LL_{\xi}^{\star}{\mit\Omega}\big) \ , \nn \\[4pt]
\ell_{2}^\star\big((\xi,\rho)\,,\,({\CX},{\CP})\big)&=
	\Big(\dd x^\mu\otimes\Tr\big(\iota_\mu \dd \,\bar{\sff}^k(\rho) \dwedge
	\bar{\sff}_k({\CP})\big) +\LL_{\xi}^{\star}{\CX}\,,\,-[\rho
	, {\CP}]_{\aso(4)}^{\star} +\LL_{\xi}^{\star} {\CP}\Big) \ , \\[4pt]
\ell_{2}^\star\big((e_{1},\omega_{1})\,,\,(e_{2},\omega_{2})\big)&=-\big(e_{1} \dwedge_\star \dd \omega_{2}
	+ \sfR_k (e_{2}) \dwedge_\star \dd \sfR^k (\omega_{1}) \,,\, e_{1}
	\dwedge_\star \dd e_{2} +\sfR_k (e_{2}) \dwedge_\star \dd \sfR^k (e_{1})\big) \ , \nn \\[4pt]
\ell_{2}^\star\big((e,\om)\,,\,(E,{\mit\Omega})\big)&=\Big(\dd x^\mu\otimes\Tr \big( \iota_\mu \dd\,\bar{\sff}^k( e) \dwedge \bar{\sff}_k (E) - \iota_\mu \dd\,\bar{\sff}^k( \om) \dwedge \bar{\sff}_k ({\mit\Omega}) \nn \\
& \hspace{4cm}  -\iota_\mu \bar{\sff}^k (e) \dwedge \dd\, \bar{\sff}_k( E) + \iota_\mu \bar{\sff}^k(\om) \dwedge \dd\, \bar{\sff}_k ({\mit\Omega}) \big) \ , \nn \\ 
	& \hspace{6cm} \tfrac32 \, \sfR_k(E) \wedge_\star \sfR^k(e) + [\omega, {\mit\Omega}]_{\aso(4)}^\star\Big) \ , \nn
\end{align}
with all other $2$-brackets equal to zero. Finally, there is a single non-zero $3$-bracket $\ell_3^\star$ defined on fields by
\begin{align}\label{eq:ell34d}
	& \ell_{3}^\star\big((e_{1},\omega_{1})\,,\,(e_{2},\omega_{2})\,,\,(e_{3},\omega_{3})
	\big) \nn \\[4pt]
	& \hspace{1cm} =-\Big( e_{1}\dwedge_\star [\omega_{2} , \omega_{3}]_{\aso(4)}^\star + \sfR_k (e_{2}) \dwedge_\star [\sfR^k (\omega_{1}) , \omega_{3}]_{\aso(4)}^\star + \sfR_k(  e_{3}) \dwedge_\star \sfR^k ([\omega_{1} , \omega_{2}]_{\aso(4)}^\star) \nn \\
& \hspace{12cm} +  6\, \Lambda\, e_{1}\dwedge_\star e_{2}\dwedge_\star
	e_{3}  \ , \nn \\
	& \hspace{2.5cm} e_{1} \dwedge_\star (\omega_{2} \wedge_\star e_{3}) 
	+e_{1} \dwedge_\star \big(\sfR_k( \omega_{3}) \wedge_\star \sfR^k( e_{2})\big) \\
& \hspace{3.5cm} +  \sfR_k (e_{2})\dwedge_\star \big( \sfR^k(\omega_{1}) \wedge_\star e_{3}\big) 
	+\sfR_k( e_{2}) \dwedge_\star \big(\sfR_l(\omega_{3}) \wedge_\star \sfR^l \sfR^k( e_{1})\big)  
	\nn \\ & \hspace{4.5cm} + \sfR_k (e_{3})\dwedge_\star \sfR^k (\omega_{1} \wedge_\star e_{2}) + \sfR_{k}\sfR_l( e_{3}) \dwedge_\star \big(\sfR_m \sfR^l(\omega_{2} ) \wedge_\star \sfR^m \sfR^k( e_{1})\big)  \Big) \ . \nn
\end{align} 
Note that the differential of a generic field $A=(e,\om)\in V_1$ vanishes, which will lead to some simplifications in the calculations involving the braided $L_\infty$-algebra brackets below.

The derivation of the braided gauge symmetries of noncommutative gravity in four dimensions is identical to that in three dimensions which was presented in Section~\ref{sec:NCgravity3D}, and so will not be repeated here; the only change is that the Lie algebra of rotations $\aso(3)$ is replaced with its counterpart $\aso(4)$ appropriate to four dimensions. We shall therefore dive straight away into a description of the novel dynamics of the braided gravity theory, where we will uncover many unexpected and surprising features.

\subsubsection*{Field equations}

By \eqref{eq:braidedeom}, the field equations $F_{(e,\om)}^\star=(0,0)$ of the braided ECP theory in four dimensions are given by
\begin{align*}
F^\star_{(e,\om)} = -\tfrac12\,\,\ell_2^\star\big((e,\om)\,,\,(e,\om)\big) - \tfrac16\,\ell_3^\star\big((e,\om)\,,\,(e,\om)\,,\,(e,\om)\big) =: (F_e^\star,F_\om^\star) \ .
\end{align*}
Expanding the brackets, the components are given by
\begin{align}\label{eq:Feell4d}
F^\star_{e} &= \tfrac{1}{2}\,\big( e \dwedge_\star \dd \om + \sfR_k( e)  \dwedge_\star \dd \sfR^k (\om) \big) + \Lambda\,e\dwedge_\star e\dwedge_\star e \nn \\
& \quad \, +\tfrac{1}{6}\, \big(e \dwedge_\star [\om,\om]_{\aso(4)}^\star +\sfR_k( e) \dwedge_\star [\sfR^k( \om),\om]_{\aso(4)}^\star + \sfR_k( e) \dwedge_\star \sfR^k( [\om,\om]_{\aso(4)}^\star)\big) \ ,
\end{align}
and
\begin{align}\label{eq:Fomell4d}
F^\star_{\om}&= \frac{1}{2}\,\big( e \dwedge_\star \dd e + \sfR_k( e) \dwedge_\star \dd \sfR^k( e) \big) \nn \\
& \quad \, +\frac{1}{6}\,\Big(e\dwedge_\star(\om \wedge_\star e) +  e \dwedge_\star \big(\sfR_k( \om) \wedge_\star \sfR^k( e)\big)+ \sfR_k( e) \dwedge_\star \big(\sfR^k( \om) \wedge_\star e\big) \nn \\
& \hspace{3cm} + \sfR_k( e) \dwedge_\star\big( \sfR_l( \om) \wedge_\star \sfR^l \sfR^k( e)\big) + \sfR_k( e) \dwedge_\star \sfR^k (\om \wedge_\star e) \nn \\
& \hspace{6cm} + \sfR_{k}\sfR_l( e) \dwedge_\star \big(\sfR_m \sfR^l(\omega) \wedge_\star \sfR^m \sfR^k( e)\big) \Big) \ .
\end{align}
According to \eqref{eq:leftcovcond}, the field equations transform under braided gauge transformations as in \eqref{eq:covell2}, with
\begin{align}\label{eq:covcondFeom4d}
\delta_{(\xi,\rho)}^{\star\lact}F^\star_e = -\rho\triangleright_\star F_e^\star + \LL_\xi^\star F_e^\star \qquad \mbox{and} \qquad \delta_{(\xi,\rho)}^{\star\lact}F^\star_\om =-\rho\triangleright_\star F_\om^\star + \LL_\xi^\star F_\om^\star \ .
\end{align}
The goal is now to simplify these field equations using various $\RR$-matrix identities to make the terms appear in manifestly covariant forms, as would be naively expected from the covariance property that is guaranteed by the general theory of Section~\ref{sec:braidedEOM}. This will introduce the curvature $2$-form $R^\star$ defined analogously to \eqref{eq:curvature}, as well as the left and right torsion $2$-forms $T_\lact^\star$ and $T_\ract^\star$ from \eqref{eq:torsionleft} and \eqref{eq:torsionright}. 

\subsubsection*{Braided torsion condition}

Let us start with the second component $F_\om^\star$. We will show that
\begin{align}\label{eq:Fom4d}
\begin{tabular}{|l|}\hline\\
$\displaystyle
F^\star_\om = \tfrac{1}{6}\,\big( e\dwedge_\star T_\lact^\star - T_{\ract}^\star \dwedge_\star e - \dd^{\om}_{\star\lact}(e\dwedge_\star e) - \dd^{\om}_{\star\ract}(e\dwedge_\star e) \big)
$\\\\\hline\end{tabular}
\end{align}
This expression is manifestly braided covariant, and it reduces exactly (with the correct numerical factor) to the expected torsion form $F_\om=e\dwedge T$ in the classical case.

For this, we use the Yang--Baxter equation \eqref{eq:YangBaxter} for the $\RR$-matrix to write the last term in \eqref{eq:Fomell4d} as
\begin{align*}
\sfR_m\sfR_l(e) \dwedge_\star\big( \sfR^m \sfR_k( \om) \wedge_\star \sfR^l \sfR^k(e)\big) =\sfR_m( e) \dwedge_\star \sfR^m \big(\sfR_k( \om) \wedge_\star \sfR^k( e)\big) \ .
\end{align*}
Using then braided antisymmetry of the $\dwedge_\star$-product on the latter, second and seventh terms, we get
\begin{align}\label{eq:Fom4dprelim}
F^\star_{\om}= \frac{1}{6}\,\Big(&3\, e\dwedge_\star \dd e - 3\, \dd e \dwedge_\star e + e\dwedge_\star (\om \wedge_\star e) + e \dwedge_\star\big( \sfR_k( \om) \wedge_\star \sfR^k(  e)\big) + \sfR_k(e) \dwedge_\star \big(\sfR^k( \om) \wedge_\star e\big)  \nn \\
	&+\sfR_k( e)\dwedge_\star \big(\sfR_l( \om ) \wedge_\star \sfR^l \sfR^k( e)\big) - (\om \wedge_\star e) \dwedge_\star e -\big( \sfR_k( \om) \wedge_\star \sfR^k( e)\big) \dwedge_\star e \Big) \ .
\end{align}
Now notice that
\begin{align}\label{eq:Fom4did1}
\sfR_k( e) \dwedge_\star (\sfR^k( \om) \wedge_\star e)= - \om \triangleright_\star (e\dwedge_\star e) + (\om \wedge_\star e) \dwedge_\star e
\end{align}
and 
\begin{align}\label{eq:Fom4did2}
\sfR_k( e)\dwedge_\star \big(\sfR_l( \om) \wedge_\star \sfR^l \sfR^k( e)\big) &= - \sfR_k( e)\dwedge_\star \big(\sfR^k( e)\,_{\star\!\!}\triangleleft \om\big) \nn \\[4pt]
& = -\big(\sfR_k( e) \dwedge_\star \sfR^k( e)\big)\,_{\star\!\!}\triangleleft \om - \big(\sfR_k( e)\,_{\star\!\!}\triangleleft \sfR_l (\om)\big) \dwedge_\star \sfR^l \sfR^k( e) \nn \\[4pt]
&=-(e\dwedge_\star e)\,_{\star\!\!}\triangleleft \om - \sfR_k (e\,_{\star\!\!}\triangleleft \om) \dwedge_\star \sfR^k( e) \nn \\[4pt] 
&= -(e\dwedge_\star e)\,_{\star\!\!}\triangleleft\om +  e\dwedge_\star (e\,_{\star\!\!}\triangleleft \om) \ . 
\end{align}
Substituting both \eqref{eq:Fom4did1} and \eqref{eq:Fom4did2} in \eqref{eq:Fom4dprelim}, and recalling that
\begin{align*}
e\,_{\star\!\!}\triangleleft \om := - \sfR_k( \om) \triangleright_\star \sfR^k( e) = -\sfR_k( \om )\wedge_\star \sfR^k(e) \ ,
\end{align*} 
after cancelling terms we find
\begin{align*}
F^\star_{\om} &= \frac{1}{6}\,\Big( 3 \, e\dwedge_\star \dd e - 3 \, \dd e\dwedge_\star e + e\dwedge_\star( \om \wedge_\star e) - \om \triangleright_\star (e\dwedge_\star e ) - (e\dwedge_\star e)\,_{\star\!\!}\triangleleft \om \\
& \quad \, \hspace{1cm} - \big(\sfR_k( \om) \wedge_\star \sfR^k( e)\big) \dwedge_\star e \Big) \nn \\[4pt]
&= \frac{1}{6}\,\Big( e \dwedge_\star \dd e + e\dwedge_\star (\om \wedge_\star e) - \dd e \dwedge_\star e - \big(\sfR_k (\om) \wedge_\star \sfR^k( e) \big) \dwedge_\star e  - \dd(e\dwedge_\star e) - \om \triangleright_\star (e\dwedge_\star e) \\
& \quad \, \hspace{1cm} - \dd(e\dwedge_\star e) - (e\dwedge_\star e)\,_{\star\!\!}\triangleleft \om \Big ) \ .
\end{align*}
After identifying left and right braided covariant derivatives, we finally arrive at \eqref{eq:Fom4d}.

\subsubsection*{Braided Einstein equation}

Next consider $F^\star_{e}$, which is our braided noncommutative analog of the four-dimensional Einstein equation with cosmological constant (that is, before solving for the Levi--Civita connection using the classical torsion equation). We will show that
\begin{align}\label{eq:Fe4d}
\begin{tabular}{|l|}\hline\\
$\displaystyle
F^\star_{e} = \tfrac{1}{6}\,\big( 2\, e\dwedge_\star R^\star + 2\, R^\star \dwedge_\star e + 6\,\Lambda\,e\dwedge_\star e\dwedge_\star e  $ \\[2mm]
$ \quad \, \hspace{2cm} + \, e \dwedge_\star \dd \om + \dd \om \dwedge_\star e + \sfR_k( e) \dwedge_\star [\sfR^k( \om),\om]_{\aso(4)}^\star \big)
$\\\\\hline\end{tabular}
\end{align} 
The first three terms are manifestly braided covariant, but the rest are not. However, a few lines of calculations shows that the last three terms form a braided covariant combination, as expected from the general covariance property which is guaranteed by the braided homotopy calculations of Section~\ref{sec:braidedEOM}. In the classical case, this expression reduces exactly (with the correct numerical factors) to the expected curvature form $F_e=e\dwedge R + \Lambda\,e\dwedge e\dwedge e$. 

Starting from \eqref{eq:Feell4d}, we use braided antisymmetry on the second and last terms to write
\begin{align*}
F^\star_{e} &= \tfrac{1}{2}\,\big(e \dwedge_\star \dd \om + \tfrac{1}{3}\, e \dwedge_\star [\om,\om]_{\aso(4)}^\star+ \dd \om \dwedge_\star e + \tfrac{1}{3}\, [\om,\om]_{\aso(4)}^\star \dwedge_\star e + 2\,\Lambda\,e\dwedge_\star e\dwedge_\star e \\
& \quad \, \hspace{1cm} + \tfrac{1}{3}\, \sfR_k(e) \dwedge_\star [\sfR^k( \om),\om]_{\aso(4)}^\star \big) \nn \\[4pt]
&= \tfrac{1}{2}\, \big(e\dwedge_\star R^\star + R^\star\dwedge_\star e + 2\,\Lambda\,e\dwedge_\star e\dwedge_\star e \\
& \quad \, \hspace{1cm} - \tfrac{1}{6}\, e\dwedge_\star [\om,\om]_{\aso(4)}^\star - \tfrac{1}{6}\,[\om,\om]_{\aso(4)}^\star \dwedge_\star e +\tfrac{1}{3}\, \sfR_k( e) \dwedge_\star [\sfR^k( \om),\om]_{\aso(4)}^\star \big) \ ,
\end{align*}
which is easily brought to the form \eqref{eq:Fe4d}. It is somewhat surprising and remarkable that the second line of \eqref{eq:Fe4d} is braided covariant, despite not being expressible in terms of manifestly covariant objects. It is instructive to explicitly demonstrate how this happens. For this, consider a left braided $\aso(4)$ gauge transformation of the first term there:
\begin{align}\label{eq:so4gt1}
\delta_{(0,\rho)}^{\star\lact} (e \dwedge_\star \dd \om)&=- (\rho\star e) \dwedge_\star \dd \om + \sfR_k( e) \dwedge_\star \dd \big( \sfR^k( \dd \rho) - [ \sfR^k( \rho), \om]_{\aso(4)}^\star \big) \nn \\[4pt]
&= - (\rho \star e) \dwedge_\star \dd \om - \sfR_k( e)\dwedge_\star [\sfR^k( \dd \rho), \om]_{\aso(4)}^\star -\sfR_k( e)\dwedge_\star [\sfR^k( \rho), \dd \om]_{\aso(4)}^\star \nn \\[4pt]
&= - \rho \triangleright_\star (e\dwedge_\star \dd \om) - \sfR_k( e)\dwedge_\star [\sfR^k( \dd \rho), \om]_{\aso(4)}^\star \ ,
\end{align}
where in the first and last equalities we used the left braided Leibniz rule. Similarly, one calculates 
\begin{align}\label{eq:so4gt2}
\delta_{(0,\rho)}^{\star\lact}(\dd \om \dwedge_\star e)= -\rho \triangleright_\star (\dd \om \dwedge_\star e) - [\dd \rho,\om]_{\aso(4)}^\star \dwedge_\star e \ .
\end{align}
The non-covariant terms involving $\dd\rho$ in \eqref{eq:so4gt1} and \eqref{eq:so4gt2} cancel exactly with the corresponding non-covariant terms in the gauge transformation of the last term in \eqref{eq:Fe4d}: after a couple of lines of similar manipulations we find
\begin{align}\label{eq:so4gt3}
\delta_{(0,\rho)}^{\star\lact}\big(\sfR_k( e) \dwedge_\star [\sfR^k( \om), \om]_{\aso(4)}^\star\big) &= - \rho \triangleright_\star \big( \sfR_k( e) \dwedge_\star [\sfR^k( \om), \om]_{\aso(4)}^\star \big) +  \sfR_k( e) \dwedge_\star [\sfR^k( \dd \rho),\om]_{\aso(4)}^\star \nn \\
& \quad \, + \sfR_k \sfR_l( e) \dwedge_\star [\sfR^k \sfR_m( \om), \sfR^l\sfR^m(\dd \rho)]_{\aso(4)}^\star \ .
\end{align}
The second term in \eqref{eq:so4gt3} cancels the non-covariant term in \eqref{eq:so4gt1}. Using the Yang--Baxter equation \eqref{eq:YangBaxter}, the last term in \eqref{eq:so4gt3} can be written as
\begin{align*}
\sfR_l \sfR_k( e) \dwedge_\star  [\sfR_m \sfR^k( \om), \sfR^m\sfR^l( \dd \rho)]_{\aso(4)}^\star
&= \sfR_l \sfR_k( e) \dwedge_\star [\sfR^l(\dd \rho ), \sfR^k( \om)]_{\aso(4)}^\star \\[4pt]
&= \sfR_k( e) \dwedge_\star \sfR^k( [\dd \rho,\om]_{\aso(4)}^\star) \\[4pt]
&= [\dd \rho,\om]_{\aso(4)}^\star \dwedge_\star e \ ,
\end{align*}
which cancels the remaining non-covariant term in \eqref{eq:so4gt2}; here we use braided symmetry of the braided Lie bracket, then the second $\RR$-matrix identity from \eqref{eq:Rmatrixidswn} followed by braided symmetry of the $\dwedge_\star$-product. Altogether, we have explicitly demonstrated \eqref{eq:covcondFeom4d}, as expected on general grounds from the proof using the braided $L_\infty$-algebra structure in Section~\ref{sec:braidedEOM}.

The field equations \eqref{eq:Fom4d} and \eqref{eq:Fe4d} of braided noncommutative gravity in four dimensions together are one of the main results of this paper. They are drastically different from the field equations previously suggested for the first order formulation of noncommutative gravity. Besides the universal (dimension-independent) differences that we discussed in the three-dimensional case of Section~\ref{sec:NCgravity3D}, where the braided field equations were relatively mild ``naive'' deformations of the classical equations, here we see that the explicit forms of the   braided field equations are very different and not simply ``naive'' deformations: the properties of the field equations of braided gravity vary drastically with dimension, in contrast to the classical case. In particular, the torsion equation $F_\om^\star=0$ now involves more braided covariant forms than just those constructed from the left and right torsion forms. The curvature equation $F_e^\star=0$ is even more striking, as in addition to the anticipated braided covariant forms, there is a further non-manifestly covariant combination whose individual forms are not braided covariant on their own. These latter forms could not have been guessed or even motivated solely by deformation of the classical equations. They are a necessary and natural consequence of twisting the kinematical and dynamical content simultaneously via the braided $L_\infty$-algebra approach to noncommutative field theory: in this approach braided gauge covariance is formulated for the doublet of fields $A=(e,\om)\in V_1$ (and the corresponding doublet of field equations), which allows for non-trivial mixing of the component fields $e$ and $\om$ in non-manifestly covariant ways, but in such a way that the braided homotopy relations ensure that all combinations are covariant under braided gauge transformations.

Generally, we expect this complexity to increase as the number and orders of non-zero brackets $\ell_n^\star$ for $n>2$ increases. In this sense braided noncommutative field theories generally contain many more covariant forms than naively expected from classical considerations.

\subsubsection*{Braided Noether identities}

By \eqref{braidedNoether}, the gauge redundancies of braided gravity in four dimensions are encoded by the braided Noether identities
\begin{align}\label{eq:braidedNoether4d}
\dsf_{(e,\om)}^\star F^\star_{(e,\om)} &= \ell_1^\star\big(F_{e}^\star,F_\om^\star\big) -\frac1{12}\,\ell_1^\star\big(\ell_3^\star((e,\om),(e,\om),(e,\om))\big) + \frac14\,\ell_2^\star\big(\sfR_k(e,\om)\, , \, \ell_2^\star(\sfR^k(e,\om),(e,\om)) \big) \nn \\ & \quad \, + \frac1{12} \, \Big(\ell_2^\star\big(\ell_3^\star((e,\om),(e,\om),(e,\om))\,,\,(e,\om)\big) - \ell_2^\star\big((e,\om)\,,\,\ell_3^\star((e,\om),(e,\om),(e,\om))\big)\Big) \nn \\
& \quad \, - \frac12\,\Big(\ell_2^\star\big((e,\om)\,,\,(F_e^\star,F_\om^\star)\big) - \ell_2^\star\big((F_e^\star,F_\om^\star)\,,\,(e,\om)\big)\Big)
\ = \ (0,0) \ ,
\end{align}
where we again used the slightly simpler form \eqref{eq:leftNoetherprelim}, along with $\ell_1^\star(e,\om)=(0,0)$. As previously, we rewrite these as differential identities among the field equations \eqref{eq:Fom4d} and \eqref{eq:Fe4d}.

The braided Noether identity for local $\aso(4)$ rotations, in $\Omega^4\big(M,\midwedge^2\FR^{4}\big)[[\hbar]]$, is given by the second entry of \eqref{eq:braidedNoether4d}. Inserting the corresponding brackets from \eqref{eq:ell14d}--\eqref{eq:ell34d}, using the usual $\RR$-matrix identities from \eqref{eq:YangBaxter} and \eqref{eq:Rmatrixidswn}, and collecting terms using \eqref{eq:Fom4did1} and \eqref{eq:Fom4did2}, we find
\begin{align}\label{eq:braidedNoetherso4}
& \frac12\,\Big(\dd_{\star \lact}^\om F_\om^\star+\dd_{\star \ract}^\om F_\om^\star+\frac32\,\sfR_k(F_e^\star)\wedge_{\star}\sfR^k(e)+\frac32\,F_e^\star\wedge_{\star}e\Big)
\nonumber \\
& \hspace{1cm} + \frac14\,\Big( \frac13\,\big(\dd\Phi^\star_{(e,\om)} - [\Phi^{\star}_{(e,\om)},\om]_{\aso(4)}^\star - [\om,\Phi_{(e,\om)}^\star]_{\aso(4)}^\star \big) - \frac{1}{2}\,\big(\widetilde{\Phi}{}^\star_{(e,\om)}\wedge_{\star} e + \sfR_l(\widetilde{\Phi}{}^\star_{(e,\om)})\wedge_{\star} \sfR^l(e) \big) \nonumber\\
& \hspace{2cm} - \frac32\,\big(e\dwedge_{\star} \sfR_k(\dd\om)+\dd\om\dwedge_{\star}\sfR_k( e)\big)\wedge_{\star}\sfR^k (e)+[\dd (e\dwedge_{\star} \sfR_k (e)),\sfR^k(\om)]_{\aso(4)}^\star\nonumber\\
& \hspace{2cm} -3\, \Lambda\,( e\dwedge_\star e\dwedge_\star e )\wedge_{\star} e - 3\, \Lambda\, \sfR_k( e\dwedge_\star e\dwedge_\star e )\wedge_{\star} \sfR^k(e)\Big) \ = \ 0 \ ,
\end{align}
where we defined
\begin{align*}
\Phi^\star_{(e,\om)} &= e\dwedge_\star(\om\wedge_{\star}e)-(\sfR_k(\om) \wedge_{\star} \sfR^k (e))\dwedge_{\star}e - \om\triangleright_\star(e\dwedge_\star e)-(e\dwedge_{\star}e)\,_{\star\!\!}\triangleleft \om \ , \\[4pt]
\widetilde{\Phi}{}^\star_{(e,\om)} &= e\dwedge_{\star}[\om,\om]_{\aso(4)}^\star+\sfR_k(e)\dwedge_{\star}[\sfR^k(\om),\om]_{\aso(4)}^\star+[\om,\om]_{\aso(4)}^\star\dwedge_{\star}e \ .
\end{align*}
In the classical case, only the first line of \eqref{eq:braidedNoetherso4} survives and reduces to the expected Noether identity $\dd^\omega F_\omega = -\frac{3}{2}\,F_e\wedge e$, which follows from (but is no longer equivalent to) the first Bianchi identity $\dd^\om T= R\wedge e$~\cite{ECPLinfty}.

The Noether identity corresponding to the braided diffeomorphism invariance of the four-dimensional theory is valued in $\Omega^1(M)\otimes\Omega^4(M)[[\hbar]]$ and is given by the first entry of \eqref{eq:braidedNoether4d}. The contributions from the first two terms vanish by definition of the differential $\ell_1^\star$, while the last line contributes terms of the same form as in the three-dimensional case, up to signs. The remaining terms can be expanded by commuting the exterior derivative $\dd$ with the legs of the twist, along with the identities \eqref{eq:Fom4did1} and \eqref{eq:Fom4did2}. Altogether, the explicit form of the braided diffeomorphism Noether identity for four-dimensional gravity may be expressed similarly to \eqref{eq:diffNoether3d}, now with the $4$-form $N_\mu^\star(F_e^\star,F_\om^\star,e,\om)\in\Omega^4\big(M,\midwedge^4\FR^{4})[[\hbar]]$ given by
\begin{align*}
N_\mu^\star &= \frac12\,\big(\iota_\mu \bar{\sff}^k (\dd e) \dwedge \bar{\sff}_k (F^\star_{e}) - \iota_\mu\bar{\sff}^k( \dd \om) \dwedge \bar{\sff}_k(F^\star_{\om}) - \iota_\mu \bar{\sff}^k( e) \dwedge \bar{\sff}_k (\dd F^\star_{e}) + \iota_\mu \bar{\sff}^k (\om) \dwedge \bar{\sff}_k (\dd F^\star_{\om}) \\
& \hspace{1cm} + \iota_\mu  \bar{\sff}_k (\dd e) \dwedge \bar{\sff}^k (F^\star_{e}) - \iota_\mu \bar{\sff}_k (\dd \om)\dwedge \bar{\sff}^k(F^\star_{\om}) - \iota_\mu \bar{\sff}_k (e) \dwedge \bar{\sff}^k (\dd F^\star_{e}) + \iota_\mu \bar{\sff}_k (\om) \dwedge \bar{\sff}^k (\dd F^\star_{\om} ) \big) \\
& \quad \, +\frac14\,\Big(\iota_\mu \bar{\sff}^k \sfR_l (\dd e )\dwedge \bar{\sff}_k\big( \sfR^l(e)\dwedge_\star \dd \om + \sfR^l (\dd \om) \dwedge_\star e \big) + \iota_\mu \bar{\sff}^k \sfR_l (\dd \om) \dwedge \dd\,\bar{\sff}_k \big(\sfR^l( e) \dwedge_\star  e\big) \\
& \hspace{1.5cm} -\iota_\mu \bar{\sff}^k \sfR_l (e) \dwedge  \bar{\sff}_k \big(\sfR^l (\dd e)\dwedge_\star \dd \om + \sfR^l (\dd \om) \dwedge_\star \dd e\big) \Big) \\
& \quad \, + \frac1{72} \, \big( \iota_\mu \bar{\sff}^k (\dd e) \dwedge \bar{\sff}_k (\Psi_{(e,\om)}^\star ) - \iota_\mu \bar{\sff}^k (e) \dwedge \dd\,\bar{\sff}_k (\Psi_{(e,\om)}^\star) \\
& \hspace{2cm} + \iota_\mu \bar{\sff}_k( \dd e )\dwedge \bar{\sff}^k (\Psi_{(e,\om)}^\star) -\iota_\mu\bar{\sff}_k(e)\dwedge\dd\,\bar{\sff}^k(\Psi_{(e,\om)}^\star) \\
& \hspace{1.7cm} - \iota_\mu \bar{\sff}^k (\dd \om) \dwedge \bar{\sff}_k(\widetilde{\Psi}{}^\star_{(e,\om)}) + \iota_\mu \bar{\sff}^k (\om) \dwedge \dd\,\bar{\sff}_k (\widetilde{\Psi}{}^\star_{(e,\om)}) \\
& \hspace{2cm} -\iota_\mu\bar\sff_k(\dd\om)\dwedge\bar\sff^k(\widetilde{\Psi}{}^\star_{(e,\om)}) + \iota_\mu\bar\sff_k(\om)\dwedge\dd\,\bar\sff^k(\widetilde{\Psi}{}^\star_{(e,\om)}) \big) \ ,
\end{align*}
where we defined
\begin{align*}
\Psi_{(e,\om)}^\star &= e\dwedge_\star [\om,\om]_{\aso(4)}^\star + [\om,\om]_{\aso(4)}^\star\dwedge_\star e + \sfR_l( e) \dwedge_\star [ \sfR^l( \om) ,\om]_{\aso(4)}^\star + 6\, \Lambda\, e\dwedge_\star e\dwedge_\star e \ , \\[4pt]
\widetilde{\Psi}{}^\star_{(e,\om)} &= e \dwedge_\star (\om \wedge_\star e) - \big(\sfR_l (\om) \wedge_\star \sfR^l( e) \big) \dwedge_\star e - \om \triangleright_\star (e \dwedge_\star e) + (e\dwedge_{\star}e)\,_{\star\!\!}\triangleleft \om \ .
\end{align*}
We see that the braided Noether identities become extremely complicated in this case, even with the simplified twists that we discussed in the three-dimensional theory, as we see in Example~\ref{ex:MoyalNoether4d} below. Nevertheless, the identities exist in explicit form and provide a highly non-trivial differential identification among the braided field equations, which captures the statement of braided gauge invariance of the field theory.

\begin{example}\label{ex:MoyalNoether4d}
Similarly to Example~\ref{ex:MoyalNoether3d}, the expression for the braided diffeomorphism Noether identity can be reduced slightly in the example of the Moyal--Weyl twist $\CF_\theta$ on $M=\FR^4$. In this case we can again write the Noether identity as in \eqref{eq:MoyalNoether3d}, now with the $4$-form
\begin{align*}
N_\mu^{\star_\theta} &= \tfrac12\,\big(\iota_\mu^{\star_\theta} \dd e \dwedge_{\star_\theta} F^{\star_\theta}_{e} - \iota_\mu^{\star_\theta}\dd \om \dwedge_{\star_\theta} F^{\star_\theta}_{\om} - \iota_\mu^{\star_\theta} e\dwedge_{\star_\theta} \dd F^{\star_\theta}_{e} + \iota_\mu^{\star_\theta} \om \dwedge_{\star_\theta} \dd F^{\star_\theta}_{\om} \\
& \hspace{1cm} + F^{\star_\theta}_{e} \dwedge_{\star_\theta} \iota_\mu^{\star_\theta}  \dd e - F^{\star_\theta}_{\om}\dwedge_{\star_\theta} \iota_\mu^{\star_\theta} \dd \om - \dd F^{\star_\theta}_{e}\dwedge_{\star_\theta} \iota_\mu^{\star_\theta} e + \dd F^{\star_\theta}_{\om} \dwedge_{\star_\theta} \iota_\mu^{\star_\theta} \om \big) \\
& \quad \, +\tfrac14\,\big(e\dwedge_{\star_\theta} \iota_\mu^{\star_\theta} \dd e \dwedge_{\star_\theta} \dd \om + \dd \om\dwedge_{\star_\theta} \iota_\mu^{\star_\theta}\dd e \dwedge_{\star_\theta} e + \dd e\dwedge_{\star_\theta} \iota_\mu^{\star_\theta}\dd\om\dwedge_{\star_\theta} e \\
& \hspace{1.5cm} - e\dwedge_{\star_\theta} \iota_\mu^{\star_\theta}\dd\om\dwedge_{\star_\theta} \dd e - \dd e\dwedge_{\star_\theta} \iota_\mu^{\star_\theta} e \dwedge_{\star_\theta} \dd \om - \dd \om \dwedge_{\star_\theta} \iota_\mu^{\star_\theta} e\dwedge_{\star_\theta} \dd e \big) \\
& \quad \, + \tfrac1{72} \, \big( \iota_\mu^{\star_\theta} \dd e\dwedge_{\star_\theta} \Psi_{(e,\om)}^{\star_\theta} - \iota_\mu^{\star_\theta} e\dwedge_{\star_\theta} \dd\Psi_{(e,\om)}^{\star_\theta} + \Psi_{(e,\om)}^{\star_\theta}\dwedge_{\star_\theta}\iota_\mu^{\star_\theta} \dd e - \dd\Psi_{(e,\om)}^{\star_\theta}\dwedge_{\star_\theta}\iota_\mu^{\star_\theta} e \\
& \hspace{1.7cm} - \iota_\mu^{\star_\theta} \dd \om \dwedge_{\star_\theta} \widetilde{\Psi}{}^{\star_\theta}_{(e,\om)} + \iota_\mu^{\star_\theta} \om \dwedge_{\star_\theta} \dd\widetilde{\Psi}{}^{\star_\theta}_{(e,\om)} - \widetilde{\Psi}{}^{\star_\theta}_{(e,\om)} \dwedge_{\star_\theta} \iota_\mu^{\star_\theta} \dd \om + \dd\widetilde{\Psi}{}^{\star_\theta}_{(e,\om)} \dwedge_{\star_\theta}\iota_\mu^{\star_\theta} \om \big) \ .
\end{align*}
\end{example}

\subsubsection*{Action functional}

Twisting the four-dimensional classical ECP pairing from~\cite[Section~8.2]{ECPLinfty} using an arbitrary Drinfel'd twist $\CF$, we obtain the (braided) cyclic structure 
	\begin{align}\label{eq:ecpbraidedpairing2}
		\begin{split}
			\langle (e,\om) \,,\, (E,{\mit\Omega}) \rangle_{\star}&:=
			\int_{M}\, \Tr \big(e\dwedge_\star E - \om \dwedge_\star {\mit\Omega} \big) \ ,\\[4pt]
			\langle(\xi,\rho)\,,\,({\CX},{\CP})\rangle_{\star} &:= \int_M\,
			\iota_{\xi}^{
				\star}{\CX} + \int_M\, \Tr\big(\rho\dwedge_\star{\CP}\big) \ ,
		\end{split}
	\end{align}
	on $V_{1}\otimes V_{2}$ and $V_{0}\otimes V_{3}$, respectively. We shall write out and simplify the four-dimensional braided noncommutative ECP functional in two steps. First we work with a general twist (or alternatively ignore strict cyclicity for a compatible twist). Then we shall employ a compatible twist, and use $\wedge_\star$-cyclicity to simplify the corresponding action functional. The ECP action functional is defined via \eqref{eq:braidedaction}:
	\begin{align}\label{eq:braidedECPLagrprelim}
		S_{\star}(e,\om) :&\!= \tfrac12\, \big\langle (e,\omega)\,,\,\ell^\star_1(e,\omega)\big\rangle_{\star} -
		\tfrac1{6}\, \big\langle
		(e,\omega)\,,\,\ell^\star_2\big((e,\omega)\,,\,(e,\omega)\big) \big\rangle_{\star}
		\nn \\
& \quad \, - \tfrac1{24}\, \big\langle (e,\omega)\,,\,
		\ell^\star_3\big((e,\omega)\,,\,(e,\omega) \,,\,(e,\omega)\big)
		\big\rangle_{\star} \nn  \\[4pt]
		&= \int_M \, \frac{1}{6}\,  \Tr \big(e\dwedge_{\star}e\dwedge_{\star} \dd \om + e \dwedge_{\star} \dd \om \dwedge_\star e - \om \dwedge_{\star} e \dwedge_{\star} \dd e+ \om \dwedge_{\star} \dd e \dwedge_{\star} e \big )\nn  \\
		&\hspace{1.5cm}+ \frac{1}{24} \, \Tr\Big( e \dwedge_{\star} e \dwedge_{\star} [\om,\om]_{\aso(4)}^\star + e \dwedge_{\star} \sfR_k( e )\dwedge_{\star} [\sfR^k( \om) ,\om]_{\aso(4)}^\star \nn  \\
		&\hspace{3.1cm} + e \dwedge_{\star} [\om,\om]_{\aso(4)}^\star\dwedge_\star e -\om \dwedge_{\star} e \dwedge_{\star} (\om \wedge_\star e)  \\
		&\hspace{3.1cm} - \om \dwedge_{\star}e \dwedge_{\star} \big(\sfR_k( \om) \wedge_\star \sfR^k( e)\big) - \om \dwedge_{\star} \sfR_k( e )\dwedge_{\star} \big(\sfR^k( \om) \wedge_\star e\big)  \nn \\
		&\hspace{3.1cm}- \om \dwedge_{\star} \sfR_k( e )\dwedge_{\star} \big(\sfR_l( \om) \wedge_\star \sfR^l \sfR^k( e)\big) \nn  + \om \dwedge_{\star}( \om \wedge_\star e) \dwedge_\star e  \\ &\hspace{3.1cm} + \om \dwedge_{\star} \big(\sfR_k(\omega) \wedge_\star \sfR^k(e)\big)\wedge_\star e + 6\,\Lambda\,e\dwedge_\star e\dwedge_\star e\dwedge_\star e\Big)\ , \nn 
	\end{align}	
	where we used $\ell^\star_1(e,\om)=(0,0)$.  
	
We note that the first, second, fifth, sixth, and seventh terms in \eqref{eq:braidedECPLagrprelim} combine modulo exact forms as 
	\begin{align*}
			\tfrac{1}{24}\,&\big( 2\, e\dwedge_\star e\dwedge_\star R^\star + 2 \, e \dwedge_\star R^\star \dwedge_\star e + e \dwedge_{\star} \sfR_k( e )\dwedge_{\star} [\sfR^k( \om) ,\om]_{\aso(4)}^\star \\& \,\,   +2 \, e\dwedge_\star e \dwedge_\star \dd \om + 2\,   e\dwedge_\star \dd \om \dwedge_\star e \big)   \\[4pt]
			&\hspace{2cm}= \tfrac{1}{24}\,\big( 2\, e\dwedge_\star e\dwedge_\star R^\star + 2 \, e \dwedge_\star R^\star \dwedge_\star e - \big(\sfR_k(\om)\wedge_\star \sfR^k (e)\big) \dwedge_\star e\dwedge_\star \om  \\
			&\hspace{2cm}\qquad \qquad + e \dwedge_\star (\om\wedge_\star e)\dwedge_\star \om -2\, \dd (e \dwedge_\star e) \dwedge_\star \om + 2\, e \dwedge_\star \dd \om \dwedge_\star e \big) \\[4pt]
			&\hspace{2cm}= \tfrac{1}{24}\, \big( 2\, e\dwedge_\star e\dwedge_\star R^\star + 2 \, e \dwedge_\star R^\star \dwedge_\star e - T_\ract^\star \dwedge_\star e \dwedge_\star \om + e\dwedge_\star T_\lact^\star \dwedge_\star \om \\ 
			&\hspace{2cm} \qquad \qquad + e\dwedge_\star e \dwedge_\star \dd \om + 2\, e\dwedge_\star \dd \om \dwedge_\star e\big)\ ,
	\end{align*}
where in the first equality we used invariance of top-vectors under braided local $\mathfrak{so}(4)$  transformations on the third term and integration by parts on the fourth term, while in the second equality we simply identified the left and right torsion 2-forms. The rest of the terms in \eqref{eq:braidedECPLagrprelim} may be analogously manipulated into the expression
	\begin{align*}
			\tfrac{1}{24}\,\big(  \om \dwedge_\star T_\ract^\star \dwedge_\star e -\om \dwedge_\star e \dwedge_\star T_\lact^\star + \dd \om \dwedge_\star  e\dwedge_\star e +\om\dwedge_\star \dd^{\om}_{\star\ract}(e\dwedge_\star e) + \om \dwedge_\star \dd^\om_{\star\lact}(e\dwedge_\star e) \big) \ . 
		\end{align*}
Summing these two expressions, we obtain the $4$-dimensional ECP action functional in the simplified form
\begin{align}\label{eq:braidedECPLagr}
		S_\star(e,\om) = \frac{1}{24} \, \int_M\, \Tr \Big( &2\, e\dwedge_\star e\dwedge_\star R^\star + 2 \, e \dwedge_\star R^\star \dwedge_\star e - T_\ract^\star \dwedge_\star e \dwedge_\star \om + e\dwedge_\star T_\lact^\star \dwedge_\star \om \\ & \, + \om \dwedge_\star T_\ract^\star \dwedge_\star e -\om \dwedge_\star e \dwedge_\star T_\lact^\star +\om\dwedge_\star \dd^{\om}_{\star\ract}(e\dwedge_\star e) + \om \dwedge_\star \dd^\om_{\star\lact}(e\dwedge_\star e)\nn \\
		&\, +e\dwedge_\star e\dwedge_\star \dd \om + 2\, e\dwedge_\star \dd \om \dwedge_\star e + \dd \om \dwedge_\star  e\dwedge_\star e + 6\,\Lambda\,e\dwedge_\star e\dwedge_\star e\dwedge_\star e\Big) \ . \nn 
	\end{align}
	
The first two terms of \eqref{eq:braidedECPLagr} are manifestly braided covariant, while the rest are not individually covariant since the gauge field $\om$ appears explicitly. However, one may check directly that they do form covariant combinations. Let us explicitly demonstrate that the eighth and the last term form a covariant combination. The potential non-covariance arises from the terms proportional to $\dd \rho$ in a braided $\mathfrak{so}(4)$ gauge transformation, since the covariant components proportional to $\rho$ vanish by invariance of a top-vector. Hence working modulo exact forms we find
	\begin{align*}
			\delta_{\rho}^{\star,\lact}\big(\om\dwedge_\star  \dd^{\om}_{\star\lact}(e\dwedge_\star e)+ \dd \om \dwedge_\star  e\dwedge_\star e \big) & = \dd \rho \dwedge_\star \dd^{\om}_{\star\lact}(e\dwedge_\star e) - [\dd \rho, \om]_{\mathfrak{so}(4)}^{\star} \dwedge_\star e \dwedge_\star e \\[4pt]
			&= \dd \rho \dwedge_\star\big( \om \triangleright_\star (e\dwedge_\star e) \big) - \dd \rho \dwedge_\star \big(\om \triangleright_\star (e\dwedge_\star e) \big) \\[4pt]
			&=0 \ , 
	\end{align*}
where in the second equality we dropped an exact top-form from the first term, and used invariance of top-vectors under the braided $\mathfrak{so}(4)$ action on the second term. The rest of the terms in \eqref{eq:braidedECPLagr} form covariant combinations via analogous calculations. Hence the invariance of the action functional under braided $\mathfrak{so}(4)$ transformations is fulfilled, as it should be by the general discussion of Section~\ref{sec:braidedgaugeinv}.
	
By restricting to compatible twists as in Proposition~\ref{prop:compatiblestrict}, the pairing becomes a strictly cyclic structure on $\big(V[[\hbar]], \ell_{1}^\star, \ell_2^\star, \ell_{3}^\star \big)$. In this case, the four-dimensional braided ECP action functional simplifies further to 
	\begin{align}\label{eq:4daction}
		S_\star(e,\om) &= \int_M\,\Tr\Big(\frac12\,e\dwedge_\star e\dwedge_\star R^\star + \frac\Lambda4 \, e\dwedge_\star e\dwedge_\star e\dwedge_\star e\Big)  \\ & \quad \, -\frac1{24} \, \int_M\,\Tr\Big(\om\dwedge_\star\big(2\,e\dwedge_\star T_\lact^\star - 2\,T_\ract^\star\dwedge_\star e + \dd_{\star\lact}^\om(e\dwedge_\star e) + \dd_{\star\ract}^\om(e\dwedge_\star e)\big)\Big) \ \, \nn 
	\end{align} 
	which follows immediately from \eqref{eq:braidedECPLagr} by $\wedge_\star$-cyclicity under the integral and the identity
	\begin{align*}
			\Tr\big(\dd\om\dwedge_\star e\dwedge_\star e\big) = \Tr\big(2\,R^\star\dwedge_\star e\dwedge_\star e - \tfrac12\,\om\dwedge_\star\dd_{\star\lact}^\om(e\dwedge_\star e) - \tfrac12\,\om\dwedge_\star\dd_{\star\ract}^\om(e\dwedge_\star e) \big) \ ,
	\end{align*}
that may be checked by a simple expansion of the terms on the right-hand side.
The first line of \eqref{eq:4daction} is exactly the ``naive'' deformation of the classical Einstein--Cartan--Palatini action functional in four dimensions, with the correct numerical prefactors; indeed, the second line vanishes in the classical case $\RR=1\otimes 1$. Using similar identities to those employed before, one can explicitly check that the variation of this action functional with respect to the spin connection $\om$ yields the field equations $\CF_\om^\star=0$ from \eqref{eq:Fom4d}, while varying with respect to coframe field $e$ yields the field equations $\CF_e^\star=0$ from \eqref{eq:Fe4d}, as expected on general grounds from Section~\ref{sec:braidedaction}.

\end{document}